\def\isarxiv{1}
\def\debug{1}

\ifdefined\isarxiv
\documentclass[11pt]{article}
\else
\documentclass[a4paper,USenglish, autoref, thm-restate]{lipics-v2021}
\nolinenumbers
\fi

\usepackage{amsmath}
\usepackage{amsthm}
\usepackage{cancel}
\usepackage{amssymb}
\usepackage{algorithm}
\usepackage{subfig}
\usepackage{color}
\usepackage{algpseudocode}

\ifdefined\isarxiv
\DeclareMathAlphabet{\mathpzc}{OT1}{pzc}{m}{it}

\usepackage{tikz}
\usepackage{hyperref}
\hypersetup{colorlinks=true,citecolor=red,linkcolor=blue}
\usetikzlibrary{arrows}
\usepackage[margin=1in]{geometry}
\graphicspath{{./figs/}}

\usepackage{mathtools}

\else
\graphicspath{{./figs/}}
\fi

\ifdefined\isarxiv

\newtheorem{theorem}{Theorem}[section]
\newtheorem{lemma}[theorem]{Lemma}
\newtheorem{definition}[theorem]{Definition}

\newtheorem{assumption}[theorem]{Assumption}

\newtheorem{fact}[theorem]{Fact}
\newtheorem{remark}[theorem]{Remark}

\else
\newtheorem{fact}[theorem]{Fact}
\newtheorem{assumption}[theorem]{Assumption}
\fi

\newcommand{\wh}{\widehat}
\newcommand{\wt}{\widetilde}
\newcommand{\ov}{\overline}
\newcommand{\eps}{\epsilon}

\newcommand{\R}{\mathbb{R}}

\renewcommand{\varepsilon}{\epsilon}

\renewcommand{\hat}{\wh}
\renewcommand{\eps}{\epsilon}
\renewcommand{\d}{\mathrm{d}}

\newcommand{\tmp}{\mathrm{tmp}}

\newcommand{\ipm}{\mathrm{ipm}}

\newcommand{\dmax}{\mathrm{d}_{\mathrm{max}}}
\newcommand{\OPT}{\mathrm{OPT}}
\newcommand{\final}{\mathrm{final}}
\newcommand{\start}{\mathrm{start}}
\newcommand{\init}{\mathrm{init}}

\DeclareMathOperator*{\E}{{\mathbb{E}}}

\DeclareMathOperator*{\Z}{\mathbb{Z}}

\DeclareMathOperator{\supp}{supp}

\DeclareMathOperator{\poly}{poly}

\DeclareMathOperator{\nnz}{nnz}

\DeclareMathOperator{\tr}{tr}

\DeclareMathOperator{\diag}{diag}

\DeclareMathOperator{\new}{new}

\definecolor{b2}{RGB}{51,153,255}
\definecolor{mygreen}{RGB}{80,180,0}
\definecolor{yl}{RGB}{255,80,0}
\definecolor{zhj}{RGB}{255,50,200}
\definecolor{mycy2}{RGB}{255,51,255}

\makeatletter
\newcommand*{\RN}[1]{\expandafter\@slowromancap\romannumeral #1@}
\makeatother
\title{Space-Efficient Interior Point Method, with applications to Linear Programming and Maximum Weight Bipartite Matching}
\ifdefined\isarxiv

\else
\titlerunning{Space-Efficient IPM for LP and Maximum Weight Bipartite Matching}
\fi

\usepackage{lineno}

\usepackage{etoolbox}

\makeatletter

\newcommand{\define}[4][ignore]{%
  \ifstrequal{#1}{ignore}{}{
  \@namedef{thmtitle@#2}{#1}}%
  \@namedef{thm@#2}{#4}%
  \@namedef{thmtypen@#2}{lemma}%
  \newtheorem{thmtype@#2}[theorem]{#3}%
  \newtheorem*{thmtypealt@#2}{#3~\ref{#2}}%
}

\newcommand{\state}[1]{%
  \@namedef{curthm}{#1}
  \@ifundefined{thmtitle@#1}{
  \begin{thmtype@#1}
    }{
  \begin{thmtype@#1}[\@nameuse{thmtitle@#1}]
  }
    \label{#1}
    \@nameuse{thm@#1}
  \end{thmtype@#1}
  \@ifundefined{thmdone@#1}{
  \@namedef{thmdone@#1}{stated}%
  }{}
}

\newcommand{\restate}[1]{%
  \@namedef{curthm}{#1}
  \@ifundefined{thmtitle@#1}{
    \begin{thmtypealt@#1}
    }{
  \begin{thmtypealt@#1}[\@nameuse{thmtitle@#1}]
  }
    \@nameuse{thm@#1}
  \end{thmtypealt@#1}
  \@ifundefined{thmdone@#1}{
  \@namedef{thmdone@#1}{stated}%
  }{}
}

\newcommand{\thmlabel}[1]{
  \@ifundefined{thmdone@\@nameuse{curthm}}{\label{#1}
    }{\tag*{\eqref{#1}}}
}
\makeatother

\date{}

\ifdefined\isarxiv
\author{ 
S. Cliff Liu\thanks{Carnegie Mellon University, \texttt{cliffliu@andrew.cmu.edu}.}
\and
Zhao Song\thanks{Adobe Research, \texttt{zsong@adobe.com}.}
\and
Hengjie Zhang\thanks{Columbia University, \texttt{hengjie.z@columbia.edu}.}
\and
Lichen Zhang\thanks{Massachusetts Institute of Technology, \texttt{lichenz@mit.edu}. Supported by NSF grant CCF-1955217 and NSF grant DMS-2022448.}
\and
Tianyi Zhou\thanks{University of California, San Diego, \texttt{t8zhou@ucsd.edu}.}
}
\else 
\author{S. Cliff Liu}{Carnegie Mellon University, PA, USA}{cliffliu@andrew.cmu.edu}{}{}
\author{Zhao Song}{Adobe Research, CA, USA}{zsong@adobe.com}{}{}
\author{Hengjie Zhang}{Columbia University, NY, USA}{hengjie.z@columbia.edu}{}{}
\author{Lichen Zhang}{Massachusetts Institute of Technology, MA, USA}{lichenz@mit.edu}{}{Supported by NSF grant CCF-1955217 and NSF grant DMS-2022448.}
\author{Tianyi Zhou}{University of California San Diego, CA, USA}{t8zhou@ucsd.edu}{}{}

\authorrunning{S. Liu, Z. Song, H. Zhang, L. Zhang and T. Zhou}
\Copyright{S. Cliff Liu, Zhao Song, Hengjie Zhang, Lichen Zhang and Tianyi Zhou}
\ccsdesc[500]{Theory of computation~Streaming, sublinear and near linear time algorithms}
\ccsdesc[500]{Theory of computation~Linear programming}

\keywords{Convex optimization, interior point method, streaming algorithm}
\ArticleNo{15}
\fi
\begin{document}

\ifdefined\isarxiv
\begin{titlepage}
  \maketitle
 \begin{abstract}
 We study the problem of solving linear program in the streaming model. Given a constraint matrix $A\in \R^{m\times n}$ and vectors $b\in \R^m, c\in \R^n$, we develop a space-efficient interior point method that optimizes solely on the dual program. To this end, we obtain efficient algorithms for various different problems:
\begin{itemize}
    \item For general linear programs whose feasible region lies in a ball of radius $R$ and contains a ball of radius $r$ with known center, we can solve them in $\wt O(\sqrt n\log(mR^2\|c\|_2/(r\epsilon)))$ passes and $\wt O(n^2)$ space for an $\epsilon$-approximate solution. To the best of our knowledge, this is the most efficient LP solver in streaming with no polynomial dependence on $m$ for both space and passes.
    \item For bipartite graphs, we can solve the minimum vertex cover and maximum weight matching problem (with polynomially bounded integer weights) in $\wt O(\sqrt{m})$ passes and $\wt O(n)$ space.
\end{itemize}

In addition to our space-efficient IPM, we also give algorithms for solving SDD systems and isolation lemma in $\wt O(n)$ spaces, which are the cornerstones for our graph results.

 \end{abstract}
  \thispagestyle{empty}
\end{titlepage}
\newpage 
\else 
 \maketitle
 \begin{abstract}
 We study the problem of solving linear program in the streaming model. Given a constraint matrix $A\in \R^{m\times n}$ and vectors $b\in \R^m, c\in \R^n$, we develop a space-efficient interior point method that optimizes solely on the dual program. To this end, we obtain efficient algorithms for various different problems:
\begin{itemize}
    \item For general linear programs whose feasible region lies in a ball of radius $R$ and contains a ball of radius $r$ with known center, we can solve them in $\wt O(\sqrt n\log(mR^2\|c\|_2/(r\epsilon)))$ passes and $\wt O(n^2)$ space for an $\epsilon$-approximate solution. To the best of our knowledge, this is the most efficient LP solver in streaming with no polynomial dependence on $m$ for both space and passes.
    \item For bipartite graphs, we can solve the minimum vertex cover and maximum weight matching problem (with polynomially bounded integer weights) in $\wt O(\sqrt{m})$ passes and $\wt O(n)$ space.
\end{itemize}

In addition to our space-efficient IPM, we also give algorithms for solving SDD systems and isolation lemma in $\wt O(n)$ spaces, which are the cornerstones for our graph results.

 \end{abstract}
\fi
\ifdefined\isarxiv
{
\hypersetup{linkcolor=black}
\tableofcontents
}
\newpage
\else 

\fi



\section{Introduction}

Given a constraint matrix $A\in \R^{m\times n}$, vectors $b\in \R^m$ and $c\in \R^n$, the linear program problem asks us to solve the primal program $(P)$ or its dual $(D)$:
\begin{align}\label{eq:lp_primal_dual}
    (P)=\max_{A^\top y\leq c, y\geq 0}~b^\top y & ~ \text{and}~(D)=\min_{Ax\geq b}~c^\top x
\end{align}
is one of the most fundamental problems in computer science and operational research. Many efforts have been dedicated to develop time-efficient linear program solvers in the past half a century, such as the simplex method~\cite{d51}, ellipsoid method~\cite{k80} and interior point method~\cite{k84}. In the last few years, speeding up linear program solve via interior point method (IPM) has been heavily studied~\cite{cls19,lsz19,blss20,jswz21,sy21,dly21,y21}. The state-of-the-art IPM has the runtime of $O(m^{\omega}+m^{2.5-\alpha/2}+m^{2+1/18})$ when $m\approx n$ and $O(mn+n^3)$ when $m\gg n$, where $\omega$ is the exponent and $\alpha$ is the dual exponent of matrix multiplication. To achieve these impressive improvements, most of these algorithms utilize randomized and dynamic data structures to maintain the primal and dual solutions simultaneously. While these algorithms are time-efficient, it is highly unlikely that they can be implemented in a \emph{space-efficient} manner: maintaining the primal-dual formulation requires $\Omega(m+n^2)$ space, which is particularly unsatisfactory when $m\gg n$.

In this paper, we study the problem of solving a linear program in the streaming model: At each pass, we can query the $i$-th row of $A$ and the corresponding of the $b$. The goal is to design an LP solver that is both space and pass-efficient. By efficient, our objective is to obtain an algorithm with no polynomial dependence on $m$, or more concretely, we present a robust IPM framework that uses only $\wt O(n^2)$ space and $\wt O(\sqrt n \log(1/\epsilon))$ passes.\footnote{We use $\wt O(\cdot)$ notation to hide polylogarithmic dependence on $n$ and $m$. In this informal discussion, $\log(1/\epsilon)$ is shorthand for $\log(mR^2\|c\|_2/(r\epsilon))$, where the feasible region $\{x: Ax\geq b\}$ lies in a ball of radius $R$ and contains a ball of radius $r$ with known center (see \ifdefined\isarxiv Theorem~\ref{thm:nsquare_lsbarrier}\else Theorem 7.4 in Full version \cite{full}\fi).} To the best of our knowledge, this is the most efficient streaming LP algorithm that achieves a space and pass independent of $m$. Current best streaming algorithms for LP either require $\Omega(n)$ passes or $\Omega(n^2+m^2)$ space for $O(\sqrt{n})$ passes. For the regime of tall dense LP ($m\gg n$), our algorithm achieves the best space and passes.

The key ingredient for obtaining these LP algorithms is a paradigm shift from the time-efficient primal-dual IPM to a less time-efficient dual-only IPM~\cite{r88}. From a time perspective, dual-only IPM requires $\wt O(\sqrt{n}\log(1/\epsilon))$ iterations, with each iteration can be computed in $\wt O(mn+\poly(n))$ time. However, it is much more space-efficient than that of primal-dual approach. Specifically, we show that per iteration, it suffices to maintain an $n\times n$ Hessian matrix in place. To obtain $\wt O(\sqrt{n}\log(1/\epsilon))$ passes, we show that non-trivial quantities such as the Lewis weights~\cite{l78,cp15} can be approximated in $\wt O(1)$ passes with only $\wt O(n^2)$ space, by storing one $n\times n$ matrix for each of the $\wt O(1)$ steps of an iterative scheme.

Now that we have a space and pass-efficient IPM for general LP in the streaming model, we instantiate it with applications for graph problems in the semi-streaming model. In the semi-streaming model, each edge is revealed along with its weight in an online fashion and might subject to an adversarial order, and the algorithm is allowed to make multiples passes over the stream in $\wt O(n)$ space.\footnote{Some authors define the space in the streaming model to be the number of cells, where each cell can hold $O(\log n)$ bits or even a number with infinite precision. Our bounds remain unchanged even if each cell only holds $O(1)$ bits, i.e., when arithmetic only applies to $O(1)$-bits operands.} We particularly focus on the \emph{maximum weight bipartite matching} problem, in which the edges with weights are streamed to us, and the goal is to find a matching that maximizes the total weights in it. While there is a long line of research (\cite{ag11,k13,dno14,ag18,alt20} to name a few) on this problem, most algorithms can only compute an approximate matching, meaning that the weight is at least $(1-\epsilon)$ of the maximum weight. For the case of exact matching, a recent work~\cite{ajj+22} provides an algorithm that takes $n^{3/4+o(1)}$ passes in $\wt O(n)$ space for computing a maximum \emph{cardinality} matching. It remains an open question to compute an exact maximum \emph{weight} bipartite matching in semi-streaming model, with $o(n)$ passes.

For integer weights bounded by $\poly(n)$, we answer this question by presenting a semi-streaming algorithm that uses $\wt O(n)$ space and $\wt O(\sqrt{m})$ passes, this means that as long as the graph is relatively sparse, i.e., $m\leq n^{2-c}$ for a constant $c>0$, we achieve $o(n)$ passes. To obtain an $\wt O(n)$ space algorithm for \emph{any graph}, we require additional machinery; more specifically, for each iteration of our dual-only IPM, we need to compute the Newton step via a symmetric diagonal dominant (SDD) solve in $\wt O(n)$ space. Since the seminal work of Spielman and Teng~\cite{st04_sdd}, many efforts have been dedicated in designing a time-efficient SDD system solver~\cite{kmp10,kmp11,kosz13,ckm+14}. This solvers run in $\wt O(m)$ time with improved dependence on the logarithmic terms. However, all of them require $\wt \Theta(m)$ space. To achieve $\wt O(n)$ space, we make use of small-space spectral sparsifiers~\cite{klmss17} as preconditioners to solve the system in a space and pass-efficient manner.

Finally, we note that with $\wt O(n)$ space, we essentially solve the dual problem, which is the \emph{generalized minimum vertex cover} on bipartite graph. To turn a solution on vertices to a solution on edges, we utilize the isolation lemma~\cite{mvv87}, which makes the maximum weight matching unique, so that the edges with small slack near the end of the central path are exactly its edges. We implement the isolation lemma in $\wt O(n)$ bits via a construction due to~\cite{crs95}.

\subsection{Our contribution}

In this section, we showcase three main results of this paper and discuss their consequences.

The first result regards solving a general linear program in the streaming model with $\wt O(n^2)$ space and $\wt O(\sqrt{n}\log(1/\epsilon))$ passes.

\begin{theorem}[
\ifdefined\isarxiv
General LP, informal version of Theorem \ref{thm:nsquare_lsbarrier}
\else
General LP, informal version of Theorem 7.4 in Full version \cite{full}
\fi
]\label{thm:nsquare_lsbarrier_informal}

Given a linear program with $m$ constraints and $n$ variables and $m\geq n$ in the streaming model, whose feasible region lies in a ball of radius $R$ and contains a ball of radius $r$ with known center, there exists an algorithm that outputs an $\epsilon$-approximate solution to the dual program (Eq.~\eqref{eq:lp_primal_dual}) in $\wt O(n^2)$ space and $\wt O(\sqrt{n}\log(mR^2\|c\|_2/(r\epsilon)))$ passes.
\end{theorem}

By $\epsilon$-approximate solution, we mean that the algorithm finds $x\in \R^n$ such that $c^\top x-c^\top x^*\leq \epsilon$, where $x^*$ is the optimal solution. The key to obtain our result is a small space implementation of leverage score and Lewis weights, so that we can utilize the Lee-Sidford barrier~\cite{ls14},  with the number of passes depending on the smaller dimension.

In conjunction with an SDD solver in $\wt O(n)$ space, our next result shows that in the semi-streaming model, we can solve the minimum vertex cover problem on a bipartite graph with $\wt O(\sqrt{m})$ passes.

\begin{theorem}[
\ifdefined\isarxiv
Minimum vertex cover, informal version of Theorem \ref{thm:minimum_vertex_cover}
\else 
Minimum vertex cover, informal version of Theorem 9.7 in Full version \cite{full}
\fi 
]
    Given a bipartite graph $G$ with $n$ vertices and $m$ edges, there exists a streaming algorithm that computes a minimum vertex cover of $G$ in $\wt{O}(\sqrt{m})$ passes and $\wt{O}(n)$ space with probability $1-1/\poly(n)$.\footnote{We can actually solve a \emph{generalized} version of the minimum vertex cover problem in bipartite graph: each edge $e$ needs to be \emph{covered} for at least $b_e \in \Z^{+}$ times, with $b_e\leq\poly(n)$, where the case of $b = \mathbf{1}_m$ is the classic minimum vertex cover.}
\end{theorem}

The reason we end up with $\wt O(\sqrt{m})$ passes instead of $\wt O(\sqrt{n})$ passes is that to compute some fundamental quantities such as leverage scores or Lewis weights, we need to solve $\Theta(m)$ SDD systems and result in a total of $\wt O(m\sqrt{n})$ passes. By using the logarithmic barrier, we only need to solve $O(1)$ SDD systems per iteration, which gives the $\wt O(\sqrt{m})$ passes.

We are now ready to present our result for bipartite matching in $\wt O(\sqrt{m})$ passes, which solves the longstanding problem of whether maximum weight matching with polynomially bounded integer weights can be solved in $o(n)$ passes for any $m=n^{2-c}$ with $c>0$.

\begin{theorem}[
\ifdefined\isarxiv
Maximum weight bipartite matching, informal version of Theorem~\ref{thm:main_theorem}
\else
Maximum weight bipartite matching, informal version of Theorem 10.1 in Full version \cite{full}
\fi
]\label{thm:main_theorem_in_intro}
    Given a bipartite graph $G$ with $n$ vertices and $m$ edges, with integer edge weights $w\in\{0,1,\cdots,W\}^m$ where $W\leq \poly(n)$, there exists a streaming algorithm that computes an (exact) maximum weight matching of $G$ in $\wt{O}(\sqrt{m})$ passes and $\wt{O}(n)$ space with  probability $1-1/\poly(n)$.
\end{theorem}

Our matching result relies on turning a near-optimal solution to the dual minimum vertex cover problem, found near the end of the central path, into a primal solution for the maximum weight matching. We achieve so by an $\wt O(n)$ space implementation of the isolation lemma~\cite{mvv87,crs95}, which makes the maximum weight matching unique, so that the edges with small slack at this point are exactly the edges of this matching.

\subsection{Related work} \label{sec:related_work}

\paragraph{Interior point method for solving LP.}
The interior point method was originally proposed by Karmarkar \cite{k84} for solving linear program. Since then, there is a long line of work on speeding up interior point method for solving classical optimization problems, e.g., linear program \cite{v87,r88,v89_lp,nn89,ds08,ls13_path2,ls14,ls15,cls19,lsz19,ls19,b20,blss20,y21,jswz21,dly21,sy21,gs22}. 
In 1987, the running time of LP solver becomes $O(n^3)$ \cite{v87,r88}.  
In 1989, Vaidya proposed an $O(n^{2.5})$ LP solver based on a specific implementation of IPMs, known as the \emph{central path algorithm} \cite{v87, v89_lp}. Lee and Sidford show how to solve LP in $\sqrt{n}(\nnz(A) + n^{\omega})$ time \cite{ls13_path1,ls13_path2,ls14}, where $\omega$ is the exponent of matrix multiplication \cite{w12,l14,aw21}\footnote{Currently, $\omega<2.258$ over every field of characteristic zero~\cite[Theorem~1.1 and Corollaries~14.1--14.2]{OpenAIMM26}.} (the first $\sqrt{n}$-iteration IPM).  In 2019, 
\cite{cls19} show how to solve LP in $n^{\omega} + n^{2.5 - \alpha/2} + n^{2+1/6}$ time, where $\alpha$ is the dual exponent of matrix multiplication \cite{gu18}\footnote{Currently, $\alpha>0.465$ over every field of characteristic zero~\cite[Theorem~1.1 and Corollaries~14.1--14.2]{OpenAIMM26}.}. This is the first breakthrough result improving $O(n^{2.5})$ from 30 years ago. 
Later, 
\cite{jswz21} improved that running time to $n^{\omega} + n^{2.5 - \alpha/2} + n^{2+1/18}$ by maintaining two layers of data-structure instead of one layer of data-structure as \cite{cls19}'s algorithm. 
In 2020, \cite{blss20} improved the running time of LP solver on tall matrices to $mn$ when $m \geq \poly(n)$. Another line of work focuses on solving linear program with small treewidth~\cite{dly21,y21} in time $\wt O(m\tau^2)$.

\paragraph{Small space algorithms for solving LP.}
Simplex algorithm is another popular approach to solve linear programs. It has an even better compatibility with streaming algorithms. 
For instance, \cite{cc07} shows that the non-recursive implementation of Clarkson's algorithm \cite{c95} gives a streaming LP solver that uses $O(n)$ passes and $\wt{O}(n \sqrt{m})$ space.
They also show that the recursive implementation gives a streaming LP solver that uses $n^{O(1/\delta)}$ passes and $(n^2 + m^{\delta}) \poly(1/\delta)$ space.
\cite{akz19} proposes a streaming algorithm for solving $n$-dimensional LP that uses $O(n r)$ pass and $O(m^{1/r} ) \poly(n, \log m)$ space, where $r \geq 1$ is a parameter. All above algorithms needs space depending on $m$.

\paragraph{Streaming algorithms for approximate matching.}
Maximum matching has been extensively studied in the streaming model for decades, where almost all of them fall into the category of approximation algorithms. 
For algorithms that only make one pass over the edges stream, researchers make continuous progress on pushing the constant approximation ratio above $1/2$, which is under the assumption that the edges are arrived in a uniform random order
\cite{kks14, abb+19, fhm+20, b20_matching}. 
The random-order assumption makes the problem easier (at least algorithmically). 
A more general setting is multi-pass streaming with adversarial edge arriving.
Under this setting, the first streaming algorithm that beats the $1/2$-approximation of bipartite cardinality matching is \cite{fkm+04}, giving a $2/3 \cdot (1-\eps)$-approximation in $1/\eps \cdot \log(1 / \eps)$ passes.
The first to achieve a $(1-\eps)$-approximation is \cite{m05}, which takes $(1/\eps)^{1/\eps}$ passes.\footnote{For the weighted case, there is a $(1/2-\eps)$-approximation algorithm that only takes one pass \cite{ps17}.}
Since then, there is a long line of research in proving upper bounds and lower bounds on the number of passes to compute a maximum matching in the streaming model
\cite{ag11, ekms12, gkk12, ekms12, k13, dno14, ag18, aksy20, ar20, alt20} (see next subsection for more details).
Notably, \cite{ag11, ag18} use linear programming and duality theory (see the next subsection for more details). 

However, all the algorithms above can only compute an approximate maximum matching: to compute a matching whose size is at least $(1 - \epsilon)$ times the optimal, one needs to spend $\text{poly}(1/\epsilon)$ passes (see \cite{dno14, ag18} and the references therein). For readers who are interested in the previous techniques for solving matching, we refer to 
\ifdefined\isarxiv
Section~\ref{sec:summary_non_ipm}
\else
Section~A in full version \cite{full}.
\fi
which contains a brief summary.

\paragraph{Recent developments for exact matching.}
Recently, \cite{ajj+22} proposes an algorithm that computes a $(1-\eps)$-approximate maximum \emph{cardinality} matching in $O(\eps^{-1}\log n \log\eps^{-1})$ passes and $\wt{O}(n)$ space. 
Their method leverages recent advances in $\ell_1$-regression with several ideas for implementing it in small space, leading to a streaming algorithm with no dependence on $\eps$ in the space usage, and thus improving over \cite{ag18}. 
The resulted semi-streaming algorithm computes an exact maximum \emph{cardinality} matching (not for weighted) in $n^{3/4 + o(1)}$ passes.

\paragraph{Streaming spectral sparsifer.}
Initialized by the study of cut sparsifier in the streaming model \cite{ag09}, a simple one-pass semi-streaming algorithm for computing a spectral sparsifier of any weighted graph is given in \cite{kl11}, which suffices for our applications in this paper.
The problem becomes more challenging in a dynamic setting, i.e., both insertion and deletion of edges from the graph are allowed. 
Using the idea of linear sketching, \cite{klmss17} gives a single-pass semi-streaming algorithm for computing the spectral sparsifier in the dynamic setting. 
However, their brute-force approach to recover the sparsifier from the sketching uses $\Omega(n^2)$ time. 
An improved recover time is given in \cite{kmm+19} but requires more spaces, e.g., $\epsilon^{-2} n^{1.5} \log^{O(1)} n$.
Finally, \cite{knst19} proposes a single-pass semi-streaming algorithm that uses $\eps^{-2} n\log^{O(1)} n$ space and $\eps^{-2} n\log^{O(1)} n$ recover time to compute an $\eps$-spectral sparsifier which has $O(\epsilon^{-2} n \log n)$ edges. 
Note that $\Omega( \epsilon^{-2}  n \log n)$ space is necessary for this problem \cite{ckst19}.

\paragraph{SDD solver.}
There is a long line of work focusing on fast SDD solvers \cite{st04,kmp10,kmp11,kosz13,ckm+14,ps14,ks16}. 
Spielman and Teng give the first nearly-linear time SDD solver, which is simplified with a better running time in later works. 
The current fastest SDD solver runs in $O(m \log^{1/2} n \poly(\log\log n) \log(1/\eps))$ time \cite{ckm+14}. 
All of them require $\wt{\Theta}(m)$ space. 

\paragraph{Subsequent work.} After the first version of this paper, \cite{ajj+22} gave a semi-streaming algorithm that computes an exact maximum cardinality bipartite matching in $n^{3/4+o(1)}$ passes (\ifdefined\isarxiv Section~\ref{sec:previous_techniques_approx_to_exact}\else Section~A.2 in the full version~\cite{full}\fi). This improves upon our $\widetilde O(\sqrt m)$ passes for unweighted graphs with $m\geq n^{3/2+o(1)}$. Subsequently, van den Brand, Song and Weng~\cite{bsw26} gave a multi-pass streaming algorithm for min-cost flow with integer capacities and costs at most $W$. It returns the flow on each edge up to additive error $\epsilon$ in $\widetilde O(\sqrt n\log(W/\epsilon))$ passes and $\widetilde O(n^{1.5}\log(W/\epsilon))$ space. Like ours, their algorithm is randomized. Rounding its output gives the exact minimum cost, but it gives an exact flow on every edge only when the min-cost flow is unique, so it does not by itself output an exact maximum weight matching. It covers capacitated flows, to which our dual-only IPM does not extend, and uses $\widetilde O(\sqrt n)$ rather than $\widetilde O(\sqrt m)$ passes, but more than $\widetilde O(n)$ space. For lower bounds, \cite{ckpssy21} shows that semi-streaming algorithms that decide whether a bipartite graph has a perfect matching need $\Omega(\sqrt{\log n})$ passes. For LP in the streaming model, \c{C}ekirge, Gay and Woodruff~\cite{cgw25} gave randomized multi-pass linear sketches that succeed with high probability, with $O(n^2\log(1/\epsilon))$ passes and $\operatorname{poly}(n,\log(1/\epsilon))$ space. They use more passes than our $\widetilde O(\sqrt n\log(mR^2\|c\|_2/(r\epsilon)))$, which holds when the feasible region lies in a ball of radius $R$ and contains a ball of radius $r$ with known center. Their solution may violate each constraint by $O(\epsilon)$ and they assume bounded data, but their bounds do not depend on $m$ and they also handle deletions of constraints.

\section{Technical overview}

We start with an overview of our IPM framework. We first note that many recent fast IPM algorithms do not fit into $\wt O(n^2)$ space. Algorithms such as~\cite{ls14,jswz21,blss20} need to maintain \emph{both} primal and dual solutions, thus require $\Omega(m)$ space. In fact, any algorithms that rely on the primal formulation will need $\Omega(m)$ space to maintain the solution. To bypass this issue, we draw inspiration from the state-of-the-art SDP solver~\cite{hjst21}: in their setting, $m=\Omega(n^2)$, which means any operation on the dimension $m$ will be too expensive to perform. They instead resort to the \emph{dual-only} formulation. The dual formulation provides a more straightforward optimization framework on small dimension and makes it harder to maintain key quantities. This is exactly what we want: an algorithm that operates on the smaller dimension, removing the polynomial dependence on $m$. While efficient maintenance is the key to design time-efficient IPM, it is less a concern for us since our constraining resource is space, not time. To this end, we show that Renegar's IPM algorithm~\cite{r88} can be implemented in a streaming fashion with only $\wt O(n^2)$ space. As the number of passes of an IPM crucially depends on the barrier function being used, the~\cite{r88} algorithm only gives a pass bound of $\wt O(\sqrt m\log(1/\epsilon))$. To further improve the number of passes required, we show that the nearly-universal barrier of Lee and Sidford~\cite{ls14,ls19} can also be implemented in $\wt O(n^2)$ space. This involves computing Lewis weights in an extremely space-efficient manner. We present an algorithm with $\wt O(1)$ steps, based on the parallel algorithm of~\cite{flps21}, that stores one $n\times n$ matrix per step and so uses only $\wt O(n^2)$ space and $\wt O(1)$ passes. This gives the desired $\wt O(\sqrt n\log(1/\epsilon))$ passes.

We now turn to our graph results, which is a novel combination of the space-efficient IPM, SDD solvers, duality and the isolation lemma. Note that for both graph problems only allow $\wt O(n)$ space, so it won't suffice to directly apply our IPM algorithms .

To give a better illustration of the $\wt{O}(n)$ space constraint, note that storing a matching already takes $\wt{\Theta}(n)$ space, meaning that we have only a 
polylogarithmic space overhead per vertex to store auxiliary information. 
The conventional way of solving maximum bipartite matching using an IPM solver would get stuck at 
the very beginning - maintaining the solution of the relaxed linear program, which is a fractional matching, already requires $\Omega(m)$ space for storing all 
LP constraints, which seems inevitable. 

Our key insight is to show that solving the \emph{dual} form of the above LP, which corresponds to the generalized (fractional) minimum vertex cover problem, is sufficient, and therefore only $\wt{O}(n)$ space is needed for maintaining a fractional solution. 
We use several techniques to establish this argument. 
The first idea is to use complementary slackness for the dual solution to learn which edges (at most $n/2$ of them) will be in the final maximum matching 
and therefore reduce the size of the graph from $m$ to at most $n/2$ edges. 
However, this is not always the case: For instance, in a bipartite graph that admits a perfect matching, all left vertices form a minimum vertex cover, but the complementary slackness theorem gives no information on which edges are in the perfect matching. 
To circumvent this problem, we need to slightly perturb the weight on every edge, so that the 
maximum weight matching becomes unique and the near-optimal vertex cover found near the end of the central path indeed provides enough information. 
We use the isolation lemma \cite{mvv87} to realize this objective.

It is then instructive to implement the isolation lemma in limited space. 
Perturbing the weight on every edge 
requires storing $\wt{O}(m)$ bits of randomness, since the perturbation should remain identical across two different passes. 
We bypass this issue by using the 
generalized isolation lemma proposed by \cite{crs95}, in which only $O(\log(Z)+\log m)$ bits of randomness are needed when the ground set has $m$ elements and $Z\leq 2^{\poly(m)}$, where $Z$ is the number of candidates. 
In our case, $Z\leq n^n$ is the number of all possible matchings. So $\wt{O}(n)$ space usage perfectly fits into the semi-streaming model. 
We design an oracle that stores $\wt{O}(n)$ random bits and outputs the same perturbations for all edges in all passes.

Now that we can focus on solving the minimum vertex cover problem in $\wt O(n)$ space. When the constraint matrix is an incidence matrix, each iteration of our IPM can be implemented as an SDD (or Laplacian) solver, so it suffices to show how to solve SDD system in the semi-streaming model, which, to the best of our knowledge, has not been done prior to our work.

In the following subsections we elaborate on each of the above components: 
\begin{itemize}
	\item In Section~\ref{sec:our_dual_only}, we provide a high-level picture of how our dual-only interior point method works. 
 
	\item In Section~\ref{sec:evidence_space}, we show evidences that our interior point method can run in space independent of $m$ for all of the three different barriers.
 
	\item In Section~\ref{sec:our_techniques_streaming_implementation}, we describe our contribution on our implementations of SDD solver, IPM, and the isolation lemma in the streaming model. We show a novel application of the isolation lemma to turn dual into primal.
\end{itemize}

\subsection{Dual-only robust IPM}\label{sec:our_dual_only}

The cornerstone of our results is to design a robust IPM framework that works only on the dual formulation of the linear program. The framework fits in barriers including the logarithmic barrier, hybrid barrier and Lee-Sidford barrier. It is also robust enough as it can tolerate approximation errors in many quantities, while preserving the convergence behavior.

Algorithm~\ref{alg:ipm_in_intro} is a simplified version of our dual-only IPM. The earlier works of Renegar's algorithm \cite{r88} require the Newton's direction be computed exactly as $\Delta x = -H(x)^{-1} \nabla f_t(x)$, in order to get double exponential convergence rate of Newton's method. To strengthen its guarantee, we develop a more robust framework for this IPM. Specifically, we show that the Hessian of the barrier functions, the gradient and the Newton's direction can all be approximated. This requires a much more refined error analysis. Below, we carefully bound the compound errors caused by three layers of approximations.
 
First, from $\Delta x$ to $\delta_x$ (Line~\ref{line:intro_delta_x}), we allow our Hessian to be spectrally approximated within a constant factor. This provides us enough leeway to implement the Hessian of barrier functions in a space-efficient manner. For example, the Hessian of the volumetric barrier is $H(x) = A_x^{\top}(3\Sigma_x - 2P_x^{(2)})A_x$, where $\Sigma_x$ is a diagonal matrix and $P_x^{(2)}$ is taking entry-wise square of a dense projection matrix. But $\wt{H}(x) = A_x^{\top}\Sigma_x A_x$ is a 5-approximation of $H(x)$ and we can compute it in the same space as computing leverage scores.

Second, from $\delta_x$ (Line~\ref{line:intro_delta_x}) to $\delta_x'$ (Line~\ref{line:intro_delta_x_2}), we allow approximation on the gradient in the sense that it has small local norm with respect to the true gradient, i.e., $\|\nabla f_t(x) - \wt{\nabla} f_t(x)\|_{H(x)^{-1}} \leq 10^{-4}$.\footnote{For a vector $y$ and a positive semidefinite matrix $A$, we define $\| y \|_A := \sqrt{y^\top A y}$.} To give a concrete example, let $\sigma\in \R^m$ denote the leverage score vector, and suppose the Hessian matrix is in the form of $H(x) = A^{\top}\Sigma A$ and the gradient is $\nabla f(x) = A^\top \sigma$. The leverage score $\sigma$ can then be approximated in an entry-wise fashion: each entry can tolerate a multiplicative $(1\pm 10^{-4}/\sqrt{n})$ error. This is because
\begin{align*}
    \|\nabla f_t(x) - \wt{\nabla} f_t(x)\|_{H(x)^{-1}}^2 
    = & ~ \mathbf{1}_m^{\top}(\Delta \Sigma) A(A^{\top}\Sigma A)^{-1} A^{\top} (\Delta \Sigma) \mathbf{1}_m\\
    = & ~ \mathbf{1}_m^{\top}(\Delta \Sigma) \Sigma^{-1/2}\Sigma^{1/2} A(A^{\top}\Sigma A)^{-1} A^{\top} \Sigma^{1/2}\Sigma^{-1/2} (\Delta \Sigma) \mathbf{1}_m\\
    \leq & ~ \mathbf{1}_m^{\top}(\Delta \Sigma) \Sigma^{-1} (\Delta \Sigma )\mathbf{1}_m\\
    = & ~ \sum_{i=1}^m \frac{(\sigma_i - \wt{\sigma}_i)^2}{\sigma_i}\\
    \leq & ~ \frac{10^{-8}}{n}\cdot n = 10^{-8}, 
\end{align*}

where the first inequality follows from property of projection matrix (for any projection matrix $P$, we have $P \preceq I$. Then we know $x^\top P x \leq x^\top x$ for all vector $x$), 
the last inequality follows from $\sum_{i=1}^m \sigma_i = n$.

Third, from $\delta_x'$ (Line~\ref{line:intro_delta_x_2}) to $\wt{\delta}$ (Line~\ref{line:intro_wt_delta_x}), we can tolerate the approximation error on the Newton's direction $\| \wt{\delta} - \delta_x' \|_{ H(x) } \leq 10^{-4}$. This is crucial for our graph applications, since we need to use small space SDD solver to approximate the Newton's direction.

\begin{algorithm}[!ht]
\caption{A simplified version of our algorithm}\label{alg:ipm_in_intro}
\begin{algorithmic}[1]
\Procedure{OurAlgorithm}{$A \in \R^{m\times n},b\in\R^m,c \in \R^n$} 
\State Choose $F(x) \in \R^{n} \rightarrow \R$ to be any $\theta^2$-self concordant barrier function
\State Let $f_t(x):=t\cdot c^{\top}x + F(x) \in \R$
\State Let $H(x) := \nabla^2 F(x) \in \R^{n\times n}$
\State Let $T$ be the number of iterations
\State Initialize $x,t$
\For {$k \leftarrow 1$ to $T$} 
    \State Let $\wt{H}(x)$ be any PSD matrix that $\frac{1}{5} \wt{H}(x) \preceq H(x) \preceq \wt{H}(x)$ \label{line:intro_wt_h}
    \State Let $\delta_x := -\wt{H}(x)^{-1}\cdot \nabla f_t(x)$  \label{line:intro_delta_x}
    \State Let $\wt{\nabla} f_t(x)$ be that $\| \wt{\nabla}f_t(x) - \nabla f_t(x) \|_{H(x)^{-1}} \leq 10^{-4}$  \label{line:intro_wt_nabla_f_t}
    \State Let $\delta_x' := -\wt{H}(x)^{-1}\cdot \wt{\nabla}f_t(x)$  \label{line:intro_delta_x_2}
    \State Let $\wt{\delta}_x$ be any vector that $\|\wt{\delta}_x - \delta_x'\|_{H(x)} \leq 10^{-4}$   \label{line:intro_wt_delta_x}
    \State $x \leftarrow x + \wt{\delta}_x$ 
    \State $t \leftarrow t\cdot (1+10^{-4}\theta^{-1})$
\EndFor 
\State Output $x$
\EndProcedure
\end{algorithmic}
\end{algorithm}

\subsection{Solve LP in small space}
\label{sec:evidence_space}

In this section, we show how to implement our IPM in space not polynomially dependent on $m$ for different barrier functions. 

For three barriers (logarithmic, hybrid and Lee-Sidford), the Hessian we use (the exact Hessian for the logarithmic barrier, and its spectral approximation $\wt{H}(x)$ for the other two) takes the form of $A^\top HA\in \R^{n\times n}$ for an $m\times m$ positive diagonal matrix $H$. For logarithmic barrier, $H_{i,i}=s_i(x)^{-2}$, as $s_i(x)$ can be computed in $O(1)$ space, it is not hard to see that the Gram matrix can be computed as $\sum_{i\in [m]} H_{i,i}\cdot a_ia_i^\top$ in $O(n^2)$ space. 

The more interesting case is to consider the hybrid barrier and Lee-Sidford barrier. The gradient and Hessian of the hybrid barrier requires us to compute $m$ leverage scores defined as ${\rm diag}(\sqrt{H}A(A^\top HA)^{-1}A^\top \sqrt{H})$. Forming this projection matrix will require a prohibitive $m^2$ space. To implement it in $n^2$ space, we rely on an observation that $\sigma_i=H_{i,i}\cdot a_i^\top (A^\top HA)^{-1}a_i$, thus, if we can manage to maintain $(A^\top HA)^{-1}$ in $O(n^2)$ space, then we can compute the leverage score. Similar to the logarithmic barrier scenario, $A^\top HA$ can be computed in 1 pass and $O(n^2)$ space, then the inverse can be computed in $O(n^2)$ space. Thus, we can supply the $i$-th leverage score in $O(n^2)$ space, and compute the gradient and Hessian in designated space constraint.

Given an oracle that can compute the $i$-th leverage score in $O(n^2)$ space, we can even implement the $\ell_{\log m}$ Lewis weights in $\wt O(n^2)$ space. To do so, we rely on the parallel iterative scheme introduced in~\cite{flps21}. Every step of this scheme updates the weights of all rows, or of a subset of the rows, at once, and it only needs the leverage scores for the current weights. Unfortunately, as we are only allowed a space budget of $\wt O(n^2)$, we cannot store the intermediate weights. To circumvent this issue, we store one $n\times n$ matrix for each step, the inverse of $A^\top WA$ for the current weights $W$, together with one bit that records which of the two kinds of step is taken. This matrix and this bit take two passes to compute. When a row $a_i$ is read, its weights in all prior steps can then be recomputed from $a_i$, the stored matrices and the stored bits, with no extra pass. The algorithm uses at most $O(\poly(\log m))$ steps, therefore, we can compute approximate Lewis weights in $\wt O(n^2)$ space and $\poly(\log m)$ passes.

\subsection{Semi-streaming maximum weight bipartite matching in \texorpdfstring{$\wt O(\sqrt{m})$}{} passes}\label{sec:our_techniques_streaming_implementation}

Recall that in the semi-streaming model, we are only allowed with $\wt O(n)$ space. For the IPMs we've developed before, we can not meet such space constraint. 
For general graphs, we have to invent more machinery to realize the $\wt O(n)$ space.

For matching, we start by noting that the constraint matrix $A\in \R^{m\times n}$ is a graph incidence matrix. This means that for logarithmic barrier, the Hessian matrix $A^\top S^{-2}A$ can be treated as a Laplacian matrix with edge weight $s_i^{-2}$. Therefore, computing the Newton direction reduces to perform an SDD solve in $\wt O(n)$ space.

{\bf SDD solver in the semi-streaming model.}
Though solving SDD system can be done in an extremely time-efficient manner, it is unclear how to compute them when only $\wt O(n)$ space is allowed. To circumvent this problem, we rely on two crucial observations. Let $L_G$ denote the SDD matrix corresponding to the Hessian.
\begin{itemize}
    \item Solving a system $L_G \cdot x=b$ will require $\wt \Omega(m)$ space, but multiplying $L_G$ with a vector $v\in \R^n$ can be done in $O(n)$ space: as $L_G=\sum_{i\in [m]} \frac{a_ia_i^\top}{s_i^2}$, $L_G \cdot v$ can be computed as $a_i(a_i^\top v)/s_i^2$ in $O(n)$ space, and accumulate the sum over a pass of the graph.
    \item Suppose we have a sparse graph $H$ with only $\wt O(n)$ edges, then the system $L_H\cdot x=b$ can be solved in $\wt O(n)$ space.
\end{itemize}

It turns out that these two observations are enough for us to solve a general SDD system in $\wt O(n)$ space. Given the graph $G$, we first compute a $(1\pm\delta)$-spectral sparsifier with only $\wt O(\delta^{-2}n\log^{O(1)} n)$ edges in a single pass~\cite{knst19}. Let $H$ denote this sparsifier, we then use $L_H^{-1}$ as a preconditioner for solving our designated SDD system. More concretely, let $r_t:=b-L_G\cdot x_t$ denote the residual at $t$-th iteration, we solve the system $L_H\cdot y_t=r_t$. As $y_t=L_H^{-1}\cdot b-L_H^{-1} L_G\cdot x_t$, we can then update the solution via the preconditioned-solution $x_{t+1}=x_t+y_t$. The residual is then $r_{t+1}=b-L_G\cdot x_{t+1}=b-L_G\cdot x_t-L_G\cdot y_t=r_t-L_G\cdot y_t$, i.e., we only need to implement one matrix-vector product with $L_G$. After $\wt O(1)$ iterations, we have refined an accurate enough solution for the SDD system.

{\bf From dual to primal.}
Though we can solve the dual in $\wt O(n)$ space, it only produces a solution to the minimum vertex cover and we need to transform it to a solution to maximum weight matching.

Turning an optimal dual solution to an optimal primal solution for general LP requires at least solving a linear system, which takes $O(n^{\omega})$ time and $O(mn)$ space (Lee, Sidford and Wong~\cite{lsw15}), which is unknown to be implemented in the semi-streaming model even for bipartite matching LP.\footnote{In general, the inverse of a sparse matrix can be dense, which means the standard Gaussian elimination method for linear system solving does not apply in the semi-streaming model.\label{footnote:inverse}} We bypass this issue by using the complementary slackness theorem to highlight at most $n/2$ tight constraints of a strictly complementary optimal dual solution and therefore sparsify the original graph from $m$ edges to at most $n/2$ edges without losing the optimal matching. However, this is only true if the solution to the primal LP is {\bf unique}.

\begin{figure}[!ht]
    \centering
    \includegraphics[width=15cm]{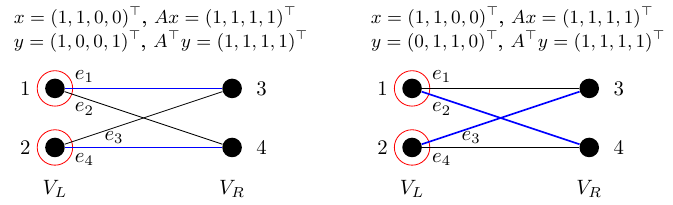}
    \caption{ 
    The red circle is a minimum vertex cover, which is an optimal dual solution. The blue edge is a maximum matching, which is an optimal primal solution. In both examples, the primal and dual satisfy complementary slackness Eq.~\eqref{eq:com_sla_theorem}.
    }
    \label{fig:bipartite_graph}
\end{figure}

To give a better illustration, let us consider a simple example. Suppose the graph has a (maximum weight) perfect matching (see Figure~\ref{fig:bipartite_graph} for example). Then the following trivial solution is optimal to the dual LP: choosing all vertices in $V_L$. 
Let us show what happens when applying the complementary slackness theorem. 
The complementary slackness theorem says that when $y$ is a feasible primal solution and $x$ is a feasible dual solution, then $y$ is optimal to the primal and $x$ is optimal to the dual {\bf if and only if} 
\begin{equation}\label{eq:com_sla_theorem}
    \langle y, Ax - \mathbf{1}_m\rangle = 0 \mathrm{~~and~~} \langle x, \mathbf{1}_n - A^{\top} y \rangle = 0.
\end{equation}
From the above case, we have $Ax - \mathbf{1}_m = 0$, so the first equality $\langle y, Ax - \mathbf{1}_m\rangle = 0$ puts no constraint on $y$. Therefore, any solution $y\geq \mathbf{0}_m$ to the linear system
$a_i^{\top} y = 1, ~\forall i\in V_L$ is an optimal solution, where $a_i$ is the $i$-th column of $A$. 
Note that this linear system has $m$ variables and $|V_L|$ equations, which is still hard to find a solution in $\wt{O}(n)$ space.

Now consider perturbing the primal objective function by some vector $b\in \R^m$ such that the optimal solution to the following primal LP is {\bf unique}:
\begin{align*}
    \textbf{Primal}~~~ \max_{y \in \R^m} &~ b^{\top} y,
     ~~~~~\text{s.t. } A^{\top} y \leq \mathbf{1}_n \text{ and }
    y\geq {\bf 0}_m .
\end{align*}

Suppose we find the optimal solution $x$ in the dual LP and we want to recover the optimal solution $y$ in the primal LP. 
Since the optimal $y$ is unique, it is a vertex of the feasible region, hence a matching with at most $n/2$ edges. 
By complementary slackness, every optimal $x$ satisfies $\langle y, Ax-b \rangle = 0$. 
These equations have the form 
\begin{equation*}
    y_i = 0,~\forall i \in [m] \text{~s.t.~} (Ax)_i-b_i > 0. 
\end{equation*}
This means that the corresponding edges are \emph{unnecessary} in order to get one maximum matching. An arbitrary optimal $x$ can have many more than $n$ tight edges, where an edge is tight if $(Ax)_i=b_i$. However, by the Goldman--Tucker theorem~\cite{gt56} there is an optimal $x$ that is strictly complementary to $y$, that is, $(Ax)_i-b_i>0$ for every edge $i$ outside the matching. The central path converges to such an $x$. So near the end of the central path, the edges with small slack $(Ax)_i-b_i$ are exactly the edges of the optimal matching. 
As a result, we can reduce the number of edges from $m$ to at most $n/2$, then compute a maximum matching in $\wt{O}(n)$ space without reading the stream.

{\bf Isolation lemma in the semi-streaming model.} It remains to show how to perturb the objective so that the primal solution is unique. As the perturbation is over all edges, one natural idea is to randomly perturb them using $\wt O(m)$ bits of randomness. This becomes troublesome when the random bits need to be stored since the perturbation should remain consistent across different passes. We resolve this problem via the isolation lemma.

Let us recall the definition of the isolation lemma (see 
\ifdefined\isarxiv
Section~\ref{sec:isolation_lemma} 
\else
Section C in full version \cite{full} 
\fi
for details).
\begin{definition}[Isolation lemma]
Given a set system $(S,\mathcal{F})$ where $\mathcal{F} \subseteq \{0,1\}^S$. Given weight $w_i$ to each element $i$ in $S$, the weight of a set $F$ in $\mathcal{F}$ is defined as $\sum_{i\in F} w_i$.
The isolation lemma says there exists a scheme that can assign weight oblivious to $\mathcal{F}$, such that there is a unique set in $\mathcal{F}$ that has the minimum (maximum) weight under this assignment.
\end{definition}

The isolation lemma says that if we randomly choose weights, then with a good probability the uniqueness is ensured. However, 
this does not apply to the streaming setting since the weight vector is over all edges, which require $\Omega(m)$ space.

To apply isolation lemma for bipartite matching, we note that the set $S$ is all the edges and the family ${\cal F}$ contains all maximum weight matchings. The total number of possible matchings is at most $(n+1)^n$, as each vertex can choose none or one of the vertices to match. We leverage this parameterization and make use of~\cite{crs95}, which requires $O(\log Z+\log m)$ random bits for any $Z\geq|{\cal F}|$ with $Z\leq 2^{\poly(m)}$. For matching, we only need $O(n\log n)$ bits, which suits in our space budget. To the best of our knowledge, this is the first use of isolation lemma in the streaming model.

\subsection{Discussions}\label{sec:discussion}

For matching, improving $\sqrt{m}$ passes to $\sqrt{n}$ passes will require us to compute fundamental quantities such as leverage scores and Lewis weights by solving $\wt O(1)$ SDD systems. As directed reachability can be solved in $n^{1/2+o(1)}$ passes by a semi-streaming version of the parallel algorithm of~\cite{jls19} (\ifdefined\isarxiv Section~\ref{sec:previous_techniques_approx_to_exact}\else Section A.2 in full version \cite{full}\fi, see also~\cite{ajj+22}), and breadth-first search trees of undirected graphs, that is, single-source shortest paths in unweighted undirected graphs, can be computed by a randomized algorithm in $O(\sqrt{n})$ passes~\cite{cfht20} in the semi-streaming model, we believe it is an important open problem to close the gap between bipartite matching and these problems.

\section*{Acknowledgement}

The authors would like to thank Jonathan Kelner for many helpful discussions and anonymous reviewers for numerous comments to improve the presentation of this paper. Lichen Zhang is supported by NSF grant CCF-1955217 and NSF grant DMS-2022448. The corrections in this version were found and written with the help of Claude Opus 5.5. The authors have verified the corrections and take full responsibility for the content of this paper.
\newpage

\bibliographystyle{alpha}
\bibliography{ref}
\addcontentsline{toc}{section}{Bibliography}

\ifdefined\isarxiv

\newpage
{\bf Roadmap.} 
The rest of the paper is organized as follows. 
In Section~\ref{sec:notations}, we define the basic notations in this paper. 
In Section~\ref{sec:ipm_preliminary}, we give some preliminaries for interior point method. 
In Section~\ref{sec:ipm_algorithm}, we present the robust dual central path method. 
In Section~\ref{sec:ipm_error_analysis}, we present the error analysis of interior point method. 
In Section~\ref{sec:lp_in_small_space}, we bound the number of pass of executing this
algorithm in the streaming model.
In Section~\ref{sec:sdd_solver_streaming_model}, we give an SDD solver in the streaming model, which is a necessary component of our interior point method for graphs. 
In Section~\ref{sec:minimum_vertex_cover}, we discuss the streaming algorithm for minimum vertex cover. 
In Section~\ref{sec:combine}, we combine the pieces together to get the final algorithm for maximum weight bipartite matching. In Section~\ref{sec:summary_non_ipm}, we discuss previous non-IPM algorithms. 
In Section~\ref{sec:solver_reduction}, we complement  Section~\ref{sec:sdd_solver_streaming_model} by providing two reductions from weaker solvers to our final SDD$_0$ solver.
In Section~\ref{sec:isolation_lemma}, we present the generalized isolation lemma in the semi-streaming model, which will be used to recover the maximum matching from a minimum vertex cover. 
In Section~\ref{sec:additional_algorithms}, we present the small space implementations of various barrier functions in Section \ref{sec:lp_in_small_space}. 
In Section~\ref{sec:lp_in_small_treewidth}, we provide a more space-efficient algorithm when the linear program has small treewidth.

\section{Notations} \label{sec:notations}

{\bf Standard notations.}
For a positive integer $n$, we denote $[n]=\{1,2,\cdots,n\}$.

We use $\E[\cdot]$ for expectation and $\Pr[\cdot]$ for probability.

For a positive integer $n$, we use $I_n$ to denote the identity matrix of size $n\times n$.

For a vector $v \in \R^n$, we use the standard definition of $\ell_p$ norms: $\forall p\geq 1$, $\|v\|_p = (\sum_{i=1}^n |v_i|^p )^{1/p}$. Specially, $\|v\|_\infty = \max_{ i \in [n] } | v_i |$. We use $\| v \|_0$ to denote the number of nonzero entries in vector $v$. We use $\supp(v)$ to denote the support of vector $v$.

We use ${\bf 0}_n$ to denote a length-$n$ vector where every entry is $0$. We use ${\bf 0}_{m \times n}$ to denote a $m \times n$ matrix where each entry is $0$. Similarly, we use the notation ${\bf 1}_n$ and ${\bf 1}_{m \times n}$.

For matrix $B$, we use $b_i$ to denote the $i$-th row of $B$.

{\bf Matrix operators.}
For a square matrix $A$, we use $\tr[A]$ to denote its trace. For a square and full rank matrix $A$, we use $A^{-1}$ denote the true inverse of $A$. For a matrix $A$, we use $A^\dagger$ to denote its pseudo inverse. We say a square matrix $A$ is positive definite, if for all $x$, $x^\top A x > 0$. We say a square matrix $A$ is positive semi-definite, if for all $x$, $x^\top A x \geq 0$. We use $\succeq$ and $\preceq$ to denote the p.s.d. ordering. For example, we say $A \succeq B$, if $x^{\top} A x \geq x^\top B x, \forall x$.

{\bf Matrix norms.} For a matrix $A$, we use $\| A \|_1$ to denote its entry-wise $\ell_1$ norm, i.e., $\| A \|_1 = \sum_{i,j} |A_{i,j}|$. We use $\| A \|_F$ to denote its Frobenius norm $\| A \|_F = ( \sum_{i,j} A_{i,j}^2 )^{1/2}$. We use $\| A \|$ to denote its spectral/operator norm. 

{\bf Matrix approximation.}
Let $A,B \in \R^{n\times n}$ be positive semi-definite matrix. Let $\epsilon \in (0,1)$. We say $A\approx_{\epsilon} B$ if
\[
(1-\epsilon) x^{\top} A x \leq x^{\top} B x \leq (1 + \epsilon) x^{\top} A x,~~\forall x\in \R^n.
\]
Note that if we have $A \approx_{\epsilon} B$, then $(1-\epsilon) \|x\|_A \leq \|x\|_B \leq (1 + \epsilon) \|x\|_A$ for all $x\in \R^n$.

{\bf Graph and corresponding matrices.}
We first give the definition of edge-vertex incident matrix.
\begin{definition}[Edge-vertex incident matrix]
    Let $G=(V_L, V_R, E)$ be a connected undirected bipartite graph. The \emph{(unsigned) edge-vertex incidence matrix} is denoted as follows
    \begin{align*}
        B(e,v)
        = \begin{cases}
        1, & \mathrm{~if~} v \mathrm{~incident~to~} e; \\
        0, & \mathrm{~otherwise}.
        \end{cases}
    \end{align*}
\end{definition}
In addition, we present the definition of signed edge-vertex incident matrix as follows:
\begin{definition}[Signed edge-vertex incident matrix]
    Let $G=(V_L, V_R, E)$ be a connected directed bipartite graph where all edges orientate from $V_L$ to $V_R$. The \emph{signed edge-vertex incidence matrix} is denoted as follows
    \begin{align*}
        B(e,v)
        = \begin{cases}
        +1, & \mathrm{~if~} v \mathrm{~incident~to~} e \mathrm{~and~} v \in V_L; \\
        -1, & \mathrm{~if~} v \mathrm{~incident~to~} e \mathrm{~and~} v \in V_R; \\
        0, & \mathrm{~otherwise}.
        \end{cases}
    \end{align*}
\end{definition}

Then, we provide the definition of SDDM matrix and SDD matrix.
\begin{definition}[SDDM, SDD matrix]\label{def:sddm_sdd}
    A square matrix $A$ is \emph{weakly diagonally dominant} if $A_{i,i} \ge \sum_{j \neq i} \left| A_{i,j} \right|$ for $i$,
    and is \emph{strictly diagonally dominant} if $A_{i,i} > \sum_{j \neq i} \left| A_{i,j} \right|$ for $i$.
    A matrix $A$ is SDD$_0$ if it is symmetric and weakly diagonally dominant, and is SDD if it is symmetric and strictly diagonally dominant. 
    A matrix $A$ is SDDM$_0$ if it is SDD$_0$ and $A_{i,j} \le 0$ for all $i \neq j$.
    An SDDM$_0$ matrix is SDDM if it is strictly diagonally dominant.
\end{definition}
Next, we introduce a fact about SDDM$_0$ matrix.
\begin{fact} \label{fct:sddm_psd}
    An SDDM$_0$ matrix must be positive semi-definite.
    If an SDDM$_0$ matrix has zero row-sums, then it is a Laplacian matrix. 
    If an irreducible SDDM$_0$ matrix has at least one positive row-sum, then it is positive definite.
\end{fact}

{\bf Bit complexity.}
Given a linear programming 
\begin{align}\label{eq:lemma_41_ls14_before}
    & ~ \min_{x\in\R^n} c^{\top} x\\
    \text{s.t.~} & ~ Ax\geq b \notag,
\end{align}

where $A\in \Z^{m\times n}$, $b\in \Z^m$, $c\in \Z^n$ all having integer coefficient.

The bit complexity $L$ is defined as
\begin{align*}
L:= \log(m) + \log(1+\dmax(A)) + \log(1 + \max\{ \|c\|_{\infty}, \|b\|_{\infty}\}),
\end{align*}
where $\dmax(A)$ denotes the largest absolute value of the determinant of a square sub-matrix of $A$.

It is well known that $\poly(L)$-bit precision is sufficient to implement an IPM (e.g., see \cite{ds08} and the references therein). 
This is because the absolute values of all intermediate arithmetic results are within $[2^{-\poly(L)}, 2^{\poly(L)}]$, and the errors in all the approximations are at least $1 / \poly(n)$. 
Therefore, truncating all the arithmetic results to $\poly(L)$ bits for some sufficiently large polynomial preserves all the error parameters and thus the same analysis holds.

We will need the following tools in our later analysis.

\begin{lemma}[\cite{ht56}]\label{lem:totally_unimodular}
    Let $A \in \R^{m \times n}$ be the unsigned edge-vertex incident matrix of a bipartite graph $G = (V_L, V_R, E)$. 
    Let $S := \{A x \le b \mid x \ge \mathbf{0}_n\}$, $S' := \{A^{\top} y \le b' \mid y \ge \mathbf{0}_m\}$, where $b \in \Z^m$, $b' \in \Z^n$. 
    Then both $A$ and $A^{\top}$ are totally unimodular, i.e., all square submatrices of them have determinants of $0,1,-1$, and all extreme points of $S$ and $S'$ are integral.
\end{lemma}

\section{Preliminary for IPM}\label{sec:ipm_preliminary}

Since IPM was proposed by Karmarkar \cite{k84} in 1984, it becomes a very popular method for analyzing the running time of linear programming or linear programming type problems \cite{v87,r88,v89_lp,nn89,nn92,nn94,a00,r01,ds08,ls13_path2,ls14,ls15,cls19,lsz19,blss20,jklps20,sy21,jswz21,y21,dly21,hjst21}. 
In this section, we will focus on introducing some mathematical background of IPM for LP. Let us consider the linear programming
\begin{align*}
\min_{x \in \R^n} c^{\top} x, & \mathrm{~~s.t.} ~Ax\geq b
\end{align*}
where $A\in \R^{m\times n}$, $b\in \R^m$, $c\in \R^n$. In the rest of this section, all the notations and discussions will be based on the above LP formulation. For convenient, we denote $a_1^\top,\cdots,a_m^\top$ as the row vectors of matrix $A \in \R^{m\times n}$. 
We first introduce some related definitions of IPM in Section \ref{sec:ipm_preliminary:def} and introduce the definitions of barrier functions in Section \ref{sec:barrier_preliminary:def}. Then, we give some tools for volumetric barrier function in Section \ref{sec:tools_volumetric_barrier}. In Section \ref{sec:lee_barrier}, we give some definitions for Lee-Sidford barrier and the computation of Lewis weight. In Section \ref{sec:ipm_preliminary:tools}, we provide some approximation tools.

\subsection{Definitions}\label{sec:ipm_preliminary:def}

We define feasible solution as follows:
\begin{definition}[Feasible solution]\label{def:feasible_solution}
For any $x\in \R^n$, we say $x$ is feasible if for all $i \in[m]$, $a_i^{\top} x > b_i$.
\end{definition}

We define slack variables:
\begin{definition}[Slack]\label{def:slack}
We define the \emph{slack} $s_i(x) := a_i^{\top} x - b_i \in \R$ for all $x \in \R^{n}$ and $i \in [m]$. 
\end{definition}

Let $F(x):\R^n \rightarrow \R$ be some barrier function. 
We define the perturbed function as follows:
\begin{definition}[Perturbed objective function]\label{def:perturbed_objective_function}
Define the \emph{perturbed objective function} $f_{t } : \R^n \rightarrow \R$, where $t > 0$ is a parameter:
\begin{align} \label{eq:perturbed_obj}
    f_{t}(x) := t \cdot c^{\top} x + F(x).
\end{align}
\end{definition}

We define the gradient and Hessian matrix with respect to the barrier function.

\begin{definition}[Gradient and Hessian]\label{def:gradient_hessian}
We define the gradient $g(x):\R^n \rightarrow \R^n$ and Hessian $H(x): \R^n \rightarrow \R^{n\times n}$ as follows:
\begin{align}\label{eq:def_gd_hes}
    g(x) := & ~\nabla F(x) \in \R^{n} \\
    H(x) := & ~ \nabla^2 F(x) \in \R^{n\times n}.
\end{align}
\end{definition}

We define our potential function $\Phi$.

\begin{definition}[Potential $\Phi$ for perturbed objective]\label{def:phi}
Given $t>0$ and feasible $x,y\in \R$, we define function $\Phi_t : \R^n  \times \R^n \rightarrow \R$:
\begin{align*}
\Phi_{t}(x,y) := \| \nabla f_t(x) \|_{H(y)^{-1}}.
\end{align*}
\end{definition}

In addition, we also need to define a potential function $\Psi$.
\begin{definition}[Potential $\Psi$ for barrier]\label{def:psi}
Given feasible $x, y \in \R$, we define function $\Psi : \R^n  \times \R^n \rightarrow \R$:
\begin{align*}
    \Psi(x,y) := \| g(x) \|_{H(y)^{-1}} .
\end{align*}
\end{definition}

\subsection{Barrier functions}
\label{sec:barrier_preliminary:def}
We first give the definition of logarithmic barrier function:
\begin{definition}[Logarithmic barrier function]\label{def:log_barrier_function}
Define the \emph{logarithmic barrier function} $\phi(x):\R^n \rightarrow \R$ as follows:
\begin{align}\label{eq:def_log_barrier}
    \phi(x) := -\sum_{i \in [m]} \ln(a_i^{\top} x - b_i),
\end{align}
Let $s_i(x) = a_i^\top x - b_i$ for each $i \in [m]$.

Thus 
\begin{align*}
    \nabla \phi(x)  & ~ = - \sum_{i \in [m]} \frac{a_i}{s_i(x)} \in \R^{n}; \\
    \nabla^2 \phi(x) & ~  = \sum_{i \in [m]} \frac{a_i a_i^{\top}}{s_i(x)^2} \in \R^{n\times n}.
\end{align*}
\end{definition}

Next, we provide the definition of volumetric barrier function.
\begin{definition}[Volumetric barrier function]\label{def:vol_barrier_function}
Define the \emph{Volumetric barrier function} $V(x):\R^n \rightarrow \R$ as follows:
\begin{align*}
    V(x) := \frac{1}{2}\log(\det(\nabla^2 \phi(x))).
\end{align*}
And thus
\begin{align*}
    \nabla V(x) = -\sum_{i=1}^m a_i \cdot \frac{\sigma_i(x)}{s_i(x)} \in \R^n,
\end{align*}
where $\sigma(x)$ is defined in Def.~\ref{def:sigma}.
\end{definition}
Then, we provide the definition of hybrid barrier function.
\begin{definition}[Hybrid barrier function]\label{def:hy_barrier_function}
Let $\rho > 0 $ be a fixed parameter. Define the \emph{Hybrid barrier function} $V_{\rho}(x):\R^n \rightarrow \R$ as follows:
\begin{align}\label{eq:def_hy_barrier}
    V_{\rho}(x) := V(x) + \rho \cdot \phi(x).
\end{align}
\end{definition}

The following is the definition of $\theta$-self-concordance barrier.

\begin{definition}[$\theta$-self-concordance barrier]\label{def:theta_self_concordance}
Let $E$ be a finite-dimensional real vector space and $Q$ be an open non-empty convex subset of $E$. A function $F: Q \rightarrow \R$ is called a self-concordant barrier if it is three times differentiable, strictly convex and satisfies the conditions
\begin{align*}
    |\nabla^3 F(x)[h,h,h]| \leq & ~ 2 \cdot ( \nabla^2 F(x)[h,h] )^{3/2} \\
    F(x) \rightarrow & ~ \infty \text{~as~} x \rightarrow \partial Q, \text{~and~} \\
    \nabla^2 F(x) \succeq & ~ \frac{1}{\theta}\nabla F(x) \nabla F(x)^{\top},
\end{align*}
for all $h \in E$, $x \in Q$.
\end{definition}

The following theorem shows that a positive multiple of the hybrid barrier function is a $\theta^2$-self-concordant barrier with $\theta=O((mn)^{1/4})$.

\begin{theorem}[\cite{av93,a00}]\label{thm:v_rho_self_concor}
The barrier function $V_{\rho}$ with $\rho = n/m$ is, after multiplication by a positive factor $O(\sqrt{m/n})$, a $\theta^2$-self-concordance barrier function where $\theta= O( (m n)^{1/4} )$.
\end{theorem}

\subsection{Tools for volumetric barrier function}
\label{sec:tools_volumetric_barrier}

The leverage score of volumetric barrier function is defined as the following:
\begin{definition}[Leverage score function]\label{def:sigma}
We define $\sigma_i(x)$ as follows
\begin{align}\label{eq:def_sigma}
    \sigma_i(x) := \frac{ a_i^{\top} (\nabla^2 \phi(x) )^{-1} a_i }{ (a_i^{\top} x - b_i)^2 },~~ \forall i \in [ m].
\end{align}
\end{definition}
In addition, we give the definition of matrix $Q$.

\begin{definition}[Matrix $Q$]\label{def:Q}
We define matrix $Q(x) \in \R^{n \times n}$ as follows
\begin{align*}
    Q(x) := \sum_{i=1}^m \sigma_i(x) \frac{ a_i a_i^\top }{ ( a_i^\top x - b_i )^2 }.
\end{align*}
\end{definition}

Then we will have 
\begin{lemma}[$Q$ is a constant spectral approximation to Hessian \cite{av93}]\label{lem:Q_approximate}
\begin{align*}
Q(x) \preceq \nabla^2 V(x) \preceq 5 Q(x).
\end{align*}
\end{lemma}

\subsection{Lee-Sidford barrier}
\label{sec:lee_barrier}

We first provide the definition of leverage score for Lee-Sidford barrier function.
\begin{definition}[Leverage score]
For a non-degenerate  
matrix $A\in \R^{m\times n}$, we define $\sigma(A)\in \R^m$ be the leverage score of $A$, e.g., $\sigma(A)_i := a_i^{\top}(A^{\top}A)^{-1}a_i$.
\end{definition}

In the following, we write $A_x := S(x)^{-1}A$, where $S(x) = \diag (s(x)) $ given $s(x)$ is the slack variable defined in Definition \ref{def:slack}.

\begin{definition}[Lee-Sidford barrier]\label{def:ls_barrier}
Let $q = \Theta(\log m)$. We define the function 
\begin{align*}
    f(x,w) := \ln \det ( A_x^{\top} W^{ 1 - 2 / q } A_x ) - ( 1 - 2 / q ) \tr[W].
\end{align*}
The Lee-Sidford barrier is defined as
\begin{align*}
\psi(x) = \max_{w\in \R^m} \frac{1}{2}f(x,w).
\end{align*}
And $w_x := \arg\max_{w\in \R^m}f(x,w)$ is Lewis weight.
\end{definition}

Next, we introduce the definition of the gradient and the Hessian of the Lee-Sidford barrier function.
\begin{lemma}[Lemma 31 of \cite{ls19}]\label{lem:gra_hess_ls_barrier}
The gradient of $\psi(x)$ is 
\begin{align*}
    \nabla \psi(x) = - A_x^{\top} w_x \in \R^n.
\end{align*}
The Hessian of $\psi(x)$ is 
\begin{align*}
    \nabla^2 \psi(x) = A_x^{\top}W_x^{1/2}(I + N_x)W_x^{1/2}A_x \in \R^{n\times n}.
\end{align*}
Further, $N_x$ is a symmetric matrix with $0\preceq N_x \preceq qI$ and therefore
\begin{align*}
    A_x^{\top}W_x A_x \preceq \nabla^2 \psi(x) \preceq (1 + q)A_x^{\top}W_x A_x.
\end{align*}
\end{lemma}

Given $\alpha >0$ and $w\in \R^m$, here we define $\rho(w)\in \R^m$ to be
\begin{align}\label{eq:rho_w}
\rho_i(w):=\frac{\sigma_i(w)}{w_i^{1+\alpha}}.
\end{align}
and define $\sigma(w)\in\R^m$ to be 
\begin{align}\label{eq:sigma_w}
\sigma_i(w):=\sigma_i(W^{1/2}A).
\end{align}

\begin{algorithm}[!ht]
\caption{Lewis Weight Computation}
\label{alg:compute_lewis_weight}
\begin{algorithmic}[1]
\Procedure{$\textsc{ComputeLewisWeight}$}{$A \in \R^{m\times n},p \in \R,\epsilon > 0$}
\State Initialize $w_i^{(0)} = \frac{n}{m}$, for all $i\in[m]$
\State Let $\alpha = \frac{2}{p-2}$, $\ov{\alpha}=\max(\alpha,1)$, $\wt{\epsilon} = \frac{\alpha^8\epsilon^4}{(25m(\sqrt{n}+\alpha)(\alpha + \alpha^{-1}))^{4}}$
\State $T = O(\max(\alpha^{-1},\alpha)\log(m/\wt{\epsilon}))$
\For{ $k = 1,2,\cdots,T$}
    \State $\wt{w}^{(k)} = \textsc{Round}(w^{(k-1)},A,\alpha)$ \Comment{Algorithm~\ref{alg:round}}
    \State $w^{(k)} = \textsc{Descent}(\wt{w}^{(k)},[m],\frac{1}{3\ov{\alpha}}\cdot \mathbf{1})$ \Comment{Algorithm~\ref{alg:descent}}
\EndFor
\State $w_R = \textsc{Round}(w^{(T)},A,\alpha)$
\State \Return $\wh{w}=\diag(A(A^{\top}W_RA)^{-1}A^{\top})^{1/\alpha}$
\EndProcedure
\end{algorithmic}
\end{algorithm}

\begin{algorithm}[!ht]
\caption{Subroutine: Descent}
\label{alg:descent}
\begin{algorithmic}[1]
\Procedure{$\textsc{Descent}$}{$w ,C\subseteq [m],\eta$}
\State Let $w'_i = w_i (1+\eta_i\cdot \frac{\rho_i(w)-1}{\rho_i(w)+1})$, for all $i\in C$, and $w'_i=w_i$, for all $i\in [m]\setminus C$
\State \Return $w'$
\EndProcedure
\end{algorithmic}
\end{algorithm}

\begin{algorithm}[!ht]
\caption{Subroutine: Round}
\label{alg:round}
\begin{algorithmic}[1]
\Procedure{$\textsc{Round}$}{$w \in \R^{m},A,\alpha$}
\State Let $\ov{\alpha}=\max(\alpha,1)$
\While{$C= \{i \mid \rho_i(w)> 1+\alpha\}$ is not empty}
    \State $w\leftarrow \textsc{Descent}(w,C,\frac{1}{3\ov{\alpha}}\cdot \mathbf{1})$ \Comment{All rows of $C$ are updated at once}
\EndWhile
\State \Return $w$
\EndProcedure
\end{algorithmic}
\end{algorithm}

We will utilize an important result due to~\cite{flps21}, in which they show $\epsilon$-approximate $\ell_p$ Lewis weights can be computed in $O(p^3\log(pm/\epsilon))$ iterations for $p\geq 4$, where each iteration computes the leverage scores of all rows for one weight vector. For completeness, we include their parallel algorithm here (Algorithms~\ref{alg:compute_lewis_weight}, \ref{alg:descent} and~\ref{alg:round}). The same paper has a sequential variant, which updates the weights one row at a time and needs the leverage score of each row for the weights changed by all earlier updates. Our small space implementation needs every step to use the leverage scores of one common weight vector, so we use the parallel algorithm. Later, we will present our small space implementation in Appendix~\ref{sec:additional_algorithms}.

\begin{lemma}[Theorem 1 of \cite{flps21}]\label{lem:flps}
Given a full-rank matrix $A\in \R^{m\times n}$, $p\geq 4$ and $\epsilon\in (0,1)$, let $\alpha=\frac{2}{p-2}$ and let $\ov{w}\in \R^m$ be the minimizer over all $w\in \R^m$ with positive entries of
\begin{align}\label{eq:lewis_objective}
-\log\det(A^{\top}WA)+\frac{1}{1+\alpha}\sum_{i\in [m]}w_i^{1+\alpha}.
\end{align}
Then Algorithm~\ref{alg:compute_lewis_weight} outputs a vector $\wh{w}\in \R^m$ with $(1-\epsilon)\wh{w}_i\leq \ov{w}_i\leq (1+\epsilon)\wh{w}_i$ for all $i\in [m]$. It makes $O(p^3\log(pm/\epsilon))$ calls to \textsc{Descent} in total, counting the calls inside \textsc{Round}. Each call computes the leverage scores $\sigma(w)$ for one weight vector $w$.
\end{lemma}
The vector $\ov{w}$ gives the $\ell_p$ Lewis weights of $A$ after an entrywise power. These weights are the unique vector $u\in \R^m$ with $u_i^{2/p}=a_i^{\top}(A^{\top}U^{1-2/p}A)^{-1}a_i$ for all $i\in [m]$ (Definition 1 of~\cite{flps21}). The vector $u=\ov{w}^{1+\alpha}$ satisfies this equation, since the minimizer satisfies $\ov{w}_i^{\alpha}=a_i^{\top}(A^{\top}\ov{W}A)^{-1}a_i$ for all $i\in [m]$ and $(1+\alpha)(1-2/p)=1$. In particular, the function $f(x,w)$ of Definition~\ref{def:ls_barrier} is the negative of the function in Eq.~\eqref{eq:lewis_objective} for the matrix $A_x$, $p=q$ and the weight vector $w^{1-2/q}$, because $1-2/q=\frac{1}{1+\alpha}$. So $w_x=\ov{w}^{1+\alpha}$, and we approximate $w_x$ by the entrywise power $\wh{w}^{1+\alpha}$ of the output of Algorithm~\ref{alg:compute_lewis_weight}.

\subsection{Approximation tools: near-optimal solution to dual and perturbation to Hessian}\label{sec:ipm_preliminary:tools}

We list two types of tools from literature. The first one is about the property of near-optimal solution. We remark that those ideas was firstly proposed in \cite{r88}, we cite the more cleaner statement from Renegar's book \cite{r01}.  The second one is about perturbation to the Hessian matrix.

\begin{lemma}[Nearly-optimal output: value version~\cite{r01}]\label{lem:optimal_x_value}
Given linear program 
\[
\min_{x \in \R^n, Ax \geq b} c^{\top} x,
\]
where $A\in \R^{m\times n}$, $b\in \R^m$, $c\in \R^n$.
Let $x^*\in \R^n$ be the optimal solution of the above LP. Let $F$ be a $\theta^2$-self-concordant barrier (Definition~\ref{def:theta_self_concordance}) with $\theta\geq 2$, and let $\Phi$ be defined with respect to $F$. Let $\epsilon_{N} \in (0, 1/10)$. If for some $t>0$ and $x\in \R^n$ we have $\Phi_t(x,x) \leq \epsilon_{N}$ (Definition~\ref{def:phi}), then we have
\begin{align*}
    c^{\top} x - c^{\top} x^* \leq \frac{\theta^2}{t}\cdot (1 + 2\epsilon_{N}).
\end{align*}
\end{lemma}

In addition, we define the Hessian approximation as follows:
\begin{lemma}[Hessian approximation, \cite{r01}]\label{lem:hessian_aprroximation}
Let $f$ be a self-concordant function with domain $D_f$. Define $H(x):=\nabla^2 f(x)$. For any two feasible point $x,z\in D_f$, if $\|x-z\|_{H(x)} < 1$, then we have
\[
(1-\|x-z\|_{H(x)})^2 \cdot H(z) \preceq H(x) \preceq (1-\|x-z\|_{H(x)})^{-2}\cdot H(z).
\]
\end{lemma}

\section{Algorithm}\label{sec:ipm_algorithm}

In this section, we present our robust dual central path. 
Later in Section~\ref{sec:ipm_error_analysis}, we will give the error analysis of the output of this algorithm. 

In Section~\ref{sec:lp_in_small_space}, we bound the number of pass of executing this algorithm in the streaming model.
Specifically, we implement and analyze streaming SDD solver in Section~\ref{sec:sdd_solver_streaming_model} which is used in Line~\ref{line:wt_delta_x}
of this algorithm.

Here, we give a brief overview of our implementation of robust dual central path (Algorithm \ref{alg:ipm}). First, the algorithm takes the input matrix $A \in \R^{m \times n}$, vector $b \in \R^m$, vector $c \in \R^n$, and the parameters $t_{\start},x_{\start},t_{\final},o,\epsilon_{\Phi}$ as input. For $t_{\start}$ and $x_{\start}$, they must satisfy $\Phi_{t_{\start}}(x_{\start},x_{\start}) \leq \epsilon_{\Phi}$ to make the algorithm get a nearly optimal solution. For each iteration, we compute $\wt{H}(x)$ as an approximation of the Hessian and compute $\wt{\nabla}f_t(x)$ as an approximation of the gradient. Specially, $\wt{H}(x)$ satisfies
\begin{align*}
\gamma \wt{H}(x) \preceq H(x) \preceq \wt{H}(x)
\end{align*}
and $\wt{\nabla}f_t(x)$ satisfies
\begin{align*}
    \| \wt{\nabla}f_t(x) - \nabla f_t(x) \|_{H(x)^{-1}} \leq \epsilon_g.
\end{align*}
In addition, we use SDD solver, that implemented in Section~\ref{sec:sdd_solver_streaming_model}, to approximate the value of $\delta'= -\wt{H}(x)^{-1}\cdot \wt{\nabla}f_t(x)$ by $\wt{\delta}_x$ and $\wt{\delta}_x$ satisfies 
\begin{align*}
    \|\wt{\delta}_x - \delta_x'\|_{H(x)} \leq \epsilon'.
\end{align*}
Then, we update $x$ by $x^{new} \gets x + \wt{\delta}_x$.

\begin{algorithm}[ht]
\caption{Robust dual central path}
\label{alg:ipm}
\begin{algorithmic}[1]
\Procedure{\textsc{InteriorPointMethod}}{$A,b,c,t_{\start},x_{\start},t_{\final},o,\epsilon_{\Phi}$} \Comment{Lemma~\ref{lem:ipm_error}}
\State $m,n\leftarrow $ dimensions of $A$
\State \Comment{$A \in \R^{m \times n}$ is the input matrix, $b \in \R^{m}$, $c\in \R^n$}
\State \Comment{$t_{\start},x_{\start}$ satisfy $\Phi_{t_{\start}}(x_{\start},x_{\start}) \leq \epsilon_{\Phi}$}
\State \Comment{$t_{\final} \in \R$ is the final goal of $t$}
\State \Comment{$o\in \{0,1\}$. If $o$ is $0$, the algorithm decreases $t$, otherwise the algorithm increases $t$}
\State \Comment{$\epsilon_{\Phi}$ is the accuracy of the potential in Assumption~\ref{ass:input_of_ipm}}
\State Let $\epsilon_g,\epsilon',\epsilon_t,\gamma,\theta$ be parameters that, together with $\epsilon_{\Phi}$, meet Assumption~\ref{ass:input_of_ipm} 
\State $t\leftarrow t_{\start}$
\State $x\leftarrow x_{\start}$
\State $T \leftarrow O(\epsilon_t^{-1}) \cdot |\log(t_{\final} / t_{\start})| $ \Comment{Number of iterations}
\For {$k \leftarrow 1$ to $T$}  \label{line:IPM_iteartion_start}
    \State Let $\wt{H}(x)$ be any PSD matrix that $\gamma \wt{H}(x) \preceq H(x) \preceq \wt{H}(x)$ 
    \State Let $\delta_x := -\wt{H}(x)^{-1}\cdot \nabla f_t(x)$ \label{line:delta_x}
    \State Let $\wt{\nabla} f_t(x)$ be that $\| \wt{\nabla}f_t(x) - \nabla f_t(x) \|_{H(x)^{-1}} \leq \epsilon_g$ \label{line:wt_nabla_f}
    \State Let $\delta_x' := -\wt{H}(x)^{-1}\cdot \wt{\nabla}f_t(x)$ \label{line:delta_x_2}
    \State Let $\wt{\delta}_x$ be any vector that $\|\wt{\delta}_x - \delta_x'\|_{H(x)} \leq \epsilon'$  \label{line:wt_delta_x}
    \State $x^{\new} \leftarrow x + \wt{\delta}_x$  \label{line:incrase_x_by_newton}
    \If{$o=0$}
        \State $t^{\new} \leftarrow t \cdot ( 1 - \epsilon_t )$
    \Else
        \State $t^{\new} \leftarrow t \cdot ( 1 + \epsilon_t )$
    \EndIf
    \If { ($o=0$ \text{and} $t < t_{\final}$) \text{or} ($o=1$ \text{and} $t>t_{\final}$)}
        \State \textbf{break}
    \EndIf
    \State $t \leftarrow t^{\new}$
    \State $x \leftarrow x^{\new}$
\EndFor \label{line:IPM_iteartion_end}
\State \Return $x$ 
\EndProcedure
\end{algorithmic}
\end{algorithm}

\newpage
\section{Error analysis of IPM}\label{sec:ipm_error_analysis}

In this section, we provide an error analysis for our IPM. The main goal is to prove Lemma~\ref{lem:ipm_error}.

\begin{lemma}\label{lem:ipm_error}
Given any feasible linear program
\begin{align*}
\min_{x\in \R^n, Ax\geq b} c^{\top} x,
\end{align*}
where $A\in \R^{m\times n}$, $b\in \R^{m}$, and $c\in \R^n$. 
Suppose the solution exists and let $x^*\in \R^n$ be the solution.
If $\textsc{InteriorPointMethod}$ (Algorithm~\ref{alg:ipm}) is given $A,b,c,t_{\start},x_{\start},t_{\final}$ that satisfy $\Phi_{t_{\start}}(x_{\start},x_{\start})\leq \epsilon_{\Phi}$, and suppose Assumption~\ref{ass:input_of_ipm} holds, then it outputs an $x$ which is a \emph{nearly-optimal solution}:
\begin{align*}
    c^{\top} x - c^{\top} x^* \leq \frac{\theta^2}{t_{\final}}\cdot (1 + 2\epsilon_{\Phi}).
\end{align*}
\end{lemma}
\begin{proof}
Since our initial point $x_{\start}$ and $t_{\start}$ satisfy 
\begin{align*} 
\Phi_{t_{\start}}(x_{\start},x_{\start})\leq \epsilon_{\Phi},
\end{align*}
as $T=O(\epsilon_t^{-1})\cdot |\log(t_{\final} / t_{\start})|$, we can apply Lemma~\ref{lem:bound_phi} to get 
\begin{align*}
\Phi_{t_{\final}}(x,x) \leq \epsilon_{\Phi}.
\end{align*}

Applying Lemma~\ref{lem:optimal_x_value} on $t_{\final}$ and $x$, which is allowed since $F$ is $\theta^2$-self-concordant (Part~\ref{part:theta} of Assumption~\ref{ass:input_of_ipm}) and $\epsilon_{\Phi}<1/10$ (Part~\ref{part:epsilon_x_phi}), we get
\begin{align*}
    c^{\top} x - c^{\top} x^* \leq \frac{\theta^2}{t_{\final}}\cdot (1 + 2\epsilon_{\Phi}).
\end{align*}
\end{proof}

The rest of this section is organized as follows: 
In Section~\ref{sec:ipm_error_analysis:parameters}, we state the choices of our parameters. We bound the potential function in Section~\ref{sec:bound_phi}. We bound the changes of the potential function when $t$ is moving in Section~\ref{sec:t_move}. We bound the total movement of $x$ in Section~\ref{sec:phi_t_move_both}. The detailed proof of $x$ movement can be splitted into three parts: Section~\ref{sec:phi_t_move_both_1}, Section~\ref{sec:phi_t_move_both_2} and Section~\ref{sec:newton_appro_hessian}. 

\subsection{Assumptions on parameters}\label{sec:ipm_error_analysis:parameters}

We state several assumptions here.
\begin{assumption}\label{ass:input_of_ipm}
We state six assumptions here.
\begin{enumerate}
    \item Let $\gamma \in (0,1]$. In each iteration, $\wt{H}(x)$ satisfies
    \begin{align*}
        \gamma \wt{H}(x) \preceq H(x) \preceq \wt{H}(x);
    \end{align*} \label{part:H(x)}
    \item Let $\epsilon_g \in (0,1/10)$. In each iteration, $\wt{\nabla}f_t(x)$ satisfies
    \begin{align*}
        \| \wt{\nabla}f_t(x) - \nabla f_t(x) \|_{H(x)^{-1}} \leq \epsilon_g;
    \end{align*} \label{part:wt_nabla_f}
    \item Let $\epsilon' \in (0,1/10)$. In each iteration, $\wt{\delta}_x$ satisfies
    \begin{align*}
        \|\wt{\delta}_x - \delta_x'\|_{H(x)} \leq \epsilon';
    \end{align*} \label{part:delta_x_app}
    \item $F(x)$ is $\theta^2$-self-concordant barrier function; \label{part:theta}
    \item Let $\epsilon_x := \epsilon_g + \epsilon'$, $\epsilon_t,\epsilon_x,\epsilon_{\Phi},\gamma,\theta$ satisfies the following inequality:
    \begin{align*}
        (1 + \epsilon_t) \cdot \big(\frac{1-\gamma}{1-(\epsilon_x + \epsilon_{\Phi})} \epsilon_{\Phi} + 4 \epsilon_{\Phi}^2 + 1.25  \epsilon_x)\big) + \epsilon_t\cdot \theta \leq \epsilon_{\Phi},
    \end{align*} \label{part:equation}
    \item $\epsilon_x + \epsilon_{\Phi} \leq 1/10$. \label{part:epsilon_x_phi}
\end{enumerate}
\end{assumption}

\subsection{Bounding potential function  \texorpdfstring{$\Phi$}{}}\label{sec:bound_phi}
The goal of this section is to prove Lemma~\ref{lem:bound_phi}.

\begin{lemma}\label{lem:bound_phi}
For each $t$, let $\Phi_t : \R^n \times \R^n \rightarrow \R$ be defined as Definition~\ref{def:phi}. 
Suppose Assumption~\ref{ass:input_of_ipm} holds. If the input $x_{\start},t_{\start}$ satisfies 
\begin{align*}
\Phi_{t_{\start}}(x_{\start},x_{\start}) \leq \epsilon_{\Phi},
\end{align*}
then
for all iteration $k\in[T]$, we have 
\begin{align*}
\Phi_{t^{(k)}}(x^{(k)},x^{(k)}) \leq \epsilon_{\Phi}.
\end{align*}
\end{lemma}
\begin{proof}
We prove by induction on iteration $k$. In the base case where $k=0$, we have $t^{(0)}:=t_{\start}$ and $x^{(0)}:=x_{\start}$ so that the condition holds by assumption.

When $k\geq 1$, for the ease of notation, we define $x= x^{(k-1)}$, $t= t^{(k-1)}$, $x^{\new} = x^{(k)}$, $t^{\new} = t^{(k)}$ and let $\delta_x$ and $\wt{\delta}_x$ be defined in Line~\ref{line:delta_x} and \ref{line:wt_delta_x} of Algorithm~\ref{alg:ipm}.

First, from the induction hypothesis, we have $\Phi_t(x,x)\leq \epsilon_{\Phi}$. Then,
\begin{align*}
\|\delta_x\|_{H(x)} = \|-\wt{H}(x)^{-1}  \cdot \nabla f_t(x)\|_{H(x)} \leq \|-H(x)^{-1}  \cdot \nabla f_t(x)\|_{H(x)} \leq \epsilon_{\Phi}.
\end{align*}
where the second step is by $\wt{H}(x) \succeq H(x)$ (Part~\ref{part:H(x)} of Assumption~\ref{ass:input_of_ipm})
the last step is by the definition of $\Phi_t(x,x)$ (Definition~\ref{def:phi}).

We bound $\|\delta_x - \wt{\delta}_x\|_{H(x)}$ as follows:
\begin{align*}
    \|\delta_x - \wt{\delta}_x\|_{H(x)} 
    \leq & ~ \|\delta_x - \delta_x'\|_{H(x)} + \|\delta_x' - \wt{\delta}_x\|_{H(x)} \\
    \leq & ~ \|\wt{H}(x)^{-1}\cdot(\nabla f_t(x) - \wt{\nabla}f_t(x))\|_{H(x)} +\epsilon'\\
    \leq & ~ \|H(x)^{-1}\cdot (\nabla f_t(x) - \wt{\nabla}f_t(x))\|_{H(x)} + \epsilon'\\
    \leq & ~ \|\nabla f_t(x) - \wt{\nabla}f_t(x)\|_{H(x)^{-1}} + \epsilon'\\
    \leq & ~ \epsilon_g + \epsilon',
\end{align*}
where the first step is by triangle inequality, the second step is by Part~\ref{part:delta_x_app} of Assumption~\ref{ass:input_of_ipm}, the third step is by $H(x) \preceq \wt{H}(x)$ (Part~\ref{part:H(x)} of Assumption~\ref{ass:input_of_ipm}), the last step is by Part~\ref{part:wt_nabla_f} of Assumption~\ref{ass:input_of_ipm}. 

Let $\epsilon_x := \epsilon_g + \epsilon'$, finally, we have
\begin{align*}
\Phi_{t^{\new}} (x^{\new}, x^{\new}) 
\leq & ~ \frac{t^{\new}}{t}\cdot \Phi_t(x^{\new}, x^{\new})  + |t^{\new} / t - 1| \cdot \Psi(x^{\new}, x^{\new})\\
\leq & ~ \frac{t^{\new}}{t}\cdot \Phi_t(x^{\new}, x^{\new})  + |t^{\new} / t - 1| \cdot \theta \\
\leq & ~ (1 + \epsilon_t) \cdot \Phi_t(x^{\new}, x^{\new})  + \epsilon_t\cdot \theta\\
\leq & ~ (1 + \epsilon_t) \cdot \big(\frac{1-\gamma}{1-(\epsilon_x + \epsilon_{\Phi})} \Phi_t(x, x) + 4 \Phi_t(x,x)^2 + 1.25  \epsilon_x\big) + \epsilon_t\cdot \theta\\
\leq & ~ (1 + \epsilon_t) \cdot \big(\frac{1-\gamma}{1-(\epsilon_x + \epsilon_{\Phi})} \epsilon_{\Phi} + 4 \epsilon_{\Phi}^2 + 1.25  \epsilon_x\big) + \epsilon_t\cdot \theta\\
\leq & ~ \epsilon_{\Phi},
\end{align*}
where the first step is by Lemma~\ref{lem:t_move}, the second step is by Part~\ref{part:theta} of Assumption~\ref{ass:input_of_ipm} and Fact~\ref{fac:bound_psi}, the fourth step is by Lemma~\ref{lem:phi_t_move_both} since $\epsilon_x + \epsilon_{\Phi} \leq 1/10$ (by Part~\ref{part:epsilon_x_phi} of Assumption~\ref{ass:input_of_ipm}) and 
$\| \wt{\delta}_x - \delta_x \|_{H(x)}\leq \epsilon_{x}$, the fifth step is by the induction hypothesis, and the last step is by Part~\ref{part:equation} of Assumption~\ref{ass:input_of_ipm}.
\end{proof}

\begin{table}[!ht]
\small
    \centering
    \begin{tabular}{|l|l|l|l|l|}
    \hline
         {\bf Lemma} & {\bf Section} & LHS & RHS \\ \hline
         Lemma~\ref{lem:bound_phi} & \ref{sec:bound_phi} & $\Phi_{t^{\new}}(x+\wt{\delta}_x, x+\wt{\delta}_x)$ & $\epsilon_{\Phi}$ \\ \hline
         Lemma~\ref{lem:t_move} & \ref{sec:t_move} & $\Phi_{t^{\new}}(x+\wt{\delta}_x,x+\wt{\delta_x})$ & $(t^{\new}/t) \cdot \Phi_t(\cdot) + (1-t^{\new}/t) \cdot \Psi(\cdot)$  \\ \hline 
         Fact~\ref{fac:bound_psi} & \ref{sec:bounding_psi} & $\Psi(x+\wt{\delta}_x,x+\wt{\delta}_x)$ & $\theta$ \\ \hline 
          Lemma~\ref{lem:phi_t_move_both} & \ref{sec:phi_t_move_both} & $\Phi_t(x+\wt{\delta}_x, x+ \wt{\delta}_x)$ & $(1-\gamma)\Phi_t(x,x)$ \\ \hline
          Lemma~\ref{lem:phi_t_move_both_1} & \ref{sec:phi_t_move_both_1} & $\Phi_t(x+\wt{\delta}_x, x+ \wt{\delta}_x)$ & $\Phi_t(x+\delta_x,x+\wt{\delta}_x)$\\ \hline
          Lemma~\ref{lem:phi_t_move_both_2} & \ref{sec:phi_t_move_both_2} & $\Phi_t(x+\delta_x, x+\wt{\delta}_x)$ & $\Phi_t(x+\delta_x, x)$ \\ \hline 
          Lemma~\ref{lem:Phi_x_delta_x_x} & \ref{sec:newton_appro_hessian} & $\Phi_t(x+\delta_x, x)$ & $\Phi_t(x,x)$ \\ \hline 
    \end{tabular}
    \caption{Summary of movement of potential function.}
    \label{tab:summary_movement_potential_function}
\end{table}

\subsection{Bounding the movement of \texorpdfstring{$t$}{}}\label{sec:t_move}
The goal of this section is to prove Lemma~\ref{lem:t_move}.
\begin{lemma}\label{lem:t_move}
Let $\Phi$ and $\Psi$ be defined in Definition~\ref{def:phi} and Definition~\ref{def:psi}.
Then, for all positive $t, t^{\new} > 0$ and feasible $x\in \R^n$, we have 
\begin{align*}
\Phi_{t^{\new}} (x, x) \leq \frac{t^{\new}}{t}\cdot \Phi_t(x, x)  + \left|t^{\new} / t - 1\right| \cdot \Psi(x, x),
\end{align*}
\end{lemma}
    
\begin{proof}

We have
\begin{align*}
        \nabla f_{ t^{\new} }(x) 
        = & ~  t^{\new} \cdot c + g( x ) \notag\\
        = & ~ \frac{ t^{\new} }{ t }  \cdot (t \cdot c + g( x ) ) + \left(1 - \frac{ t^{\new} }{ t } \right)  \cdot g(x) \notag\\
        = & ~ \frac{ t^{\new} }{ t }  \cdot \nabla f_{t}( x ) + \left(1 - \frac{ t^{\new} }{ t }\right) \cdot g( x ) , 
\end{align*}
where the first step follows from the definition of $f_t(x)$ (Definition~\ref{def:perturbed_objective_function}), the second step follows from moving terms, and the last step follows from the definition of $f_t(x)$ (Definition~\ref{def:perturbed_objective_function}).

Finally, we can upper bound $\Phi_{t^{\new}} (x,x)$ as follows:
\begin{align*}
    \Phi_{t^{\new}} (x, x)
    = & ~ \| \nabla f_{ t^{\new} }(x) \|_{ H(x)^{-1} } \\
    \leq & ~ \frac{ t^{\new} }{ t } \cdot \| \nabla f_{t}( x ) \|_{ H(x)^{-1} } + \left(1 - \frac{ t^{\new} }{ t }\right) \cdot \| g(x) \|_{ H(x)^{-1} } \\
    = & ~ \frac{ t^{\new} }{ t } \cdot \Phi_t(x, x) + \left(1 - \frac{ t^{\new} }{ t }\right) \cdot \Psi(x, x) .
\end{align*}
Thus, we complete the proof.
\end{proof}

\subsection{Upper bounding the potential function}\label{sec:bounding_psi}
We state a fact regarding the relationship between potential $\psi$ and self-concordance parameter $\theta$. For more details, we refer readers to~\cite{nn94}.

\begin{fact}\label{fac:bound_psi}
Let function $\Psi : \R^n  \times \R^n \rightarrow \R$ be defined as Definition~\ref{def:psi}. 
Let $F(x)$ be $\theta^2$-self-concordance function  (see Definition~\ref{def:theta_self_concordance}). Then,
for all feasible $x \in \R^n$, we have
\begin{align*}
\Psi(x, x) \leq \theta.
\end{align*}
\end{fact}

\subsection{Move both: final}\label{sec:phi_t_move_both}

We prove that if we can compute a good enough approximation of the Newton direction, then the potential function can still be controlled. The proof will rely on Lemma~\ref{lem:phi_t_move_both_1} and~\ref{lem:phi_t_move_both_2}.
\begin{lemma}\label{lem:phi_t_move_both}
For any feasible $x \in \R^n$, let $\delta_x \in \R^n$ be defined as in Line~\ref{line:delta_x} of Algorithm~\ref{alg:ipm}, i.e.,
\begin{align*}
\delta_x := -\wt{H}(x)^{-1} \cdot \nabla f_t(x).
\end{align*}
Given $\|\wt{\delta}_x - \delta_x \|_{H(x)} \leq \epsilon_x$, $\| \delta_x \|_{H(x)} \leq \epsilon_{\Phi}$ and $\gamma \wt{H}(x) \preceq H(x) \preceq \wt{H}(x)$. If $\epsilon_x + \epsilon_{\Phi} < 1/10$, then we have
\begin{align*}
\Phi_t( x + \wt{\delta}_x, x + \wt{\delta}_x ) \leq \frac{1-\gamma}{1-(\epsilon_x + \epsilon_{\Phi})} \Phi_t(x, x) + 4 \Phi_t(x,x)^2 + 1.25  \epsilon_x.
\end{align*}
\end{lemma}

\begin{proof}

The proof is done by showing the following:
\begin{align*}
    \Phi_t(x+\wt{\delta}_x, x+\wt{\delta}_x)
    \leq & ~ \Phi_t(x+\delta_x, x+\wt{\delta}_x) + 1.25 \epsilon_x\\
    \leq & ~ (1-(\epsilon_x + \epsilon_{\Phi}))^{-1} \Phi_t(x+\delta_x, x) + 1.25\epsilon_x \\
    \leq & ~ \frac{1-\gamma}{1-(\epsilon_x + \epsilon_{\Phi})} \Phi_t(x, x) + 4 \Phi_t(x,x)^2 + 1.25  \epsilon_x ,
\end{align*}
where the first step is by Lemma~\ref{lem:phi_t_move_both_1} and $\frac{\epsilon_x}{1-\epsilon_{\Phi}-2\epsilon_x}\leq 1.25\epsilon_x$, which holds since $\epsilon_{\Phi}+2\epsilon_x\leq 2(\epsilon_x+\epsilon_{\Phi})<1/5$, the second step follows from Lemma~\ref{lem:phi_t_move_both_2}, the last step follows from the fact that $\gamma \wt{H}(x) \preceq H(x) \preceq \wt{H}(x)$ and applying this fact to Lemma~\ref{lem:Phi_x_delta_x_x}. 
\end{proof}

\subsection{Move both: part 1}\label{sec:phi_t_move_both_1}

We show that move both parameters using approximate Newton direction is close to move only one parameter using approximate direction.

\begin{lemma}\label{lem:phi_t_move_both_1}
Given $\|\wt{\delta}_x - \delta_x \|_{H(x)} \leq \epsilon_x$ and $\| \delta_x \|_{H(x)} \leq \epsilon_{\Phi}$. If $\epsilon_{\Phi} + 2\epsilon_x < 1$, then we have
\begin{align*}
    \Phi_t(x + \wt{\delta}_x, x + \wt{\delta}_x) \leq \Phi_t(x+\delta_x, x+\wt{\delta}_x) + \frac{\epsilon_x}{1-\epsilon_{\Phi} - 2\epsilon_x}.
\end{align*}
\end{lemma}
\begin{proof}

We have
\begin{align}\label{eq:phi_t_move_both_1}
    & ~ \Phi_t(x + \wt{\delta}_x, x + \wt{\delta}_x) - \Phi_t(x+\delta_x, x+\wt{\delta}_x) \notag\\
    = & ~ \| \nabla f_t(x + \wt{\delta}_x) \|_{H(x + \wt{\delta}_x)^{-1}} - \| \nabla f_t(x + \delta_x) \|_{H(x + \wt{\delta}_x)^{-1}} \notag\\
    \leq & ~ \| \nabla f_t(x + \wt{\delta}_x) - \nabla f_t(x + \delta_x)\|_{H(x + \wt{\delta}_x)^{-1}} \notag \\
    = & ~ \| g(x + \wt{\delta}_x) - g(x + \delta_x)\|_{H(x + \wt{\delta}_x)^{-1}},
\end{align}
where the first step is by the definition of $\Phi_t$ (Definition~\ref{def:phi}), and the second step is by the triangle inequality.

Define $y := x + \wt{\delta}_x$, $u := \wt{\delta}_x - \delta_x$ and $\lambda := \|u\|_{H(y)}$, so that $x+\delta_x = y - u$.

Since $\| x + \wt{\delta}_x - x\|_{H(x)} \leq \|\delta_x\|_{H(x)} + \|\wt{\delta}_x - \delta_x\|_{H(x)} \leq \epsilon_x + \epsilon_{\Phi} < 1$, by Lemma~\ref{lem:hessian_aprroximation} we have 
\begin{align}\label{eq:h_x_h_x_wt_delta_x}
(1-(\epsilon_{\Phi} + \epsilon_x))^2 H(x) \preceq H(x + \wt{\delta}_x) \preceq (1-(\epsilon_{\Phi} + \epsilon_x))^{-2} H(x).
\end{align}
The upper bound of Eq.~\eqref{eq:h_x_h_x_wt_delta_x} gives
\begin{align}\label{eq:r_move_both_1}
\lambda \leq (1-(\epsilon_{\Phi} + \epsilon_x))^{-1} \|\wt{\delta}_x - \delta_x\|_{H(x)} \leq \frac{\epsilon_x}{1-\epsilon_{\Phi} - \epsilon_x} < 1,
\end{align}
where the last step is by $\epsilon_{\Phi}+2\epsilon_x<1$.

For every $s\in[0,1]$ we have $\| y - (y - s\cdot u)\|_{H(y)} = s\cdot \lambda < 1$. So by Lemma~\ref{lem:hessian_aprroximation}, applied at the point $y$, we have 
\begin{align*}
H(y - s\cdot u) \preceq (1 - s\cdot \lambda)^{-2} H(y).
\end{align*}
Hence the symmetric matrix $M_s := H(y)^{-1/2} H(y - s\cdot u) H(y)^{-1/2}$ satisfies $0\preceq M_s \preceq (1-s\cdot \lambda)^{-2} I$, and
\begin{align}\label{eq:h_middel_u}
\| H(y - s\cdot u)\cdot u \|_{H(y)^{-1}} = \| M_s H(y)^{1/2} u\|_2 \leq (1-s\cdot \lambda)^{-2} \| H(y)^{1/2} u\|_2 = (1-s\cdot \lambda)^{-2}\cdot \lambda.
\end{align}

Finally, we have
\begin{align*}
    & ~ \| g(x + \wt{\delta}_x) - g(x + \delta_x)\|_{H(x + \wt{\delta}_x)^{-1}}\\
    = & ~ \Big\| \int_0^1 H(y - s\cdot u)\cdot u ~ \d s \Big\|_{H(y)^{-1}}\\
    \leq & ~ \int_0^1 \big\| H(y - s\cdot u)\cdot u \big\|_{H(y)^{-1}} \d s\\
    \leq & ~ \int_0^1 (1-s\cdot \lambda)^{-2}\cdot \lambda ~ \d s\\
    = & ~ \frac{\lambda}{1-\lambda}\\
    \leq & ~ \frac{\epsilon_x}{1-\epsilon_{\Phi} - 2\epsilon_x},
\end{align*}
where the first step is by $g(y) - g(y-u) = \int_0^1 H(y - s\cdot u)\cdot u ~\d s$, which is the fundamental theorem of calculus for the function $s\mapsto g(y - s\cdot u)$, the second step is by the triangle inequality, the third step is by Eq.~\eqref{eq:h_middel_u}, the fourth step computes the integral, and the last step is by Eq.~\eqref{eq:r_move_both_1} since $\lambda\mapsto \frac{\lambda}{1-\lambda}$ is increasing on $[0,1)$. Combining this with Eq.~\eqref{eq:phi_t_move_both_1} completes the proof.

\end{proof}

\subsection{Move both: part 2}\label{sec:phi_t_move_both_2}

We show that only move approximately on one parameter does not deviate too much from that parameter before moving.

\begin{lemma}\label{lem:phi_t_move_both_2}
Given $\|\wt{\delta}_x - \delta_x \|_{H(x)} \leq \epsilon_x$ and $\| \delta_x \|_{H(x)} \leq \epsilon_{\Phi}$. If $\epsilon_x + \epsilon_{\Phi} < 1$, then we have
\begin{align*}
    \Phi_t( x+\delta_x, x+\wt{\delta}_x) \leq 
     (1-(\epsilon_x + \epsilon_{\Phi}))^{-1} \cdot \Phi_t( x+\delta_x, x).
\end{align*}
\end{lemma}
\begin{proof}

Since $\| x + \wt{\delta}_x - x\|_{H(x)} \leq \|\delta_x\|_{H(x)} + \|\wt{\delta}_x - \delta_x\|_{H(x)} \leq \epsilon_x + \epsilon_{\Phi} < 1$, by Lemma~\ref{lem:hessian_aprroximation} we have 
\begin{align*}
(1 - (\epsilon_x + \epsilon_{\Phi}))^2 H(x) \preceq H(x + \wt{\delta}_x) \preceq (1 - (\epsilon_x + \epsilon_{\Phi}))^{-2} H(x).
\end{align*}

Therefore,
\begin{align*}
\Phi_t( x+\delta_x, x+\wt{\delta}_x)
= & ~ \|  \nabla f_t(x + \delta_x) \|_{H(x + \wt{\delta}_x )^{-1} }\\
\leq & ~ (1-(\epsilon_x + \epsilon_{\Phi}))^{-1} \|  \nabla f_t(x + \delta_x) \|_{H(x)^{-1} }\\
= & ~  (1-(\epsilon_x + \epsilon_{\Phi}))^{-1}  \Phi_t( x+\delta_x, x).
\end{align*}
\end{proof}

\subsection{Newton step via approximate Hessian}\label{sec:newton_appro_hessian}

We state a useful fact which will be used later.
\begin{fact}\label{fac:Q}
Let $\gamma \in (0,1)$. Suppose $Q\in\R^{n\times n}$ satisfies $\gamma Q \preceq H(x) \preceq Q$, then we have 
\begin{enumerate}
    \item $\gamma H(x)^{-1} \preceq Q(x)^{-1} \preceq H(x)^{-1}$; \label{fac:Q_part_1}
    \item $0 \preceq H(x)^{-1} - Q(x)^{-1} \preceq (1-\gamma) H(x)^{-1}$. \label{fac:Q_part_2}
\end{enumerate}
\end{fact}
Here, we show the upper bound of $\Phi_t(x + \delta_x, x)$.
\begin{lemma}\label{lem:Phi_x_delta_x_x}
Let $\gamma \in (0,1)$ and suppose $Q\in\R^{n\times n}$ satisfies $\gamma Q \preceq H(x) \preceq  Q$. 
Let the movement be $\delta_x := -Q^{-1} \nabla f_t(x)$.
If $\Phi_t(x,x) \leq 0.1$, then we have 
    \begin{align*}
        \Phi_t(x + \delta_x, x) \leq (1-\gamma) \Phi_t(x, x) + 3\Phi_t(x,x)^2.
    \end{align*}
\end{lemma}

\begin{proof}
We first upper bound $\Phi_t(x+\delta_x,x)$ by two terms $(a)$ and $(b)$ separately,
\begin{align*}
    & ~ \Phi_t(x + \delta_x, x) \\
    = & ~ \| g(x+\delta_x) + t\cdot c \|_{H(x)^{-1}}\\
    = & ~ \Big\| g(x) + t\cdot c + H(x)\cdot \delta_x + \int_{0}^1 \left[ H( x+t\cdot \delta_x ) - H(x) \right] \cdot \delta_x \d t \Big\|_{H(x)^{-1}}\\
    \leq & ~ \underbrace{ \Big\| (I-H(x)Q^{-1}) \cdot ( g(x) + t\cdot c) \Big\|_{H(x)^{-1}} }_{a} + \underbrace{ \Big\| \int_{0}^1 [ H( x+t\cdot \delta_x ) - H(x) ] \cdot \delta_x \d t \Big\|_{H(x)^{-1}} }_{b},
\end{align*}
where the first step is by the definition of $\Phi$, the second step is by Taylor expansion of $g(x+\delta_x)$ on $x$.

In the following, we bound $(a)$ and $(b)$ separately.

\begin{align*}
    (a) 
    \leq & ~  \|I-H(x)Q^{-1}\|_{H(x)^{-1}} \cdot \|g(x) + t\cdot c\|_{H(x)^{-1}}\\
    = & ~ \max_{v}\frac{\langle v, (H(x)^{-1} - Q^{-1})v \rangle}{ \langle v, H(x)^{-1} v \rangle } \cdot \Phi_t(x,x)\\
    \leq & ~ (1-\gamma)\Phi_t(x,x).
\end{align*}
where the first step is by the definition of matrix norm and $\Phi_t(x,x)$, the second step follows from the definition of matrix norm and $\Phi_t(x,x)$ (Def.~\ref{def:phi}),
the last step follows from $0 \preceq H(x)^{-1} - Q^{-1} \preceq (1-\gamma) H^{-1}$ (Part~\ref{fac:Q_part_2} of Fact~\ref{fac:Q}). 

Next, we bound the term $(b)$. By the triangle inequality for integrals, we have 
\begin{align*}
(b) \leq \int_0^1 \Big \| ( H(x + \xi \cdot \delta_x) - H(x) ) \cdot \delta_x \Big\|_{H(x)^{-1}} \d \xi .
\end{align*}
Fix any $\xi \in [0,1]$. Define $y = x + \xi \cdot \delta_x$ and let $\epsilon_{\Phi}=\Phi_t(x,x)$.

By Lemma~\ref{lem:hessian_aprroximation} and $\|y-x\|_{H(x)} \leq \|\delta_x\|_{H(x)} \leq \epsilon_{\Phi}$, we have $-3\epsilon_{\Phi} H(x) \preceq H(y) - H(x) \preceq 3\epsilon_{\Phi} H(x)$.
Then
\begin{align*}
    & ~ \Big \| ( H(y) - H(x) ) \cdot \delta_x \Big\|_{H(x)^{-1}} \\
    \leq & ~  3 \epsilon_{\Phi} \|\delta_x\|_{H(x)}\\
    \leq & ~ 3\epsilon_{\Phi}^2,
\end{align*}
where the first step follows from $-3\epsilon_{\Phi} I \preceq H(x)^{-1/2} (H(y)-H(x)) H(x)^{-1/2} \preceq 3 \epsilon_{\Phi} I$. Since this bound holds for every $\xi\in[0,1]$, we get $(b)\leq 3\epsilon_{\Phi}^2$.

Overall we have 
\begin{align*}
\Phi_t(x+\delta_x,x) \leq (1-\gamma) \Phi_t(x,x) + 3 \Phi_t(x,x)^2.
\end{align*}

\end{proof}

\section{Solving LP in small space}\label{sec:lp_in_small_space}

We consider the streaming model and the regime where $m\gg n$. The input stream contains $A\in \R^{m\times n}$ and $b\in \R^m$. Further, $(a_i\in \R^n,b_i\in \R)$ are given one by one in the stream. We also assume $c\in\R^n$ is stored in memory.
\begin{assumption}\label{ass:input_of_stream}
We assume $(a_i\in \R^n,b_i\in \R)$ are given together in the input stream , for all $i\in [m]$. We also assume $c\in\R^n$ is stored in memory.
\end{assumption}

We will implement Algorithm~\ref{alg:ipm} with different barrier functions within $\wt O(n^2)$ space. Note that this is not possible for all previous primal-dual central paths, since they must maintain both primal and dual solution, which already costs $O(m)\gg O(n^2)$ space. We provide a generic algorithm template in Section \ref{sec:gen_alg_template}. In Section \ref{sec:log_small_space}, we give the implementation of IPM in small space with logarithmic barrier:
\begin{align*}
    \phi(x) := -\sum_{i \in [m]} \ln(a_i^{\top} x - b_i).
\end{align*}
In Section \ref{sec:hybrid_small_space}, we give the implementation of IPM in small space with hybrid barrier:
\begin{align*}
    V_{\rho}(x) := V(x) + \rho \cdot \phi(x),
\end{align*}
where $V(x)$ is the volumetric barrier function that defined as $V(x) := \frac{1}{2}\log(\det(\nabla^2 \phi(x)))$.

In Section \ref{sec:nsquare_lsbarrier}, we give the implementation of IPM in small space with Lee-Sidford barrier:
\begin{align*}
    f(x,w) := \ln \det ( A_x^{\top} W^{ 1 - 2 / q } A_x ) - ( 1 - 2 / q ) \tr[W],
\end{align*}
where $W=\diag(w)$.

\subsection{Generic algorithm template}
\label{sec:gen_alg_template}
In this section, we provide a generic algorithm template for dual-only robust central path that uses small spaces.

\begin{algorithm}[!ht]
\caption{Robust dual central path}
\label{alg:ipm_stream}
\begin{algorithmic}[1]
\Procedure{IPMStream}{$A,b,c,t_{\start},x_{\start},t_{\final},o,{\rm Bar}$} \Comment{Lemma~\ref{lem:ipm_error}}
\State $m,n\leftarrow $ dimensions of $A$
\State \Comment{$A \in \R^{m \times n}$ is the input matrix, $b \in \R^{m}$, $c\in \R^n$}
\State \Comment{$t_{\start},x_{\start}$ satisfy $\Phi_{t_{\start}}(x_{\start},x_{\start}) \leq \epsilon_{\Phi}$}
\State \Comment{$t_{\final} \in \R$ is the final goal of $t$}
\State \Comment{$o\in \{0,1\}$. If $o$ is $0$, the algorithm decreases $t$, otherwise the algorithm increases $t$}
\State \Comment{${\rm Bar}$ is a \textsc{Barrier} data structured}
\State Let $\epsilon_g,\epsilon_t,\epsilon_{\Phi},\gamma$ be parameters that meet Assumption~\ref{ass:input_of_ipm} 
\State $t\leftarrow t_{\start}$
\State $x\leftarrow x_{\start}$
\State ${\rm Bar}.\textsc{Init}(A, b, x_{\start})$ \Comment{Initialize the barrier data structure}
\State $T \leftarrow O(\epsilon_t^{-1}) \cdot |\log(t_{\final} / t_{\start})| $ \Comment{Number of iterations}
\For {$k \leftarrow 1$ to $T$}
    \State $\wt H(x)\gets {\rm Bar}.\textsc{ApproxHessian}(A, b, x, \gamma, t, c)$
    \State $\wt \nabla f_t(x)\gets {\rm Bar}.\textsc{ApproxGradient}(A, b, x, \epsilon_{g}, t, c)$
    \State $\wt \delta_x\gets -\wt H(x)^{-1}\cdot \wt \nabla f_t(x)$
    \State $x^{\new} \leftarrow x + \wt{\delta}_x$ 
    \If{$o=0$}
        \State $t^{\new} \leftarrow t \cdot ( 1 - \epsilon_t )$
    \Else
        \State $t^{\new} \leftarrow t \cdot ( 1 + \epsilon_t )$
    \EndIf
    \If { ($o=0$ \text{and} $t < t_{\final}$) \text{or} ($o=1$ \text{and} $t>t_{\final}$)}
        \State \textbf{break}
    \EndIf
    \State $t \leftarrow t^{\new}$
    \State $x \leftarrow x^{\new}$
\EndFor
\State \Return $x$ 
\EndProcedure
\end{algorithmic}
\end{algorithm}

As demonstrated by the generic algorithm, we need to implement the Bar data structure for different barrier functions in small space. Note that without counting for the space used by Bar, the algorithm uses $O(n^2)$ space per iteration.

\subsection{Logarithmic barrier}
\label{sec:log_small_space}
We start with perhaps the simplest barrier to compute, the logarithmic barrier. We provide an implementation of Barrier data structure under logarithmic barrier in Appendix~\ref{sec:app_log_bar}.

\begin{theorem}\label{thm:nsquare_logarithmic}
Under Assumption~\ref{ass:input_of_stream}, given any feasible linear program
\begin{align*}
\min_{x\in \R^n, Ax\geq b} c^{\top} x,
\end{align*}
where $A\in \R^{m\times n}$, $b\in \R^{m}$, and $c\in \R^n$. 
Suppose the solution exists and let $x^*\in \R^n$ be the solution. For any $\epsilon >0$, we can outputs an $x$ which is a \emph{nearly-optimal solution}
\begin{align*}
    c^{\top} x - c^{\top} x^* \leq \epsilon.
\end{align*}
in $O(n^2)$ space and $O(\sqrt{m}\log (mR^2\|c\|_2/(r\epsilon)))$ passes, provided the feasible region $\{x: Ax\geq b\}$ is contained in a ball of radius $R$ and contains a ball of radius $r$ with known center.
\end{theorem}

\begin{proof}
Let $F(x)$ be the logarithmic barrier (Def.~\ref{def:log_barrier_function}) with $\theta = \Theta(\sqrt{m})$. We show that \textsc{Interior\-Point\-Method} (Algorithm~\ref{alg:ipm}) can be implemented in $O(n^2)$ space. Lemma~\ref{lem:ipm_error} then finishes the proof.

Let $x_0$ be the known center of the ball of radius $r$, so that $s_i(x_0)\geq r\|a_i\|_2$ for all $i\in[m]$, and let $\epsilon_{\Phi}:=1/100$. We may assume $2R\|c\|_2>\epsilon$, since otherwise every feasible point, for example $x_0$, is $\epsilon$-optimal. We use the standard method that follows the central path twice, as in Algorithm~\ref{alg:minimum_vertex_cover}. The first walk uses the cost $-g(x_0)$, for which $x_0$ is on the central path at $t=1$, and decreases $t$. The second walk uses the cost $c$ and increases $t$ up to $t_{\final}:=\theta^2(1+2\epsilon_{\Phi})/\epsilon$, so that $c^{\top}x-c^{\top}x^*\leq\epsilon$ by Lemma~\ref{lem:ipm_error}. Since the feasible region lies in a ball of radius $R$ and contains the Dikin ellipsoid $\{y:\|y-x\|_{H(x)}<1\}$ of every feasible $x$, we have $\|v\|_{H(x)^{-1}}\leq 2R\|v\|_2$ for every feasible $x$ and every $v\in\R^n$. For the logarithmic barrier, $\|g(x_0)\|_2\leq m/r$, and the barriers below have weights at most $\poly(m)$, so $\|g(x_0)\|_2\leq \poly(m)/r$ in all cases. Together with the inverse Hessian bound above, this shows that the first walk can stop at some $t_1=\Omega(r/(R\poly(m)))$, close to the analytic center. It also shows that the second walk can start there at $t_{\start}=\Omega(1/(R\|c\|_2))$ with $\Phi_{t_{\start}}(x_{\start},x_{\start})\leq\epsilon_{\Phi}$. So the two walks take $O(\theta\log(mR/r))$ and $O(\theta\log(mR\|c\|_2/\epsilon))$ iterations, which is $O(\theta\log(mR^2\|c\|_2/(r\epsilon)))$ in total.

By definition of $\phi(x) $(Def.~\ref{def:log_barrier_function}), we have
\begin{align*}
    \nabla \phi(x)  & ~ = - \sum_{i \in [m]} \frac{a_i}{s_i(x)} \in \R^{n}; \\
    H(x) = \nabla^2 \phi(x) & ~  = \sum_{i \in [m]} \frac{a_i a_i^{\top}}{s_i(x)^2} \in \R^{n\times n}.
\end{align*}
In each iteration, we are given $x\in \R^n$. By Assumption~\ref{ass:input_of_stream}, when we read $(a_i,b_i)$, we can compute $s_i = a_i^{\top}x-b_i$, and then accumulate $-a_i/s_i$ to $\nabla \phi(x)$ and accumulate $a_ia_i^{\top}/s_i^{2}$ to $H(x)$. In this way, we can calculate exact $\nabla \phi(x)$ and $\nabla^2 \phi(x)$, and therefore $\wt{\delta}_x:= -H(x)^{-1}\nabla f_t(x)$ can be calculated without any error in $O(n^2)$ space. Since Assumption~\ref{ass:input_of_ipm} holds, by Lemma~\ref{lem:ipm_error}, we finish the proof.
\end{proof}

\subsection{Hybrid barrier}

\label{sec:hybrid_small_space}

Next, we show the space and passes needed for the hybrid barrier function. We provide an
implementation of hybrid barrier data structure.
\begin{theorem}\label{thm:nsquare_hybrid}
Under Assumption~\ref{ass:input_of_stream}, given any feasible linear program
\begin{align*}
\min_{x\in \R^n, Ax\geq b} c^{\top} x,
\end{align*}
where $A\in \R^{m\times n}$, $b\in \R^{m}$, and $c\in \R^n$. 
Suppose the solution exists and let $x^*\in \R^n$ be the solution. For any $\epsilon >0$, we can outputs an $x$ which is a \emph{nearly-optimal solution}
\begin{align*}
    c^{\top} x - c^{\top} x^* \leq \epsilon.
\end{align*}
in $O(n^2)$ space and $O((nm)^{1/4}\log (mR^2\|c\|_2/(r\epsilon)))$ passes, provided the feasible region $\{x: Ax\geq b\}$ is contained in a ball of radius $R$ and contains a ball of radius $r$ with known center.
\end{theorem}

\begin{proof}
Similar to the proof of Theorem~\ref{thm:nsquare_logarithmic}, here we only show how to calculate $\wt{\delta}_x$.

Here we let $F(x):=k\cdot V_{\rho}(x)$, where $V_{\rho}$ is the hybrid barrier of $\rho = (n/m)$ (Def.~\ref{def:hy_barrier_function}) and $k=O(\sqrt{m/n})$ is the positive factor of Theorem~\ref{thm:v_rho_self_concor}. Then $F$ is $\theta^2$-self-concordant with $\theta = \Theta((nm)^{1/4})$.
Let $\sigma(x)$ be the definition~\ref{def:sigma}:
\begin{align}
    \sigma_i(x) := \frac{ a_i^{\top} (\nabla^2 \phi(x) )^{-1} a_i }{ (a_i^{\top} x - b_i)^2 },~~ \forall i \in [ m].
\end{align}

Let $Q(x)$ be the definition~\ref{def:Q}:
\begin{align*}
    Q(x) := \sum_{i=1}^m \sigma_i(x) \frac{ a_i a_i^\top }{ ( a_i^\top x - b_i )^2 },
\end{align*}
In the proof of Theorem~\ref{thm:nsquare_logarithmic}, we already showed that $\nabla^2 \phi(x)$ can be calculated and stored using one pass. In the next pass, when we get $(a_i,b_i)$, we can compute $\sigma_i(x)$ exactly, and we accumulate  $\frac{\sigma_i(x)}{s_i(x)^2} a_i a_i^\top$ to $Q(x)$. Finally, we will get exact $Q(x)$.

Let 
\begin{align*} 
\wt{H}(x) = k\cdot \big( 5Q(x) + (n/m) \nabla^2 \phi(x) \big).
\end{align*}

By Lemma~\ref{lem:Q_approximate}, we have $Q(x) \preceq \nabla^2 V(x) \preceq 5 Q(x)$. Since $H(x) = k\cdot ( \nabla^2 V(x) + (n/m) \nabla^2 \phi(x))$, we have $\frac{1}{5}\wt{H}(x)\preceq H(x) \preceq \wt{H}(x)$. 

On the other side, $\nabla f_t(x) = tc - k\sum_{i=1}^m (\sigma_i(x) + n/m)\frac{a_i}{s_i(x)}$ can be computed exactly. Therefore, we can compute $\wt{\delta}_x := -\wt{H}(x)^{-1} \nabla f_t(x)$ in $O(n^2)$ space. 

By setting $\epsilon_{\Phi} = 1/100$ and $\epsilon_t = 1/\Omega(\theta)$, we meet Assumption~\ref{ass:input_of_ipm} and get $O(n^2)$ space and $O((nm)^{1/4}\log (mR^2\|c\|_2/(r\epsilon)))$ 
pass for hybrid barrier.
\end{proof}
For completeness, we include our implementation (Algorithm \ref{alg:hybridbar_space}) in Appendix \ref{sec:app_hybrid_bar}.

\subsection{Lee-Sidford barrier}\label{sec:nsquare_lsbarrier}

The near-universal Lee-Sidford barrier is the crux of many fastest algorithms~\cite{ls14,blss20,blnpsssw20,bllsssw21}. To compute the Newton direction, it is imperative to give a small space implementation of $\ell_p$ Lewis weights, for $p=\Theta(\log m)$. We show that Lewis weights can be computed with $\wt O(1)$ leverage scores recursively. We defer the algorithm to Appendix~\ref{sec:app_lee_bar}.

\begin{theorem}[Formal version of Theorem \ref{thm:nsquare_lsbarrier_informal}]\label{thm:nsquare_lsbarrier} Under Assumption~\ref{ass:input_of_stream}, given any feasible linear program
\begin{align*}
\min_{x\in \R^n, Ax\geq b} c^{\top} x,
\end{align*}
where $A\in \R^{m\times n}$, $b\in \R^{m}$, and $c\in \R^n$. 
Suppose the solution exists and let $x^*\in \R^n$ be the solution. For any $\epsilon >0$, we can outputs an $x$ which is a \emph{nearly-optimal solution}
\begin{align*}
    c^{\top} x - c^{\top} x^* \leq \epsilon.
\end{align*}
in $\wt{O}(n^2)$ space and $\wt{O}(\sqrt{n}\log (mR^2\|c\|_2/(r\epsilon)))$ passes, provided the feasible region $\{x: Ax\geq b\}$ is contained in a ball of radius $R$ and contains a ball of radius $r$ with known center.
\end{theorem}
\begin{proof}
Similar to the proof of Theorem~\ref{thm:nsquare_logarithmic}, 
we can assume we get $x_{\start},t_{\start}$ such that 
\begin{align*} 
\Phi_{t_{\start}}(x_{\start},x_{\start})\leq \epsilon_{\Phi},
\end{align*}
where we set $\epsilon_{\Phi}:=(\log m)^{-2}$.

Let $F(x)=\psi(x)$ be Lee-Sidford barrier (Def.~\ref{def:ls_barrier}). 

We let 
\begin{align*}
\wt{H}(x) = (1+q)A^{\top}S(x)^{-1}W_xS(x)^{-1}A.
\end{align*}

By Lemma~\ref{lem:gra_hess_ls_barrier}, 
\begin{align*} 
\gamma \cdot \wt{H}(x) \preceq \nabla^2 \psi(x) \preceq \wt{H}(x),
\end{align*}
where $\gamma = \frac{1}{1+q}$.

Now we show how to compute $w_x$ in $\wt{O}(n^2)$ space and in $\poly(\log m)$ passes. We use Algorithm~\ref{alg:lewis_weight_small_space} to compute $w_x$. It implements Algorithm~\ref{alg:compute_lewis_weight} with $p=q$. Since $\psi$ is defined through $A_x=S(x)^{-1}A$, we apply the algorithm to $A_x$. When we read $(a_i,b_i)$, we form the row $a_i/s_i(x)$ of $A_x$. In the rest of this proof, $a_i$ denotes this row and $A$ denotes $A_x$. 

Instead of storing the weights $w \in \R^m$ which cost $m$ space, we store matrices $Q \in \R^{n\times n}$ such that given $a_i$, we can output $w_i$ using these matrices.

{\sloppy Algorithm~\ref{alg:compute_lewis_weight} changes the weights only through calls to \textsc{Descent} (Algorithm~\ref{alg:descent}). A call is either on all rows, or it is inside \textsc{Round} (Algorithm~\ref{alg:round}) and on the rows of the set $C$. Let $J$ be the total number of calls. For $t\in \{0,1,\cdots,J\}$, let $v^{(t)}\in \R^m$ be the weight vector after the first $t$ calls, so that $v^{(0)}_i=\frac{n}{m}$ for all $i\in [m]$, and let $V^{(t)}=\diag(v^{(t)})$. For $t\geq 1$, let $z_t=1$ if the $t$-th call is inside \textsc{Round} and $z_t=0$ if it is on all rows. Every call computes its step from the leverage scores of the current weight vector only. So we store one matrix $Q^{(t)}=(A^{\top}V^{(t)}A)^{-1}$ and the bit $z_t$ for each $t\leq J$. One pass computes $Q^{(0)}=(\sum_{i=1}^m \frac{n}{m}\cdot a_ia_i^{\top})^{-1}$.\par}

{\sloppy Suppose we have stored $Q^{(0)},\cdots,Q^{(t)}$ and $z_1,\cdots,z_t$ for some $t\geq 0$, and that given $a_i$ we can compute $v^{(t)}_i$. We show how to compute $z_{t+1}$ and $Q^{(t+1)}$ in two passes. In the first pass, for each $i\in [m]$, given $a_i\in \R^n$, we can compute $\sigma_i(v^{(t)}) = v^{(t)}_i\cdot a_i^{\top}Q^{(t)}a_i$, so we can compute $\rho_i(v^{(t)})$. We record whether $\rho_i(v^{(t)})>1+\alpha$ for at least one $i\in [m]$. If so, the set $C$ of \textsc{Round} is not empty, the next call is $\textsc{Descent}(v^{(t)},C,\frac{1}{3\ov{\alpha}}\cdot \mathbf{1})$ and we set $z_{t+1}=1$. Otherwise \textsc{Round} returns $v^{(t)}$. If the main loop of Algorithm~\ref{alg:compute_lewis_weight} has not finished its $T$ iterations, the next call is then $\textsc{Descent}(v^{(t)},[m],\frac{1}{3\ov{\alpha}}\cdot \mathbf{1})$ and we set $z_{t+1}=0$. If the main loop has finished, then $w_R=v^{(t)}$, $t=J$ and there is no further call. In the second pass, for each $i\in [m]$ we compute $\rho_i(v^{(t)})$ again. If $z_{t+1}=0$ or $\rho_i(v^{(t)})>1+\alpha$, we set $v^{(t+1)}_i=v^{(t)}_i\cdot (1+\frac{1}{3\ov{\alpha}}\cdot \frac{\rho_i(v^{(t)})-1}{\rho_i(v^{(t)})+1})$. Otherwise we set $v^{(t+1)}_i=v^{(t)}_i$. Then we accumulate $a_ia_i^{\top}v^{(t+1)}_i$ to $K^{(t+1)}$.\par}

After this pass, we get 
\begin{align*} 
K^{(t+1)} = A^{\top}V^{(t+1)}A \in \R^{n\times n} 
\end{align*}
and then we take a inverse to get 
\begin{align*}
Q^{(t+1)} = (K^{(t+1)} )^{-1} = (A^{\top}V^{(t+1)}A)^{-1}
\end{align*}
in $O(n^2)$ space.

The purpose of computing $Q^{(t+1)}$ is that in the future, once we read $a_i$ in the input stream, we can recover 
\begin{align*} 
\sigma_i(v^{(t+1)}) = v^{(t+1)}_i\cdot a_i^{\top}Q^{(t+1)}a_i.
\end{align*}

Thus, to compute $v^{(t+1)}_i$, we need to recursively compute $v^{(t)}_i$, $v^{(t-1)}_i$, $\cdots$, until we reach $v^{(0)}_i = \frac{n}{m}$. This recursion uses only $a_i$, the stored matrices and the stored bits, so it takes no extra pass. 

We run the algorithm with accuracy $\epsilon_w/4$, where $\epsilon_w:=\epsilon_g/\sqrt{n}$. By Lemma~\ref{lem:flps}, the number of calls is $J=O(q^3\log(qm/\epsilon_w))$, which is $\poly(\log m)$ for $q=\Theta(\log m)$ and the choice of $\epsilon_g$ below. The whole process takes $\wt{O}(n^2)$ space and $\poly( \log m)$ passes, since it stores $J+1$ matrices and uses two passes per call.

The output is the vector $\wh{w}$ with $\wh{w}_i=(a_i^{\top}Q^{(J)}a_i)^{1/\alpha}$, and we set $\wt{w}_i:=\wh{w}_i^{1+\alpha}$. By Lemma~\ref{lem:flps} and the discussion after it, $(1-\epsilon_w/4)^{1+\alpha}\wt{w}\leq w_x\leq (1+\epsilon_w/4)^{1+\alpha}\wt{w}$ entrywise. Since $\alpha\leq 1$ and $\epsilon_w\leq 1$, this gives $(1-\epsilon_w)w_x\leq \wt{w}\leq (1+\epsilon_w)w_x$ entrywise. Since $\nabla^2\psi(x)\succeq A_x^{\top}W_xA_x$ and $\sum_{i\in[m]}(w_x)_i=n$, the gradient error is
\begin{align*}
\|A_x^{\top}(\wt{w}-w_x)\|_{\nabla^2\psi(x)^{-1}}\leq \|W_x^{-1/2}(\wt{w}-w_x)\|_2\leq \epsilon_w\sqrt{n}=\epsilon_g,
\end{align*}
and $(1+q)(1-\epsilon_w)^{-1}A_x^{\top}\wt{W}A_x$ satisfies Part~\ref{part:H(x)} of Assumption~\ref{ass:input_of_ipm} with $\gamma=(1-\epsilon_w)/((1+\epsilon_w)(1+q))$. Since $\psi$ is $\theta^2$-self-concordant with $\theta^2=O(n\log^{O(1)} m)$~\cite{ls19}, we take $\epsilon_t=\Theta(1/(\theta\log^3 m))$ and $\epsilon_g=\Theta(\log^{-3} m)$ with small enough constants. This choice meets Part~\ref{part:equation} of Assumption~\ref{ass:input_of_ipm}. The IPM has $\wt{O}(\sqrt{n}\log(mR^2\|c\|_2/(r\epsilon)))$ iterations, each using $\poly(\log m)$ passes, and we can reuse space in each iteration. So overall we use $\wt{O}(n^2)$ space and $\wt{O}(\sqrt{n}\log(mR^2\|c\|_2/(r\epsilon)))$ passes.

\end{proof}
\section{SDD solver in the streaming model} \label{sec:sdd_solver_streaming_model}

In this section, we present an SDDM solver in the streaming model. Later in Section~\ref{sec:solver_reduction} we reduce the problem of solving an SDD$_0$ system to solving an SDDM system and therefore give an SDD$_0$ solver in the streaming model. 
The reason of using SDD solver but not a simpler Laplacian solver is because we attach an identity matrix of size $n\times n$ into the input edge-vertex incident matrix (see Line~\ref{line:mvc_A_2_b_2_c_2} of Algorithm~\ref{alg:minimum_vertex_cover}) so that only an SDD solver could handle it.

The definitions of SDD matrix and SDDM matrix can be found in Definition~\ref{def:sddm_sdd}.

In Section \ref{sec:ls}, we give the related definitions and lemmas about spectral sparsifier. 
In Section \ref{sec:preconditioner}, we introduce the preconditioner. In Section \ref{sec:iterative_solver}, we introduce a streaming SDDM solver. In Section \ref{sec:iterative_solver_space_pass}, we show the space and passes needed for the iterative solver. In Section \ref{sec:iterative_solver_accuracy}, we show the accuracy of the iterative solver. In Section \ref{sec:SDD_main_result}, we provide our main result for solving SDDM and SDD$_0$ systems.

\subsection{SDD and Laplacian systems}\label{sec:ls}

Spectral sparsifier is the crux for our space and pass-efficient SDD solver. We briefly review its literature.

{\bf Streaming spectral sparsifer.}
Initialized by the study of cut sparsifier in the streaming model \cite{ag09}, a simple one-pass semi-streaming algorithm for computing a spectral sparsifier of any weighted graph is given in \cite{kl11}, which suffices for our applications.
The problem becomes more challenging in a dynamic setting, i.e., both insertion and deletion of edges from the graph are allowed. 
Using the idea of linear sketching, \cite{klmss17} gives a single-pass semi-streaming algorithm for computing the spectral sparsifier in the dynamic setting. 
However, their brute-force approach to recover the sparsifier from the sketching uses $\Omega(n^2)$ time. 
An improved recover time is given in \cite{kmm+19} but requires more spaces, e.g., $\epsilon^{-2} n^{1.5} \log^{O(1)} n$.
Finally, \cite{knst19} proposes a single-pass semi-streaming algorithm that uses $\eps^{-2} n\log^{O(1)} n$ space and $\eps^{-2} n\log^{O(1)} n$ recover time to compute an $\eps$-spectral sparsifier which has $O(\epsilon^{-2} n \log n)$ edges. 
Note that $\Omega( \epsilon^{-2}  n \log n)$ space is necessary for this problem \cite{ckst19}.

\begin{definition}[$\delta$-spectral sparsifier] \label{def:sparsifier}
    Given a weighted undirected graph $G$ and a parameter $\delta > 0$, an edge-reweighted subgraph $H$ of $G$ is an \emph{$\delta$-spectral sparsifier} of $G$ if\footnote{We also say $L_H$ is an $\delta$-spectral sparsifier of $L_G$.}
    \begin{equation*}
        (1 - \delta) \cdot x^{\top} L_G x  \le  x^{\top} L_H x \le (1 + \delta) \cdot x^{\top} L_G x , ~~~ \forall x \in \R^n,
    \end{equation*}
    where $L_G$ and $L_H$ are the Laplacians of $G$ and $H$, respectively.
\end{definition}

Here, we show the space and passes needed to compute a $\delta$-spectral sparsifier of the given graph.
\begin{lemma}[\cite{klmss17}] \label{lem:sss}
    Let $G$ be a weighted graph and $\delta \in (0,1)$ be a parameter. 
    There exists a streaming algorithm that takes $G$ as input, uses $\delta^{-2}n\poly(\log n)$ space and $1$ pass, and outputs a weighted graph $H$ with $\delta^{-2}n\poly(\log n)$ edges such that with probability at least $1-1/\poly(n)$, $H$ is a $\delta$-spectral sparsifier of $G$.
\end{lemma}

We will use the classic SDD solver in the sequential model, which is formally described below.

\begin{theorem}[\cite{st04}] \label{thm:classic_lsolve}
    There is an algorithm which takes input an SDD$_0$ matrix $A$, a vector $b \in \mathbb{R}^n$, and a parameter $\epsilon \in (0, 1/2)$,
    if there exists $x^* \in \mathbb{R}^n$ such that $A x^* = b$,
    then with probability $1 - 1 / \poly(n)$, the algorithm returns an $x \in \mathbb{R}^n$ such that $\|x - x^*\|_A \le \epsilon \cdot \|x^*\|_A$ in \begin{align*}
        \nnz(A) \cdot \poly(\log n) \cdot \log(1 / \epsilon)
    \end{align*} 
    time.
    The returned $x$ is called an \emph{$\epsilon$-approximate solution} to the SDD$_0$ system $Ax = b$.
\end{theorem}

\subsection{The preconditioner}
\label{sec:preconditioner}

To prove that our SDDM solver (Algorithm~\ref{alg:stream_lsolve}) gives the desired accuracy,
we need the concept of a \emph{preconditioner} (and how to compute the preconditioner of an SDDM matrix).

We define preconditioner as follows:
\begin{definition}[Preconditioner]\label{def:preconditioner}
For any positive definite matrix $A \in \R^{n \times n}$ and accuracy parameter $\delta > 0$, we say $P$ is a $\delta$-preconditioner if
\begin{align*}
    \| P^{-1} A x - x \|_{A} \leq \delta \cdot \| x \|_A, ~~~ \forall x \in \R^n.
\end{align*}
\end{definition}

Next, we give some properties of the preconditioner.
\begin{lemma}\label{lem:sss_property}
    Let $A$ be an SDDM matrix (Definition~\ref{def:sddm_sdd}) and let $A = L_G + D$ where $L_G$ is a Laplacian matrix of graph $G$ and $D \succ 0$ is a diagonal matrix with non-negative entries.\footnote{By Fact~\ref{fct:sddm_psd}, such decomposition always exists.}
    For any $\delta \in (0, 1/2)$, if $H$ is a $\delta$-spectral sparsifier of $G$, and we define $P:= L_H + D$. Then $P$ satisfies the following two conditions: 
    \begin{itemize}
        \item $P $ is a $(2\delta)$-preconditioner(Definition~\ref{def:preconditioner}) of $A$, 
        \item $P \succ 0$.
    \end{itemize}
\end{lemma}
\begin{proof}
    By Definition~\ref{def:sparsifier},
    $$(1 + \delta) x^{\top} L_G x \ge x^{\top} L_H x \ge (1 - \delta) x^{\top} L_G x \ge 0,$$
    then $L_H \succeq 0$ as it must be symmetric. 
    Since $D \succ 0$, we have that $A \succ 0$, $P \succ 0$, and 
    \[
    (1 + \delta) x^{\top} A x \ge x^{\top} P x \ge (1 - \delta) x^{\top} A x > 0.
    \]
    So we obtain
    \begin{equation*}
        \frac{1}{1 - \delta} x^{\top} A^{-1} x \ge x^{\top} P^{-1} x \ge \frac{1}{1 + \delta} x^{\top} A^{-1} x,
    \end{equation*}
    which, by $\delta \in (0, 1/2)$, implies
    \begin{equation*}
        (1 + 2\delta) x^{\top} A^{-1} x \ge x^{\top} P^{-1} x \ge (1 - 2\delta) x^{\top} A^{-1} x.
    \end{equation*}
    Since $A$ is positive definite, replacing $x$ with $A^{-1/2} x$ we get
    \begin{equation} \label{eq:IH}
        (1 + 2\delta) x^{\top} x \ge x^{\top} A^{1/2} P^{-1} A^{1/2} x \ge (1 - 2\delta) x^{\top} x.
    \end{equation}
    
    Therefore,
    \begin{align} \label{eq:HGx}
        2 \delta x^{\top} x \ge x^{\top} \underbrace{ ({A}^{1/2} P^{-1} {A}^{1/2} - I) }_{ :=M } x \ge -2 \delta x^{\top} x.
    \end{align}
    Since Eq.~\eqref{eq:HGx} holds for any $x \in \mathbb{R}^n$ and $M$ is symmetric, using the spectral theorem, we get
    \begin{equation*}
        2 \delta \ge \lambda_1(M) \ge \lambda_n(M) \ge -2 \delta.
    \end{equation*}
    Next, we have that 
    \begin{align*}
        \|Mx\|_2 \leq \max\{ | \lambda_n(M) | , | \lambda_1(M) | \}\|x\|_2 \le 2 \delta \|x\|_2.
    \end{align*}
    Let $y :=  {A}^{-1/2} x$, we get that 
    \begin{align} \label{eq:Gy}
        2 \delta \|{A}^{1/2} y\|_2 
        \ge & ~ \|M {A}^{1/2} y\|_2 \notag \\
        = & ~ \|{A}^{1/2} P^{-1} A y - {A}^{1/2} y\|_2 .
    \end{align}
    Finally, rewriting both sides of Eq.~\eqref{eq:Gy} by the definition of matrix norm, we obtain
    \begin{equation*}
        2 \delta \|y\|_{A} \ge \|P^{-1} A y - y\|_{A}
    \end{equation*}
    for any $y \in \mathbb{R}^n$, giving the lemma.
\end{proof}

\subsection{An iterative solver} \label{sec:iterative_solver}

In this section, we present a streaming SDDM solver (Algorithm~\ref{alg:stream_lsolve}) that takes matrix $A$, vector $b$ and the error parameter $\epsilon$ as input. Particularly, $A$ is an SDD$_0$ matrix that satisfies $A = L_G +D \in \R^{n \times n}$, where $L_G$ is from input stream and $D$ is a diagonal matrix stored in the memory. 

Here, we give a brief overview of our implementation (Algorithm \ref{alg:stream_lsolve}). We first compute and store a $\delta$-spectral sparsifier $H$ of $L_G$. Next, we compute and store a $\delta$-spectral sparsifier $H$ of $L_G$. Then, we update $r_t$ and $x_t$ iteratively by computing a $\wt{\delta}$-approximate solution $y_t$ to $P y = r_t$. Finally, the solver will return an $\epsilon$-approximate solution to $A^{-1}b$. This algorithm takes $\widetilde{O}(n)$ space and $O(\log(1/\epsilon) / \log\log n)$ passes for approximately solving SDDM system $Ax = b$ with error parameter $\epsilon \in (0, 1/10)$.

\begin{algorithm}[ht]
\caption{ 
A streaming SDDM solver. It takes matrix $A$, vector $b$ and the error parameter $\epsilon$ as input. It return an $\epsilon$-approximate solution to $A^{-1}b$}
\label{alg:stream_lsolve}
\begin{algorithmic}[1]
\Procedure{StreamLS}{$A = L_G + D \in \R^{n\times n}$, $b\in \R^{n}$, $\epsilon\in \R$} \Comment{Theorem \ref{thm:stream_lsolve} 
}
\State \Comment{Note that $L_G$ is from input stream. $D$ is diagonal matrix stored.}
\State $\delta \leftarrow 1 / \log n $
\State $\wt{\delta} \leftarrow \delta /2$
\State $T \leftarrow O(\max\{1, ( \log(1/\epsilon) ) / {\log\log n}\})$ \Comment{Number of iterations}
\State Compute and store a $\delta$-spectral sparsifier $H$ of $L_G$ \label{line:compute_H} \Comment{Use one pass, Lemma~\ref{lem:sss}}
\State Compute a $2\delta$-preconditioner of $A$ as $P :=  L_H + D$
\State $r_0 \leftarrow b$, $x_0 \leftarrow \mathbf{0}_n$ \Comment{$r_0,x_0\in\R^n$}
\For {$t \leftarrow 0$ to $T-1$} \label{line:iteration_begin}
    \State Compute a $\wt{\delta}$-approximate solution $y_t$ to $P y = r_t$ by an SDD$_0$ solver and store $y_t$ in the memory \label{line:compute_y} \Comment{Theorem~\ref{thm:classic_lsolve}}
    \State $r_{t+1} \leftarrow r_t - A y_t$ \label{line:compute_ry} \\
    \State $x_{t+1} \leftarrow x_t + y_t$
\EndFor \label{line:iteration_end}
\State {\bf return} $x_T$.
\EndProcedure
\end{algorithmic}
\end{algorithm}

\subsection{An iterative solver: space and passes}
\label{sec:iterative_solver_space_pass}
We show that Algorithm~\ref{alg:stream_lsolve} takes $\wt O(n)$ space and $O(\log(1/\epsilon))$ passes.

\begin{lemma} \label{lem:streamls_space}
    Let $A = L_G + D \in \R^{n\times n}$ be an SDDM matrix where $L_G \in \R^{n\times n}$ is the Laplacian matrix of graph $G$ with weight $w$, and $D\in \R^{n\times n}$ is a diagonal matrix. If we can read all edge-weight pairs $(e,w_e)$ in one pass, and if we can read the diagonal of matrix $D$ in one pass, then
    $\textsc{StreamLS}(A, b, \epsilon)$ takes $O(\max\{1, \log(1/\epsilon) / \log\log n\})$ passes and $\widetilde{O}(n)$ space.\footnote{The algorithm can be implemented in the standard RAM model with finite precision by introducing an $\widetilde{O}(1)$ factor in the encoding, which translates to a multiplicative factor of $\widetilde{O}(1)$ in the space \cite{st04_sdd}.}
\end{lemma}
\begin{proof}
    By Lemma~\ref{lem:sss} and $\delta = 1/\log n$, the $\delta$-spectral sparsifier $H$ has $\widetilde{O}(n)$ edges and can be computed in $1$ pass and $\wt{O}(n)$ space with probability $1 - 1 / \poly(n)$. Therefore, computing the $2\delta$-preconditioner $P$ also takes $1$ pass and $\wt{O}(n)$ space.
    Note that $P \succ 0$ and thus the system $P y = r_t$ always has a solution.

    It remains to prove that each iteration of Lines~\ref{line:iteration_begin}-\ref{line:iteration_end} takes $1$ pass and $\widetilde{O}(n)$ space.
    Since any iteration $t$ only needs the vectors subscripted by $t$ and $t+1$, we can reuse the space such that the total space is $\widetilde{O}(n)$.

    Since $\text{nnz}(P) = \widetilde{O}(n)$ and
    \begin{equation*}
        \wt{\delta} = \delta/2 = 1/(2\log n),
    \end{equation*}
    in Line~\ref{line:compute_y}, by Theorem~\ref{thm:classic_lsolve}, with probability $1 - 1 / \poly(n)$ a $\wt{\delta}$-approximate solution $y_t$ can be found in $\widetilde{O}(n \log(1/ \wt{\delta} )) = \widetilde{O}(n)$ time, and therefore in $\wt{O}(n)$ space.
    Note that this step does not read the stream.

    In Line~\ref{line:compute_ry}, computing $r_{t+1}$ requires computing $A y_t$, which is done by reading the stream of $L_G$ and $D$ for $1$ pass and multiplying the corresponding entries and adding up to the corresponding coordinate.
    All vectors are in $\mathbb{R}^n$, so the total space used in each iteration is $\widetilde{O}(n)$.
    The lemma follows immediately from 
    \begin{align*}
        T = O(\max\{1, \log(1/\epsilon) / \log\log n\})
    \end{align*}
    and a union bound. 
\end{proof}

\subsection{An iterative solver: the accuracy}
\label{sec:iterative_solver_accuracy}

We prove the number of iterations required to converge is at most $\wt O(1)$, by demonstrating the effective of the iterative refinement via preconditioner. 

\begin{lemma} \label{lem:xt_iteration}
Let $\delta \in (0,1/10)$. 
    For any $t \in [0, T]$, $\|x_t - A^{-1} b\|_{A} \le (4 \delta)^t \cdot \|A^{-1} b\|_{A}$.
\end{lemma}
\begin{proof}
    The proof is by an induction on $t$.
    In the basic case of $t = 0$, we have $x_t = 0$ so the statement clearly holds.
    Assuming the lemma holds for $t$, we prove the inductive step for $t+1$.

    Since $y_t$ is a $\wt{\delta}$-approximate solution (Line~\ref{line:compute_y}, Algorithm~\ref{alg:stream_lsolve}), by Theorem~\ref{thm:classic_lsolve} we have
    \begin{equation*}
        \|y_t - P^{-1} r_t\|_{P} \le \wt{\delta} \cdot \| P^{-1} r_t\|_{P} .
    \end{equation*}
    By definition of the matrix norm, this becomes
    \begin{equation} \label{eq:y_t_H}
        (y_t - P^{-1} r_t)^{\top} P (y_t - P^{-1} r_t) \le \wt{\delta}^2 \cdot ( P^{-1} r_t)^{\top} P ( P^{-1} r_t ) .
    \end{equation}
    Since $(1-\delta)A\preceq P\preceq (1+\delta)A$ (see the proof of Lemma~\ref{lem:sss_property}), assuming $y_t \neq P^{-1} r_t$ and applying this on both sides of Eq.~\eqref{eq:y_t_H}, we have
    \begin{equation} \label{eq:y_t_G}
        (1 - \delta) \cdot (y_t - P^{-1} r_t)^{\top} A (y_t - P^{-1} r_t) \le (1 + \delta) \wt{\delta}^2 \cdot ( P^{-1} r_t )^{\top} A ( P^{-1} r_t ) .
    \end{equation}
    If $y_t = P^{-1} r_t$, then Eq.~\eqref{eq:y_t_G} also holds since the right-hand side of Eq.~\eqref{eq:y_t_G} is non-negative due to $A \succeq 0$.
    Note that Eq.~\eqref{eq:y_t_G} implies
    \begin{align} \label{eq:y_t_HG}
        \| y_t - P^{-1} r_t \|_{A} 
        \le & ~  \wt{\delta} \sqrt{ (1 + \delta) / (1 - \delta) } \cdot \| P^{-1} r_t\|_{A} \notag \\
        \le & ~ 2 \wt{\delta} \cdot \| P^{-1} r_t\|_{A} \notag \\
        \le & ~ \delta \cdot \| P^{-1} r_t \|_{A} ,
    \end{align}
    where the second step follows from $(1+\delta)/(1-\delta) \leq 4$ when $\delta \in (0,1/2)$, and the last step follows from $\wt{\delta} = \delta / 2$.
    
    Before continuing, we observe the following, which easily follows from the update rule and an induction on $t$:
    \begin{equation}\label{eq:r_t}
        r_t = r_0 - A \sum_{i=0}^{t-1} y_i = r_0 - A (x_t - x_0) = b - A x_t .
    \end{equation}
    Using Eq.~\eqref{eq:r_t}, we bound the left-hand side of Eq.~\eqref{eq:y_t_HG} from above by the following:
    \begin{align}\label{eq:bound_y_t}
        & ~ \| y_t - P^{-1} r_t \|_{A} \notag \\
        \le & ~ \delta \cdot \| P^{-1} (b - A x_t)\|_{A} \notag \\
        = & ~ \delta \cdot \| P^{-1} A ( A^{-1} b - x_t) - ( A^{-1} b - x_t) + ( A^{-1} b - x_t)\|_{A} \notag \\
        \le & ~ \delta \cdot \left(\| P^{-1} A (A^{-1} b - x_t) - ( A^{-1} b - x_t)\|_{A} + \| A^{-1} b - x_t\|_{A} \right) \notag \\
        \le & ~ \delta \cdot \left(2\delta \| A^{-1} b - x_t\|_{A} + \| A^{-1} b - x_t\|_{A} \right) \notag \\
        = & ~ 2 \delta \cdot \| A^{-1} b - x_t\|_{A} , 
    \end{align}
    where the third step follows from the triangle inequality, the fourth step follows from Lemma~\ref{lem:sss_property}, the last step follows from $\delta \in (0,1/2)$.

    Finally, we are ready to prove the inductive step:
    \begin{align*}
      \|x_{t+1} - A^{-1} b\|_{A} 
        = & ~ \|x_t - A^{-1} b + y_t\|_{A} \\
        = & ~ \|(x_t - A^{-1} b + P^{-1} r_t) + (y_t - P^{-1} r_t)\|_{A} \\
        \le & ~ \|x_t - A^{-1} b + P^{-1} r_t\|_{A} + \|y_t - P^{-1} r_t\|_{A} \\
        = & ~ \|x_t - A^{-1} b + P^{-1} b - P^{-1} A x_t\|_{A} + \|y_t - P^{-1} r_t\|_{A} \\
        \le & ~ \|P^{-1} A ( A^{-1} b -  x_t) - ( A^{-1} b -  x_t)\|_{A} + 2 \delta \cdot \| A^{-1} b - x_t\|_{A} \\
        \le & ~ 2 \delta \cdot \| A^{-1} b -  x_t\|_{A} + 2\delta \cdot \|{A}^{-1} b - x_t\|_{A} \\
        \le & ~ (4 \delta)^{t+1} \cdot \| A^{-1} b\|_{A} ,
    \end{align*}
    where the first step follows from $x_{t+1} = x_t + y_t$, the third step follows from the triangle inequality,
    the fourth step follows from Eq.~\eqref{eq:r_t},
    the fifth step follows from Eq.~\eqref{eq:bound_y_t},
    the sixth step follows from Lemma~\ref{lem:sss_property},
    and the last step follows from the induction hypothesis,
    completing the proof.
\end{proof}

\subsection{Main result}
\label{sec:SDD_main_result}
We are ready to show SDDM and ${\rm SDD}_0$ systems can be solved in the advertised space and passes.

\begin{lemma} \label{lem:stream_lsolve}
Let $\epsilon \in (0,1)$ and $b\in \R^{n}$. Let $A = L_G + D \in \R^{n\times n}$ be an SDDM matrix where $L_G \in \R^{n\times n}$ is the Laplacian matrix of graph $G$ with weight $w$, and $D\in \R^{n\times n}$ is a diagonal matrix.
If we can read all edge-weight pairs $(e,w_e)$ in one pass, and if we can read the diagonal of matrix $D$ in one pass, then with probability $1 - 1/\poly(n)$, 
$\textsc{StreamLS}(A, b, \epsilon)$ returns an $\epsilon$-approximate solution $x$,
i.e.,
\begin{equation*}
    \|x - A^{-1} b\|_{A} \le \epsilon \cdot \|A^{-1} b\|_{A} .
\end{equation*}
in $O(\max\{1, \log(1/\epsilon) / \log\log n\})$ passes and $\widetilde{O}(n)$ space.
\end{lemma}
\begin{proof}
    It follows by Lemma~\ref{lem:streamls_space}, our choices of $\delta, T$, and Lemma~\ref{lem:xt_iteration}.
\end{proof}

By Lemma~\ref{lem:stream_lsolve} and the reduction in Section~\ref{sec:solver_reduction}, we obtain our main result.
\begin{theorem} \label{thm:stream_lsolve}
    There is a streaming algorithm which takes input an SDD$_0$ matrix $A$, a vector $b \in \mathbb{R}^n$, and a parameter $\epsilon \in (0, 1)$,
    if there exists $x^* \in \mathbb{R}^n$ such that $A x^* = b$,
    then with probability $1 - 1 / \poly(n)$, the algorithm returns an $x \in \mathbb{R}^n$ such that $\|x - x^*\|_A \le \epsilon \cdot \|x^*\|_A$ in $O(\max\{1, \log(1/\epsilon) / \log\log n\})$ passes and $\widetilde{O}(n)$ space.
\end{theorem}

\section{Minimum vertex cover}\label{sec:minimum_vertex_cover}

Given a linear programming form of fractional minimum vertex cover, we design an algorithm that outputs the set of tight constraints of a strictly complementary optimal solution. This set is crucial for us to turn the optimal dual solution into an optimal primal solution via complementary slackness (Theorem~\ref{thm:strong_duality_cs}).

The basic idea is to run interior point method (procedure \textsc{InteriorPointMethod}, Section~\ref{sec:ipm_algorithm}, Algorithm~\ref{alg:ipm}) to get a near-optimal solution. 
However, the hard part is that even if we get a near-optimal solution, we are still far from figuring out all tight constraints of ``some'' optimal solution, e.g., when $\langle x, c \rangle$ is close to $\langle x^*, c \rangle$, $x$ is not necessarily close to $x^*$. The key observation is that our LP is integral and does not have large bit complexity, and that after the isolation step the primal optimum is unique. Then, at a point near the end of the central path, the edges with small slack are the tight constraints of a strictly complementary optimal solution. For the integral minimum vertex cover, we instead use the isolation lemma a second time, on the dual (Lemma~\ref{lem:lemma_43_ls14}). As a byproduct, we also give a formal theorem on solving fractional minimum vertex cover. 

In Section \ref{sec:algorithm_minimum_vertex_cover}, we present out implementation of the minimum vertex cover algorithm. In Section \ref{sec:correctness_minimum_vertex_cover}, we show the correctness of our implementation. In Section \ref{sec:pass_complexity_minimum_vertex_cover}, we also show the complexity of it. In Section \ref{sec:building_blocks}, we explain why Hessian matrix is an SDDM matrix. In Section \ref{sec:minimum_vertex_cover_solver}, we provide our main result of the minimum vertex cover solver.

\ifdefined\debug

\begin{table}[h]
    \centering
    \begin{tabular}{|l|l|l|l|}
        \hline
        {\bf Statement} & {\bf Section} & {\bf Where we use} & {\bf Comments} \\
        \hline
        Lemma~\ref{lem:correct_minimum_vertex_cover} & \ref{sec:correctness_minimum_vertex_cover} & Theorem~\ref{lem:correctness_main} & Main lemma, correctness\\
        \hline    
        Lemma~\ref{lem:time_minimum_vertex_cover} &
        \ref{sec:pass_complexity_minimum_vertex_cover} &  Theorem~\ref{lem:running_time_main} & Main lemma, running time\\
        \hline
        Lemma~\ref{lem:lemma_43_ls14} &
        \ref{sec:correctness_minimum_vertex_cover} &  Theorem~\ref{thm:minimum_vertex_cover} & From nearly-optimal solution to exact solution\\
        \hline
        Lemma~\ref{lem:lambda_min_of_hessian} &
        \ref{sec:correctness_minimum_vertex_cover}& Lemma~\ref{lem:correct_minimum_vertex_cover} & Bound the eigenvalue of Hessian\\
        \hline 
        Lemma~\ref{lem:hessian_sdddm} &
        \ref{sec:building_blocks} & Lemma~\ref{lem:lambda_min_of_hessian}, \ref{lem:time_minimum_vertex_cover} & Proof: Hessian is SDDM\\
        \hline
    \end{tabular}
    \caption{The structure of Section~\ref{sec:minimum_vertex_cover}.}
\end{table}
\else
\fi

\subsection{Our algorithm}\label{sec:algorithm_minimum_vertex_cover}

\begin{algorithm}[h]
\caption{Minimum Vertex Cover. Given the linear programming form of fractional minimum vertex cover, the algorithm outputs the set of tight constraints of a strictly complementary optimal solution, which is the support of the optimal solution of the primal LP when that solution is unique.}
\label{alg:minimum_vertex_cover}
\begin{algorithmic}[1]
\Procedure{\textsc{MinimumVertexCover}}{$G=(V_L,V_R,E),b_1 \in \Z^{m},c_1 \in \Z^n, n,m$} \Comment{Lemma~\ref{lem:correct_minimum_vertex_cover}}
\State Let $A_1$ be the signed edge-vertex incident matrix of $G$ with direction $V_L$ to $V_R$ \label{line:mvc_A_1}
\State $L\leftarrow $ the bit complexity of $A_1,b_1,c_1$ \label{line:mvc_L}
\State Modify $A_1,b_1,c_1$ by adding constraints $0\leq x_{V_L} \leq 2^{L+1}$ and $-2^{L+1}\leq x_{V_R}\leq 0$ to get $A_2\in\Z^{(m+2n)\times n},b_2,c_2$ \label{line:mvc_A_2_b_2_c_2}
\State $A_3\leftarrow A_2$, $b_3\leftarrow b_2$, $c_3\leftarrow c_2$ \label{line:mvc_A_3_b_3_c_3}
\Comment{Theorem~\ref{thm:minimum_vertex_cover} replaces $c_3$ by the objective of Lemma~\ref{lem:lemma_43_ls14}}
\State Let $F(x)$ be logarithmic barrier function \Comment{Def.~\ref{def:log_barrier_function}}
\State $x_{\init} \leftarrow 2^L \cdot (\mathbf{1}_{V_L} - \mathbf{1}_{V_R})$ \label{line:mvc_x_init}
\State $t_{\init} \leftarrow 1$ \label{line:mvc_t_init}
\State $c_{\init} \leftarrow -g(x_{\init})$
\Comment{$g(x) = \nabla F(x)$}\label{line:mvc_c_init}
\State $\epsilon_{\Phi} \leftarrow 1/100$ \label{line:mvc_epsilon_Phi}
\State $t_1 \leftarrow \epsilon_{\Phi} \cdot (m2^{4L+10})^{-1}$ \label{line:mvc_t_1}
\State $x_{\tmp} \leftarrow \textsc{InteriorPointMethod}(A_3,b_3,c_{\init},t_{\init},x_{\init},t_1,0,\epsilon_{\Phi}/4)$ \Comment{Algorithm~\ref{alg:ipm}} \label{line:first_call_to_ipm} \label{line:mvc_x_tmp}
\State $t_2 \leftarrow nm2^{3L+10}$ \label{line:mvc_t_2}
\State $x \leftarrow \textsc{InteriorPointMethod}(A_3,b_3,c_3,t_1,x_{\tmp},t_2,1,\epsilon_{\Phi})$ \Comment{Algorithm~\ref{alg:ipm}} \label{line:second_call_to_ipm} \label{line:mvc_x}
\State In one pass, collect $S\leftarrow\{i\in[m] : t_2\cdot s_i(x)<2\}$, where $s(x)=A_3x-b_3$, and stop as soon as $|S|>n$ \label{line:mvc_xxxx}
\State \Return $S$
\EndProcedure
\end{algorithmic}
\end{algorithm}

{\bf Description.} The constraints of minimum vertex cover is $Ax\geq b,~x\geq 0$ while our IPM algorithm only accepts the form $Ax\geq b$. In order to call IPM, at the beginning we transform $A_1,b_1,c_1$ to $A_2,b_2,c_2$. The upper bounds $2^{L+1}$ make the feasible region bounded and cut off no optimal solution. Then the algorithm executes the standard two IPM walks. The first walk starts from a far point on the central path which is ensured to be feasible and then goes along the central path to the analytic center. Then we switch the central path which correspond to the true LP we want to solve, starting from this analytic center and solve it. Basically, the first walk is to find a good initial point for the second walk. Finally, one pass over the slacks of the output point collects the edges whose slack is less than $2/t_2$.

\subsection{Correctness of Algorithm~\ref{alg:minimum_vertex_cover}}\label{sec:correctness_minimum_vertex_cover}
We prove that Algorithm~\ref{alg:minimum_vertex_cover} can return a set of tight constraints of some optimal solutions to guide us find a maximum weight bipartite matching.

\begin{lemma}[Correctness of Algorithm~\ref{alg:minimum_vertex_cover}]\label{lem:correct_minimum_vertex_cover}
In Algorithm~\ref{alg:minimum_vertex_cover},
let $A=A_1$ be signed edge-vertex incident matrix of $G$ with direction $V_L$ to $V_R$. Let $[n] = V_L \cup V_R$. 
Let $b = b_1\geq\mathbf{0}_m$ and $c = c_1=\mathbf{1}_{V_L}-\mathbf{1}_{V_R}$ be the input, and consider the linear programming
\begin{align}\label{eq:lp_a_1_b_1_c_1}
\min_{x\in \R^n} & ~ c^{\top} x \\
\mathrm{s.t.} & ~ A x \geq b \notag \\
& ~ x_v \geq 0, \forall v \in V_L \notag \\
& ~ x_v \leq 0, \forall v \in V_R \notag
\end{align}
Suppose that the primal LP Eq.~\eqref{eq:primal} with weights $b$ has a unique optimal solution $y^*$ (after replacing $x_v$ by $-x_v$ for $v\in V_R$, Eq.~\eqref{eq:lp_a_1_b_1_c_1} is the dual LP Eq.~\eqref{eq:dual}, whose LP dual is Eq.~\eqref{eq:primal}). If the SDD solves inside the two calls to \textsc{InteriorPointMethod} succeed, which happens with probability $1-1/\poly(n)$ (Theorem~\ref{thm:stream_lsolve}, applied in the proof of Lemma~\ref{lem:time_minimum_vertex_cover}), then Algorithm~\ref{alg:minimum_vertex_cover} returns $S=\mathrm{supp}(y^*)=\{i\in[m] : y^*_i\neq 0\}$, the set of tight constraints of a strictly complementary optimal solution of Eq.~\eqref{eq:lp_a_1_b_1_c_1}.
\end{lemma}
\begin{proof}
First we need to show these two calls into \textsc{InteriorPointMethod} satisfy the initial condition that $x_{\start},t_{\start}$ is a good start point.

{\bf The first call}
In the first call (Line~\ref{line:first_call_to_ipm}), we have
\begin{align*}
\| t_{\init} \cdot c_{\init} + g(x_{\init}) \|_{H(x_{\init})^{-1}} = \|0\|_{H(x_{\init})^{-1}} = 0 \leq \epsilon_{\Phi}/4,
\end{align*}
where the first step is by definition of $c_{\init}$.

Since the algorithm \textsc{InteriorPointMethod} ends up in parameter $x_{\tmp}$ and $t_1$, by Lemma~\ref{lem:bound_phi}, we have
\begin{align}\label{eq:first_call}
\| t_1 \cdot c_{\init} + g(x_{\tmp})\|_{H(x_{\tmp})^{-1}} \leq \epsilon_{\Phi}/4.
\end{align}
Note that $t_1$ and $x_{\tmp}$ is the input to the second call.

{\bf The second call}
In the second call (Line~\ref{line:second_call_to_ipm}), the desired term can be upper bounded by
\begin{align*}
    & ~ \| t_1 \cdot c_3 + g(x_{\tmp}) \|_{H(x_{\tmp})^{-1}} \\
    \leq & ~ \| t_1 \cdot c_{\init} + g(x_{\tmp})\|_{H(x_{\tmp})^{-1}} + \| t_1 \cdot c_{\init} - t_1 \cdot c_3 \|_{H(x_{\tmp})^{-1}}\\
    \leq & ~ \epsilon_{\Phi}/4 + t_1 \cdot \| c_{\init} - c_3 \|_{H(x_{\tmp})^{-1}}
\end{align*}
where the first step is by triangle inequality, the second step is from Eq.~\eqref{eq:first_call}.

Now let's bound the term $\| c_{\init} - c_3 \|_{H(x_{\tmp})^{-1}}$.
\begin{align*}
    \| c_{\init} - c_3 \|_{H(x_{\tmp})^{-1}}
    \leq & ~ \|c_{\init} - c_3\|_2 \cdot \lambda_{\min}(H(x_{\tmp}))^{-1/2}\\
    \leq & ~ \|c_{\init} - c_3\|_2 \cdot 2^{L+3} \\
    \leq & ~ (\|c_{\init}\|_2 + \|c_3\|_2) \cdot 2^{L+3} \\
    \leq & ~ (\sqrt{m} 2^{L+3} + \|c_3\|_2)\cdot 2^{L+3}\\
    \leq & ~ (\sqrt{m} 2^{L+3} + m2^{3L+4})\cdot 2^{L+3}\\
    \leq & ~ m 2^{4L+8},
\end{align*}
where the first step is by $\|x\|_{H} = \sqrt{x^{\top} H x}\leq \|x\|_2 \cdot \sqrt{\lambda_{\max}(H)}$ and $\lambda_{\max}(H) = \lambda_{\min}(H^{-1})^{-1}$, the second step is by Lemma~\ref{lem:lambda_min_of_hessian}, the fourth step is by $c_{\init} = -g(x_{\init} )$, $x_{\init} = 2^L (\mathbf{1}_{V_L} - \mathbf{1}_{V_R})$ and the definition of $g = \nabla \phi(x)$ (Definition~\ref{def:log_barrier_function}), the fifth step is by $\|c_3\|_2=\|c_2\|_2=\sqrt{n}\leq m2^{3L+4}$.

By our choice of $t_1 := \epsilon_{\Phi} \cdot (m2^{4L+10})^{-1}$ (Line~\ref{line:mvc_t_1}), we finally have
\begin{align*} 
\| t_1 \cdot c_3 + g(x_{\tmp}) \|_{H(x_{\tmp})^{-1}} 
\leq & ~ \epsilon_{\Phi}/4 + \epsilon_{\Phi}/4 \\
\leq & ~ \epsilon_{\Phi}.
\end{align*} 

Let $m':=m+2n$ be the number of rows of $A_3$. By Lemma~\ref{lem:bound_phi}, the output $x$ of the second call is feasible and satisfies $\Phi_t(x,x)\leq\epsilon_{\Phi}$, where $t\in[t_2,(1+\epsilon_t)t_2]$ is the final value of $t$ in that call.

We first describe the LP dual of $\min\{c_3^{\top}x : A_3x\geq b_3\}$. Every optimal solution of Eq.~\eqref{eq:lp_a_1_b_1_c_1} satisfies $|x_v|\leq \|b\|_\infty<2^{L}$ for all $v\in[n]$, since otherwise moving $x_v$ toward $0$ keeps feasibility and decreases the objective. So the box constraints of Line~\ref{line:mvc_A_2_b_2_c_2} cut off no optimal solution and are not tight at any optimal solution. Hence LP $(A_3,b_3,c_3)$ has the same optimal value and the same optimal solutions as Eq.~\eqref{eq:lp_a_1_b_1_c_1}. Its LP dual is $\max\{b_3^{\top}\hat y : A_3^{\top}\hat y=c_3,\ \hat y\geq \mathbf{0}_{m'}\}$. Writing $\hat y=(y,z,u)$ with $y\in\R^m$ for the edge rows, $z\in\R^n$ for the sign rows and $u\in\R^n$ for the upper-bound rows, it reads
\begin{align*}
\max_{y,z,u\geq 0} ~ b^{\top}y-2^{L+1}\mathbf{1}_n^{\top}u ~~~ \mathrm{s.t.} ~ A_{\rm u}^{\top}y+z-u=\mathbf{1}_n,
\end{align*}
where $A_{\rm u}$ is the unsigned edge-vertex incidence matrix of $G$. By complementary slackness, every optimal solution has $u=\mathbf{0}_n$, so the optimal solutions are $(y,\mathbf{1}_n-A_{\rm u}^{\top}y,\mathbf{0}_n)$ with $y$ optimal for Eq.~\eqref{eq:primal}. Hence $\hat y^*:=(y^*,\mathbf{1}_n-A_{\rm u}^{\top}y^*,\mathbf{0}_n)$ is the unique optimal solution. Here $y^*\in\{0,1\}^m$ is the indicator vector of a matching, since it is a vertex of the feasible region of Eq.~\eqref{eq:primal} (Lemma~\ref{lem:totally_unimodular}). The matrix $A_3^{\top}$ is totally unimodular, so every vertex of the dual feasible region is integral with entries bounded by $\|c_3\|_1=n$ in absolute value. Every vertex other than $\hat y^*$ therefore has objective value at most $\OPT-1$, where $\OPT$ is the optimal value. Every recession direction $\hat r=(y_r,z_r,u_r)$ of the dual feasible region satisfies $u_r=A_{\rm u}^{\top}y_r+z_r$, so its objective value is at most $\|b\|_\infty\mathbf{1}_m^{\top}y_r-2^{L+1}(2\cdot\mathbf{1}_m^{\top}y_r+\mathbf{1}_n^{\top}z_r)\leq -\|\hat r\|_1$.

Next we turn $x$ into a nearly optimal dual solution. Let $s:=s(x)$, $S_x:=\diag(s)$, $H:=H(x)=A_3^{\top}S_x^{-2}A_3$ and $\eta:=-H^{-1}(t c_3+g(x))$, so that $\|\eta\|_{H}=\Phi_t(x,x)\leq\epsilon_{\Phi}$. Define $\hat y:=\frac{1}{t}S_x^{-1}(\mathbf{1}_{m'}-S_x^{-1}A_3\eta)$. Since $g(x)=-A_3^{\top}S_x^{-1}\mathbf{1}_{m'}$, we have $A_3^{\top}\hat y=\frac{1}{t}(-g(x)-H\eta)=c_3$. Since $\|S_x^{-1}A_3\eta\|_\infty\leq\|S_x^{-1}A_3\eta\|_2=\|\eta\|_{H}\leq\epsilon_{\Phi}$, we have $\hat y\geq \mathbf{0}_{m'}$ and
\begin{align}\label{eq:hat_y_vs_slack}
\frac{1-\epsilon_{\Phi}}{t\cdot s_i}\leq\hat y_i\leq\frac{1+\epsilon_{\Phi}}{t\cdot s_i},~~\forall i\in[m'].
\end{align}
So $\hat y$ is dual feasible, and $\OPT-b_3^{\top}\hat y\leq c_3^{\top}x-b_3^{\top}\hat y=\hat y^{\top}s\leq (1+\epsilon_{\Phi})m'/t\leq 2m'/t_2$.

Write $\hat y=\sum_j\lambda_j v_j+\hat r$, where the $v_j$ are vertices, $\lambda_j\geq0$, $\sum_j\lambda_j=1$, and $\hat r$ is a recession direction. Then $\OPT-b_3^{\top}\hat y\geq\sum_{j: v_j\neq\hat y^*}\lambda_j+\|\hat r\|_1$, so $\sum_{j:v_j\neq \hat y^*}\lambda_j\leq 2m'/t_2$ and $\|\hat r\|_1\leq 2m'/t_2$. Since $\|v_j-\hat y^*\|_\infty\leq n+1$, we get
\begin{align*}
\|\hat y-\hat y^*\|_\infty\leq (n+2)\cdot 2m'/t_2\leq 1/4,
\end{align*}
where the last step uses $t_2=nm2^{3L+10}$ and $n\leq 2m$ (we may assume that $G$ has no isolated vertex). So for every edge $i\in[m]$, $\hat y_i\geq 3/4$ if $y^*_i=1$ and $\hat y_i\leq 1/4$ if $y^*_i=0$. By Eq.~\eqref{eq:hat_y_vs_slack}, $t\leq (1+\epsilon_t)t_2$ and $\epsilon_t\leq 1/2$, we get $t_2\cdot s_i(x)\leq t\cdot s_i(x)\leq 4(1+\epsilon_{\Phi})/3<2$ if $y^*_i=1$, and $t_2\cdot s_i(x)\geq 4(1-\epsilon_{\Phi})/(1+\epsilon_t)>2$ if $y^*_i=0$. Hence Line~\ref{line:mvc_xxxx} collects exactly $S=\mathrm{supp}(y^*)$. This set has at most $n/2$ elements, so the algorithm does not stop early. By the Goldman--Tucker theorem~\cite{gt56}, Eq.~\eqref{eq:lp_a_1_b_1_c_1} has an optimal solution that is strictly complementary to $y^*$. The set $S$ is exactly the set of tight edge constraints of this solution.

\end{proof}

Next, we provide how to use the perturbed linear programming to approximate the solution of the linear programming.
\begin{lemma}[Lemma 43 of \cite{ls14}]\label{lem:lemma_43_ls14}
    Given a feasible and bounded linear programming
    \begin{align}\label{eq:lemma_43_ls14_before}
        & ~ \min_{x\in\R^n} c^{\top} x\\
        \mathrm{s.t.~} & ~ Ax\geq b \notag,
    \end{align}
    where $A\in \Z^{m\times n}$, $b\in \Z^m$, $c\in \Z^n$ all having integer coefficients.
    Let $L$ be the bit complexity of Eq.~\eqref{eq:lemma_43_ls14_before}.
    
    Let $r\in \Z^{n}$ be chosen uniformly at random from the integers $\{-2^{L+1}n,\cdots,2^{L+1}n\}$.
    Consider the perturbed linear programming
    \begin{align}\label{eq:lemma_43_ls14_after}
        & ~ \min_{x\in\R^n} (2^{2L+3}n\cdot c + r)^{\top} x\\
        \mathrm{s.t.~} & ~ Ax\geq b \notag.
    \end{align}
    Then with probability at least $1/2$ over the randomness on $r\in \Z^n$, we have the following. 
    
    Let $\OPT$ be defined as the optimal value of linear programming Eq.~\eqref{eq:lemma_43_ls14_after}. 
    Let $x$ be any feasible solution for Eq.~\eqref{eq:lemma_43_ls14_after} with objective value less than $\OPT + n^{-1}2^{-3L-2}$, then we can find the tight constraints of a basic feasible optimal solution of Eq.~\eqref{eq:lemma_43_ls14_before} using one matrix vector multiplication with $A$. 
    Moreover, we have $\|x-x^*\|_{\infty} \leq 1/n$, where $x^*$ is the unique optimal solution for Eq.~\eqref{eq:lemma_43_ls14_after}. Additionally, the bit complexity of Eq.~\eqref{eq:lemma_43_ls14_after} is at most $3L + \log(8n)$. 
\end{lemma}

Next, we show the lower bound of $\lambda_{\min}(H(x)) $.

\begin{lemma}\label{lem:lambda_min_of_hessian}
Let $H(x)$ be defined as in Lemma~\ref{lem:hessian_sdddm}. Then for every feasible $x$, we have
\begin{align*}
    \lambda_{\min} (H(x)) \geq 2^{-2L-6}.
\end{align*}
\end{lemma}
\begin{proof}
We can lower bound $\lambda_{\min}(H(x))$ in the following sense,
\begin{align*}
 \lambda_{\min} (H(x)) 
 \geq & ~ \lambda_{\min} (D(x)) \\
\geq & ~ \min_{i \in [n]} D(x)_{i,i} \\
\geq & ~ 2^{-2L-6}
\end{align*}
where the first step is by Lemma~\ref{lem:hessian_sdddm} that $H(x) = L(x) + D(x)$ and both $L(x)$ and $D(x)$ are positive semi-definite matrices, the second step is by the fact that $D(x)$ is diagonal matrix, the third step is by Lemma~\ref{lem:hessian_sdddm} that $D(x)_{i,i} = s_{i+m}(x)^{-2} + s_{i+m+n}(x)^{-2}$, where $s_{i+m}(x)+s_{i+m+n}(x)=2^{L+1}$ by the box constraints of Line~\ref{line:mvc_A_2_b_2_c_2}, so one of the two slacks is at most $2^{L}$.
\end{proof}

\subsection{Pass complexity of Algorithm~\ref{alg:minimum_vertex_cover}}\label{sec:pass_complexity_minimum_vertex_cover}

We show that Algorithm~\ref{alg:minimum_vertex_cover} takes $\wt O(\sqrt{m})$ passes and $\wt O(n)$ space.
\begin{lemma}[Pass complexity of Algorithm~\ref{alg:minimum_vertex_cover}]\label{lem:time_minimum_vertex_cover}
Suppose there's an oracle $\mathcal{I}$ running in $f(n)$ space that can output $(b_1)_i$ given any $i\in[m]$. If the input satisfies $\|b\|_{\infty}, \|c\|_{\infty}$ is polynomially bounded, then Algorithm~\ref{alg:minimum_vertex_cover} can be implemented in the streaming model within $\wt{O}(n) + f(n) + |S|$ space and $\wt{O}(\sqrt{m})$ passes.
\end{lemma}
\begin{proof}
In Line~\ref{line:mvc_A_1}, we defined $A_1$ but never explicitly compute and store $A_1$ in memory.

In Line~\ref{line:mvc_L}, we let $L \leftarrow O(\log n)$ as the upper bound on the bit complexity. This cost one pass. Indeed, Lemma~\ref{lem:totally_unimodular} implies that $|\dmax(A_1)|\leq 1$.
So $L$ is the upper bound on the bit complexity.

In Line~\ref{line:mvc_A_2_b_2_c_2} and Line~\ref{line:mvc_A_3_b_3_c_3}, we defined $A_2,b_2,c_2,A_3,b_3,c_3$. This doesn't involve computation.

In Lines~\ref{line:mvc_x_init}, \ref{line:mvc_t_init}, \ref{line:mvc_c_init}, \ref{line:mvc_epsilon_Phi}, \ref{line:mvc_t_1}, \ref{line:mvc_t_2}, they are done locally in memory.

In the rest of the proof, in Line~\ref{line:mvc_x_tmp} and Line~\ref{line:mvc_x}, we will calculate the pass complexity and space usage for calling IPM in graph input setting.

In Algorithm~\ref{alg:ipm}, we let 
\begin{align*}
\delta_x=\delta_x' = H(x)^{-1}\cdot \nabla f_t(x)
\end{align*} 
and 
\begin{align*} 
\wt{\delta}_x = \textsc{StreamLS}(H(x),\nabla f_t(x),10^{-5}).
\end{align*}

First, by Definition~\ref{def:perturbed_objective_function}, 
\begin{align*} 
\nabla f_t(x) = c\cdot t - \sum_{i\in [m]} a_i / s_i(x).
\end{align*}
By assumption, using $f(n)$ space, when given $i$, we can output $b_i$. Thus $\nabla f_t(x)$ can be computed exactly in one pass and $n+f(n)$ space.

Second, by Lemma~\ref{lem:hessian_sdddm}, the Hessian $H(x)\in \R^{n\times n}$  can be decomposed into 
\begin{align*} 
H(x) = L_G(x) + D(x),
\end{align*}
where $L_G(x) \in \R^{n\times n}$ is the Laplacian of some graph $G$ with edge weight $w_i=s_i(x)^{-2}$ and $D(x)\in \R^{n\times n}$ is diagonal matrix where
\begin{align*} 
D(x)_{i,i} = s_{i+m}(x)^{-2} + s_{i+m+n}^{-2}.
\end{align*}

Therefore, we are able to first read $D$ in one pass and store its entries in space $O(n)$, and then apply Theorem~\ref{thm:stream_lsolve} 
to show $\wt{\delta}_x$ can be computed in $\wt{O}(1)$ pass and $\wt{O}(n)$ space, with 
\begin{align*} 
\|\wt{\delta}_x - \delta_x\|_{H(x)} \leq \epsilon_x = 10^{-5}
\end{align*}

Since there are totally $T=O(\sqrt{m}\log(m/\epsilon_{\ipm}))$ iterations, the totally number of passes used is $\wt{O}(\sqrt{m})$, the space usage is $\wt{O}(n) + f(n)$.

And Line~\ref{line:mvc_xxxx} takes one pass and stores $|S|\leq n+1$ indices.

\end{proof}

\subsection{Building blocks}\label{sec:building_blocks}

Next, we show that the Hessian matrix is an SDDM matrix.
\begin{lemma}[Hessian is SDDM matrix] \label{lem:hessian_sdddm}
Let the graph $G=(V,E)$ be the input of Algorithm~\ref{alg:minimum_vertex_cover}.
Let $n=|V|$, $m=|E|$.
Let $A_3\in \R^{(m+2n)\times n}$, $b_3\in \R^{m+2n}$ be defined as in Line~\ref{line:mvc_A_3_b_3_c_3} (Algorithm~\ref{alg:minimum_vertex_cover}).

Let $x\in \R^{n}$ be any feasible point, i.e. $A_3 x > b_3$. Let $s(x)\in \R^{m+2n}$ and $H(x) \in \R^{n\times n}$ be the slack and hessian defined w.r.t. $A_3$ and $b_3$. Then $H(x) \in \R^{n\times n}$ can be written as
\begin{align*}
H(x) = L(x) + D(x),
\end{align*}
where $L(x) \in \R^{n\times n}$ is the Laplacian matrix of graph $G$ with edge weight $\{s_1(x)^{-2},\cdots,s_m(x)^{-2}\}$, $D(x) \in \R^{n\times n}$ is the diagonal matrix with each diagonal entry $D(x)_{v,v} = s_{v+m}(x)^{-2}+s_{v+m+n}(x)^{-2}$ for $v\in[n]$, where rows $m+v$ and $m+n+v$ of $A_3$ are the two box constraints of vertex $v$.
\end{lemma}
\begin{proof}
Let $A_1\in \R^{m\times n}$ be signed edge-vertex incident matrix of $G$ with direction $V_L$ to $V_R$, then $A_3\cdot x\geq b_3$ can be written as
\begin{align*}
\begin{bmatrix}
A_1\\
I_{V_L} - I_{V_R}\\
-I_{V_L} + I_{V_R}
\end{bmatrix}
\cdot 
x \geq
\begin{bmatrix}
b_1\\
\mathbf{0}_n\\
-2^{L+1}\mathbf{1}_n
\end{bmatrix},
\end{align*}
where $I_{V_L}\in \R^{n\times n}$ denotes the diagonal matrix with 
\begin{align*}
    I_{i,i} = 
    \begin{cases}
    1, & \mathrm{~if~} i\in V_L;\\
    0, & \mathrm{~otherwise}.
    \end{cases}
\end{align*}
and
$I_{V_R}\in \R^{n\times n}$ denotes the diagonal matrix with 
\begin{align*}
    I_{i,i} = 
    \begin{cases}
    1, & \mathrm{~if~} i\in V_R;\\
    0, & \mathrm{~otherwise}.
    \end{cases}
\end{align*}

Denote $S\in \R^{m\times m}$ the diagonal matrix with each diagonal entry $S_{i,i}:= s_i(x)$.
Thus $H(x)$ can be written as
\begin{align*}
    H(x) 
    = & ~ A_1^{\top}S^{-2}A_1 + \sum_{v\in[n]} e_ve_v^{\top} \cdot (s_{v+m}(x)^{-2}+s_{v+m+n}(x)^{-2})\\
    = & ~ L(x) + D(x),
\end{align*}
where the first step is by definition of $H(x)$ (Definition~\ref{def:log_barrier_function}), the second step is by the definition of $L(x)$ and $D(x)$.
\end{proof}

\subsection{A minimum vertex cover solver}\label{sec:minimum_vertex_cover_solver}
Note that the two calls to \textsc{InteriorPointMethod} in Algorithm~\ref{alg:minimum_vertex_cover} give a high-accuracy fractional minimum vertex cover solver, also for general graphs (not necessarily bipartite), if we use the unsigned incidence matrix.
\begin{theorem}\label{thm:frac_minimum_vertex_cover}
    Let $G$ be a graph with $n$ vertices and $m$ edges. Consider the \emph{fractional minimum vertex cover} problem Eq.~\eqref{eq:dual} (with $A$ the unsigned edge-vertex incidence matrix of $G$) in which every edge $e$ needs to be covered at least $b_e$ times. 
    Let $x^*\in \R^n$ be the optimal solution of Eq.~\eqref{eq:dual}. Suppose $b\in\{1,\cdots,W\}^m$ with $W\leq\poly(n)$. There exists a streaming algorithm (the two calls to \textsc{InteriorPointMethod} in Algorithm~\ref{alg:minimum_vertex_cover} with $t_2:=3(m+2n)/\delta$, where for non-bipartite $G$ we use the unsigned incidence matrix with constraints $0\leq x\leq 2^{L+1}$, whose Hessian is SDD and is handled by Theorem~\ref{thm:stream_lsolve}) such that for any $\delta > 0$, it outputs a feasible vertex cover $x\in \R^n$ such that
    \[
        \mathbf{1}_n^{\top} x \leq \mathbf{1}_n^{\top}x^* + \delta
    \]
    in $\wt{O}(\sqrt{m})\cdot \log (1/\delta)$ passes and $\wt{O}(n) \cdot \log (1/\delta)$ space with probability $1-1/\poly(n)$.
\end{theorem}
\begin{proof}
    The first call and the start of the second call are analyzed as in the proof of Lemma~\ref{lem:correct_minimum_vertex_cover}, and the passes and the space as in Lemma~\ref{lem:time_minimum_vertex_cover}. The box constraints cut off no optimal solution, and the logarithmic barrier of the $m+2n$ constraints is $\theta^2$-self-concordant with $\theta^2=m+2n$. So Lemma~\ref{lem:ipm_error} gives $\mathbf{1}_n^{\top}x-\mathbf{1}_n^{\top}x^*\leq (m+2n)(1+2\epsilon_{\Phi})/t_2<\delta$.
\end{proof}

As a byproduct, we obtain a fast semi-streaming algorithm for (exact, integral) minimum vertex cover in bipartite graph.

\begin{theorem}\label{thm:minimum_vertex_cover}
    Given a bipartite graph $G$ with $n$ vertices and $m$ edges, there exists a streaming algorithm that computes a minimum vertex cover of $G$ in $\wt{O}(\sqrt{m})$ passes and $\wt{O}(n)$ space with probability $1-1/\poly(n)$.
\end{theorem}

\begin{proof}
    Because $G$ is bipartite, by Theorem~\ref{lem:totally_unimodular}, all extreme points of the polytope of LP~\eqref{eq:dual} are integral. The same holds after adding the box constraints of Algorithm~\ref{alg:minimum_vertex_cover}, which cut off no optimal solution (see the proof of Lemma~\ref{lem:correct_minimum_vertex_cover}). 
    We run the two calls to \textsc{InteriorPointMethod} of Algorithm~\ref{alg:minimum_vertex_cover} with $b=\mathbf{1}_m$ and with $c_3$ replaced by the perturbed objective of Lemma~\ref{lem:lemma_43_ls14} for the LP of $(A_2,b_2,c_2)$. The perturbed LP has bit complexity $O(L)$ and $\|c_3\|_2\leq 2^{O(L)}$, so we choose $t_1=2^{-O(L)}$ and $t_2=2^{O(L)}$ as in the proof of Lemma~\ref{lem:correct_minimum_vertex_cover}. Then the output $x$ is within $n^{-1}2^{-3L'-2}$ of the optimal value of the perturbed LP, where $L'$ is the bit complexity of the LP of $(A_2,b_2,c_2)$.
    Since LP~\eqref{eq:dual} is feasible and bounded, the high-accuracy solution obtained from IPM can be rounded coordinatewise to the optimal integral solution with probability at least $1/2$ (by $\|x-x^*\|_\infty\leq 1/n$ in Lemma~\ref{lem:lemma_43_ls14} and integrality of $x^*$). Repeating with $O(\log n)$ independent choices of $r$ and keeping the feasible rounded vector of smallest objective (one pass each to check feasibility) boosts this to $1-1/\poly(n)$, which in total takes $\wt{O}(\sqrt{m})$ passes and $\wt{O}(n)$ space with probability $1 - 1 / \poly(n)$.
\end{proof} 

\section{Maximum weight bipartite matching}\label{sec:combine}
In this section, we give our main algorithm that combines all previous subroutines to give our final Theorem~\ref{thm:main_theorem}. The proof consists of the correctness part and the pass complexity part. We give a roadmap of this section as follows. In Section~\ref{sec:algorithm_main}, we present the main algorithm. In Section~\ref{sec:running_time_main}, we prove the pass complexity. The correctness is a bit complex. In Section~\ref{sec:combine_primal_to_dual}, we give preliminary knowledge on primal-to-dual transformation, and then in section~\ref{sec:properties_of_lp_solution}, combined with isolation lemma, we show how to isolate a maximum weight matching from an optimal solution of the dual. In Section~\ref{sec:correctness_main}, we conclude the correctness part.

\begin{theorem}[Main theorem, formal version of Theorem~\ref{thm:main_theorem_in_intro}]\label{thm:main_theorem}
    Given a bipartite graph $G$ with $n$ vertices and $m$ edges, with integer edge weights $w\in\{0,1,\cdots,W\}^m$ where $W\leq \poly(n)$, there exists a streaming algorithm that computes a maximum weighted matching of $G$ in $\wt{O}(\sqrt{m})$ passes and $\wt{O}(n)$ space with probability $1 - 1 / \poly(n)$.
\end{theorem}
\begin{proof}
    By combining Lemma~\ref{lem:running_time_main} and Lemma~\ref{lem:correctness_main}.
\end{proof}

\ifdefined\debug
\begin{table}[ht]
    \centering
    \begin{tabular}{|l|l|l|l|}
        \hline
        {\bf Statement} & {\bf Section} & {\bf Where we use} & {\bf Comments} \\
        \hline
        Theorem~\ref{thm:main_theorem} & \ref{sec:combine} & / & Main theorem of this section\\ 
        \hline
        Lemma~\ref{lem:running_time_main} & \ref{sec:running_time_main} & Theorem~\ref{thm:main_theorem} & Pass complexity \\
        \hline
        Lemma~\ref{lem:correctness_main} & \ref{sec:correctness_main} & Theorem~\ref{thm:main_theorem} & Correctness \\
        \hline
        Theorem~\ref{thm:strong_duality_cs} & \ref{sec:combine_primal_to_dual} & Lemma~\ref{lem:properties_of_lp_solution} & Preliminary on proof of Lemma~\ref{lem:correctness_main} \\
        \hline
        Lemma~\ref{lem:properties_of_lp_solution} & \ref{sec:properties_of_lp_solution} & Lemma~\ref{lem:correctness_main} & Properties of primal and dual solutions\\
        \hline
    \end{tabular}
    \caption{The structure of Section~\ref{sec:combine}.}
    \label{tab:my_label}
\end{table}
\fi

\subsection{Algorithms}\label{sec:algorithm_main}

Algorithm~\ref{alg:main} is our main algorithm. The outer loop uses $O(\log n)$ calls to boost the success probability to $1-1/\poly(n)$. In each call, it first prepares the isolation oracle from Section~\ref{sec:isolation_lemma} and then passes it to the minimum vertex cover solver from Section~\ref{sec:minimum_vertex_cover}. The isolation step succeeds with probability at least $1/4$. If it succeeds, the solver returns the optimal matching with probability $1-1/\poly(n)$. We choose the maximum weighted matching over all $O(\log n)$ tries as the final answer.

\begin{algorithm}[ht]
\caption{
We present our main algorithm. It takes vertices $V_L,V_R$, edges $E$ and weights $w$ as input, and eventually it will output a maximum weighted matching $M$.}
\label{alg:main}
\begin{algorithmic}[1]
\Procedure{$\textsc{Main}$}{$G=(V_L,V_R,E,w)$} \Comment{Theorem~\ref{thm:main_theorem}}
\State \Comment{$w$ denotes edge weights that are integers.}
\State $n\leftarrow |V_L| + |V_R|$,~$m\leftarrow |E|$
\State $M \leftarrow \emptyset$ \Comment{$M$ is maximum matching}
\For{$i = 1 \to O(\log n)$} \label{line:main_for_loop}
    \State $\ov{b}\leftarrow \textsc{Isolation}(m,n^n)$ \Comment{$\ov{b}\in \Z^m$, Algorithm~\ref{alg:isolation_lemma}, Lemma~\ref{lem:isolation_lemma}} \label{line:main_ov_b}
    \State $b\leftarrow \ov{b} + n^{15}\cdot w$\label{line:main_b}
    \State $c\leftarrow \mathbf{1}_{V_L} - \mathbf{1}_{V_R}$ \label{line:main_c}
    \State $S \leftarrow \textsc{MinimumVertexCover}(G,b,c,n,m)$  \Comment{$S \subseteq[m]$, Algorithm~\ref{alg:minimum_vertex_cover}} \label{line:main_S}
    \If{$|S|\leq n$} \label{line:main_if_S_leq_n} 
        \State Let $M'$ be maximum weight matching found in edge set $S$ \label{line:main_find_matching}
        \If{$w(M') > w(M)$} \Comment{Update maximum weight matching}
            \State $M\leftarrow M'$
        \EndIf
    \EndIf
\EndFor
\State \Return $M$
\EndProcedure
\end{algorithmic}
\end{algorithm}

\subsection{Pass complexity}\label{sec:running_time_main}
The goal of this section is to bound the pass number (Lemma~\ref{lem:running_time_main}).

\begin{lemma}[Pass complexity]\label{lem:running_time_main}
    Given a bipartite graph with $n$ vertices and $m$ edges, with integer edge weights $w\in\{0,1,\cdots,W\}^m$ where $W\leq \poly(n)$, Algorithm~\ref{alg:main} can be implemented such that it runs in $\wt{O}(\sqrt{m})$ passes in streaming model in $\wt{O}(n)$ space.
\end{lemma}
\begin{proof}
In Line~\ref{line:main_ov_b}, \ref{line:main_b}, we actually do not explicitly calculate $\ov{b}$ and $b$ and stored them in memory. We instead use an oracle $\mathcal{I}$ stated in Lemma~\ref{lem:isolation_main_lemma} that gives $b_i$ bits by bits. So this step does not cost space.

In Line~\ref{line:main_c}, we calculate and store $c\in \Z^n$ in $\wt{O}(n)$ space.

In Line~\ref{line:main_S}, we call $\textsc{MinimuVertexCover}$. In order to use Lemma~\ref{lem:time_minimum_vertex_cover}, we need to prove the following properties.

1. By Lemma~\ref{lem:isolation_lemma}, $\|\ov{b}\|_{\infty}\leq n^{14}$, so $\|b\|_{\infty}\leq n^{14}+n^{15}\|w\|_{\infty}\leq \poly(n)$ since $\|w\|_{\infty}\leq W\leq \poly(n)$. Here the isolation lemma is applied on the ground set $[m]$ with $Z=n^n$, so $\|\ov{b}\|_\infty\leq t^2m^5$ with $t\leq 5+4n\log n$, and $m\leq n^2/4$. And $\|c\|_{\infty} = 1$.  So both $\|b\|_{\infty}$ and $\|c\|_{\infty}$ is polynomially bounded.

2. By Lemma~\ref{lem:isolation_main_lemma}, there is an oracle $\mathcal{I}$ that uses $O(\log(n^n) +\log(m))=\wt{O}(n)$ space that can output $\ov{b}_i$. Since each edge $e_i$ comes with its weight $w_i$ in the stream, we can output $b_i = \ov{b}_i + n^{15}w_i$ when given $i\in[m]$.

3. Line~\ref{line:mvc_xxxx} of Algorithm~\ref{alg:minimum_vertex_cover} stops as soon as more than $n$ tight edges are found, so $|S|\leq n+1$ and $S$ is stored in $O(n)$ space.

By applying Lemma~\ref{lem:time_minimum_vertex_cover}, this call can be done in $\wt{O}(n)$ space and $\wt{O}(\sqrt{m})$ passes.

In Line~\ref{line:main_find_matching}, since we are finding maximum matching in a graph with $n$ vertices and $n$ edges, we can store them in memory, then run a polynomial-time maximum weight bipartite matching algorithm (e.g., the Hungarian algorithm with adjacency lists) on the stored graph $(V_L\cup V_R, S)$, which has at most $n$ edges. 
This cost $\wt{O}(n)$ space without any pass. 
We find the set of edges that is a matching and has the maximum weight.

With $O(\log n)$ iterations overhead (Line~\ref{line:main_for_loop}), the number of passes blow up by an $O(\log n)$ factor. So overall, we used $\wt{O}(n)$ space and $\wt{O}(\sqrt{m})$ passes.
\end{proof}

\subsection{Correctness}\label{sec:correctness_main}
The goal of this section is to prove the correctness of our algorithm (Lemma~\ref{lem:correctness_main}).

\begin{lemma}[Correctness]\label{lem:correctness_main}
    Given a bipartite graph $G$ with $n$ vertices and $m$ edges. With probability at least $1-1/\poly(n)$, Algorithm~\ref{alg:main} outputs one maximum matching.
\end{lemma}
\begin{proof}
Consider the linear programming Eq.~\eqref{eq:primal}. According to part 1 of Lemma~\ref{lem:properties_of_lp_solution}, with probability at least $1/4$, there is a unique solution $y^*$ to Eq.~\eqref{eq:primal} with the weights $b$ of Line~\ref{line:main_b}, and $y^*$ is the indicator vector of a maximum weight matching $M^*$ of $G$ with respect to $w$.

By Lemma~\ref{lem:correct_minimum_vertex_cover}, in this case, with probability $1-1/\poly(n)$ the algorithm \textsc{MinimumVertexCover} returns $S=\mathrm{supp}(y^*)=M^*$, which is the set of tight constraints of a strictly complementary optimal solution $x^*,s^*$ of the dual problem Eq.~\eqref{eq:dual} (note that the linear programming in Lemma~\ref{lem:correct_minimum_vertex_cover} and Eq.~\eqref{eq:dual} only differ in signs).
By part 2 of Lemma~\ref{lem:properties_of_lp_solution}, $|S|\leq n/2$, and $S$ contains the maximum weight matching $M^*$, which we find in Line~\ref{line:main_find_matching}. In every iteration, $M'$ is a matching of $G$, so the output $M$ is always a matching.

Overall, in each iteration of the for loop (Line~\ref{line:main_for_loop}), with probability at least $(1/4)\cdot(1-1/\poly(n))\geq 1/8$ we can find a maximum weight matching. After $O(\log n)$ loops, we can find a maximum matching with probability at least $1-1/\poly(n)$.
\end{proof}

\subsection{Primal to dual}\label{sec:combine_primal_to_dual}
This section provides a preliminary overview for proving the correctness of Lemma~\ref{lem:correctness_main}. We present Theorem~\ref{thm:strong_duality_cs}, which is later used in Lemma~\ref{lem:properties_of_lp_solution}.

\begin{definition}[Maximum weighted matching]\label{def:primal}
Given a bipartite graph $G=(V,E)$ with $|V|=n$ and $|E|=m$.
Let $A \in \{0,1\}^{m \times n}$ be the unsigned edge-vertex incident matrix. Given weight $b \in \Z^m$ on every edge, the maximum weighted matching can be written as the following linear programming:
\begin{align}\label{eq:primal}
{\bf Primal}~~~
\max_{y \in \R^m} & ~ b^{\top} y \\
\mathrm{~s.t.~} & ~ A^{\top}y \leq \mathbf{1}_n \notag \\
& ~ y\geq 0 \notag 
\end{align}
\end{definition}
Its dual form is
\begin{definition}[Fractional minimum vertex cover]\label{def:dual}
Let $A \in \Z^{m \times n}$, $b \in \Z^m$ be defined as in Definition~\ref{def:primal}. The dual form of Eq.~\eqref{eq:primal} is\footnote{The dual LP is a \emph{generalized} version of the minimum vertex cover problem: each edge $i$ needs to be \emph{covered} by at least $b_i$ times, where the case of $b = \mathbf{1}_m$ is the classic minimum vertex cover.}
\begin{align}\label{eq:dual}
{\bf Dual}~~~
\min_{x \in \R^n} & ~ \mathbf{1}_n^{\top}  x\\
\mathrm{s.t.} & ~ A x \ge b \notag \\
 & ~ x\geq 0 \notag
\end{align}
\end{definition}

Next, we show the strong duality from complementary slackness.
\begin{theorem}[Strong duality from complementary slackness \cite{ps98}] \label{thm:strong_duality_cs}
    Let $y \in \R^m$ be a feasible solution to the primal Eq.~\eqref{eq:primal}, and let $x \in \R^n$ be a feasible solution to the dual Eq.~\eqref{eq:dual}.
    Let $s := Ax - b \in \R^m$. 
    Then $x \in \R^n, s \in \R^m, y \in \R^m$ satisfy 
    \begin{equation*}
        y^{\top} s = 0 \text{~~and~~} x^{\top} (\mathbf{1}_n - A^{\top} y) = 0
    \end{equation*}
    {\bf if and only if} $x \in \R^n, s \in \R^m$ is optimal to the dual and $y \in \R^m$ is optimal to the primal.
\end{theorem}

\subsection{Properties of primal and dual LP solutions}\label{sec:properties_of_lp_solution}

We show that with isolation lemma, we can obtain a unique solution to the primal LP, and it highlights a set of at most $n/2$ tight constraints of a strictly complementary optimal solution of the dual LP.

\begin{lemma}[Properties of LP solutions]\label{lem:properties_of_lp_solution}
    Given a bipartite graph $G$ with $n$ vertices and $m$ edges. Let $w\in \mathbb{N}^m$ be edge weight. Let $\ov{b}\in \mathbb{Z}^{m}$ be the output of $\textsc{Isolation}(m,Z)$ where $Z := n^n$ (Algorithm~\ref{alg:isolation_lemma}). Let $b := \ov{b} + n^{15}\cdot w \in \Z^m$. Let $A \in \{0,1\}^{m \times n}$ be edge-vertex incident matrix of $G$ (unsigned). Consider the linear programming in Eq.~\eqref{eq:primal} and Eq.~\eqref{eq:dual} with parameter $A$ and $b$. With probability at least $1/4$,
    we have 
    \begin{enumerate}
        \item There is a unique solution $y^*\in \R^m$ to the primal LP (Eq.~\eqref{eq:primal}). Furthermore, $y^*\in \{0,1\}^m$, $y^*$ is the maximum candidate matching of $G$;
        \item Let $x^*\in \R^n$, $s^*\in\R^m$ be an optimal solution to the dual LP (Eq.\eqref{eq:dual}). 
        Then we have the following properties on $s^*$.
        \begin{enumerate}
            \item For any $i \in [m]$, if $s^*_{i} > 0$ then $y^*_i = 0$;
            \item If $(x^*,s^*)$ is strictly complementary, i.e., $s^*_i>0$ for every $i\in[m]$ with $y^*_i=0$ (such an optimal solution exists by the Goldman--Tucker theorem~\cite{gt56}), then $\{i\in[m] : s^*_i=0\}=\mathrm{supp}(y^*)$, so $\|s^*\|_0\geq m-n/2$.
        \end{enumerate} 
    \end{enumerate}
\end{lemma}
\begin{proof}
{\bf Part 1}

Let the feasible space of $y$ be $S := \{y\in \R^m \mid Ay\leq \mathbf{1}_n, y\geq 0\}$. We implicitly have that $y\leq \mathbf{1}_m$, so $S$ is a bounded region. Let $\ov{S}$ denote all extreme points on $S$.

First, we argue that there is a unique extreme point in $\ov{S}$ which has the optimal solution.

By Lemma~\ref{lem:totally_unimodular}, 
we know all extreme points is integral. Since $\mathbf{0}_m\leq y\leq \mathbf{1}_m$, all extreme points are in $\{0,1\}^m$, which correspond to a matching. Because we set $b := \ov{b} + n^{15}\cdot w$ as our objective vector, we can write 
\begin{align*} 
\langle b , y \rangle= \langle \ov{b}, y \rangle + n^{15} \cdot w^{\top}y.
\end{align*}
Since $w$ is integral and $(n/2)\|\ov{b}\|_{\infty}\leq n^{15}/2 < n^{15}$ (by the bound $\|\ov{b}\|_\infty\leq n^{14}$ from the proof of Lemma~\ref{lem:running_time_main}), the extreme point who has the optimal objective value must be a maximum weighted matching. Let $\mathcal{F}$ be the set of maximum weight matchings of $G$ with respect to $w$. We have $|\mathcal{F}| \leq n^n = Z$. The proof of Lemma~\ref{lem:isolation_lemma} holds verbatim with ``maximum'' in place of ``minimum'' (Steps 2--4 are symmetric), so with probability at least $1/4$ a unique $M^*\in\mathcal{F}$ maximizes $\ov{b}(M)$. Every matching outside $\mathcal{F}$ has $b$-value smaller than that of $M^*$ by at least $n^{15}-(n/2)\|\ov{b}\|_\infty>0$.

Hence, with probability at least $1/4$, $y^*:=\mathbf{1}_{M^*}$ is the unique extreme point in $\ov{S}$ which has the optimal objective value.

Because our feasible space $S$ is bounded, all point $y \in S\backslash \ov{S}$ can be written as a linear combination of extreme points on $S$. That is, if we write $\ov{S} = \{y^{(1)},\cdots,y^{(s)}\}$ where $s:=|\ov{S}|$, then all point $y \in S\backslash \ov{S}$ can be written as
\begin{align*}
y = \sum_{i\in [s]} a_i y^{(i)},
\end{align*}
where $0 \leq a_i < 1, \forall i \in [s]$ and $\sum_{i \in[s]}a_i = 1$. Therefore, we have
\begin{align*}
b^{\top} y = \sum_{i\in[s]} a_i \cdot (b^{\top} y^{(i)}) < \max_{i\in[s]} b^{\top} y^{(i)}.
\end{align*}

So $y^*$ is actually the unique optimal solution among all points in $S$.

{\bf Part 2}
Part (a) follows trivially from Theorem~\ref{thm:strong_duality_cs}.

Now we prove Part (b). If $(x^*,s^*)$ is strictly complementary, then $s^*_i=0$ implies $y^*_i\neq 0$. Part (a) gives the converse, so in this case $\{i\in[m] : s^*_i=0\}=\mathrm{supp}(y^*)$. By Part 1, $y^*\in\{0,1\}^m$ is the indicator vector of a matching, so $|\mathrm{supp}(y^*)|\leq n/2$ and $\|s^*\|_0\geq m-n/2$. Without strict complementarity, an optimal dual solution can have many more than $n$ zero slacks even though $y^*$ is unique, so the bound needs a strictly complementary solution. Algorithm~\ref{alg:minimum_vertex_cover} finds the set of tight constraints of such a solution (Lemma~\ref{lem:correct_minimum_vertex_cover}).
\end{proof}
 
\appendix
\section*{Appendix}

{\bf Roadmap} We first give a brief summary of prior non-IPM techniques for computing the maximum matching in Section \ref{sec:summary_non_ipm}. Then we show how to reduce the problem of solving an SSD$_0$ system to solving an SDDM system in Section \ref{sec:solver_reduction}. We present the streaming implementation of the isolation lemma in Section \ref{sec:isolation_lemma}. In Section \ref{sec:additional_algorithms}, we provide some small space implementations of various barrier functions.
In Section~\ref{sec:lp_in_small_treewidth}, we provide a more space-efficient algorithm when the linear program has small treewidth.

\section{A brief summary of prior non-IPM techniques} \label{sec:summary_non_ipm}

In this section, we summarize the previous techniques for computing the maximum matching in the streaming model. 

\begin{itemize}
    \item In Section~\ref{sec:previous_techniques_approx}, we introduce some representative approximation algorithms for bipartite matching. 
    \item In Section~\ref{sec:previous_techniques_approx_to_exact}, we present a method to compute an exact bipartite matching, showcasing the current state-of-the-art in the field.
    \item In Section~\ref{sec:previous_techniques_folklore}, we discuss a simple folklore semi-streaming algorithm that uses $O(n \log n)$ passes.
\end{itemize}

\subsection{Approximation algorithms}\label{sec:previous_techniques_approx}

Given a parameter $\eps \in (0, 1)$, many streaming algorithms are to find a matching of size $(1 - \eps)$ times the size of the maximum matching. 
The space and passes usages of these approximation algorithms are increasing functions of $1/\eps$.\footnote{We will be focusing on approximate algorithms that find a matching that is \emph{close to} (or can potentially be used to find) an exact maximum matching, so all the constant-approximate algorithms are not introduced here. We refer the interested readers to \cite{ab19} and the references therein.}

A natural idea to find an approximate matching is to iteratively sample a small subset of edges and use these edges to refine the current matching. These algorithms are called \emph{sampling-based} algorithms. 
In \cite{ag18}, Ahn and Guha show that by adaptively sampling $\wt{O}(n)$ edges in each iteration, one can either obtain a certificate that the sampled edges admit a desirable matching, or these edges can be used to refine the solution of a specific LP. 
The LP is a nonstandard relaxation of the matching problem, and will eventually be used to produce a good approximate matching. 
The algorithm of Ahn and Guha can compute a $(1-\eps)$-approximate matching for weighted (not necessarily bipartite) graph in $\wt{O}(1/\eps)$ passes and $\wt{O}(n \poly(1/\eps))$ space. 
However, the degree of $\poly(1/\eps)$ in the space usage can be very large, making their algorithm inapplicable for small (non-constant) $\eps = o(1 / \log n)$ in the semi-streaming model. 

Finding a $(1-\eps)$-approximate maximum matching with no space dependence on $\eps$ requires different methods. 
Inspired by the well-studied \emph{water filling} process in online algorithms (see \cite{djk13} and the references therein), Kapralov proposes an algorithm that generalizes the water filling process to multiple passes \cite{k13}. This algorithm works in the vertex arrival semi-streaming model, where a vertex and all of its incident edges arrive in the stream together.
The observation is that the water filling from pass $(k-1)$ to pass $k$ follows the same manner as that in the first pass (with a more careful double-counting method), then solving differential equations gives a $(1 - 1 / \sqrt{2 \pi k})$-approximate matching in $k$ passes.

Kapralov's algorithm removes the $\poly(\log n)$ factor in the number of passes comparing to \cite{ag11}, giving a $(1 - \eps)$-approximate maximum matching in $O(1/\eps^2)$ passes, albeit in a stronger vertex arrival model. 
Recently, Assadi, Liu, and Tarjan give a simple semi-streaming algorithm based on auction that computes a $(1 - \eps)$-approximate maximum matching in $O(1/\eps^2)$ passes, removing the vertex arrival condition \cite{alt20}. 
Very recently, \cite{ajj+22} propose an algorithm that computes a $(1-\eps)$-approximate maximum cardinality matching in $O(\eps^{-1}\log n \log\eps^{-1})$ passes and $\wt{O}(n)$ space. 
Their method leverages recent advances in $\ell_1$-regression with several ideas for implementing it in low space, leading to a streaming algorithm with no dependence on $\eps$ in the space usage, and thus improving over \cite{ag18}. 
In the next subsection, we show how they manage to get an $n^{3/4+o(1)}$-pass semi-streaming algorithm using this new approximate algorithm.

\subsection{From approximate to exact maximum matching}\label{sec:previous_techniques_approx_to_exact}

One of the methods to compute an exact maximum cardinality matching is to augment an approximate matching by repeatedly finding augmenting paths.\footnote{Given a matching in a graph, an \emph{augmenting path} is a path that starts and ends at an unmatched vertex, and alternately contains edges that are outside and inside the matching.}
Note that currently there is no semi-streaming algorithm that solves directed graph reachability -- a problem that is no harder than finding one augmenting path -- in $o(\sqrt{n})$ passes \cite{jls19}.
The linear-work parallel algorithm of \cite{jls19} can be translated into a semi-streaming algorithm that finds an augmenting path in $n^{1/2 + o(1)}$ passes.
Under this observation, a followup\footnote{The second arXiv version of \cite{ajj+22} was released ten months after our first released version. The first arXiv version of \cite{ajj+22} does not contain the $n^{3/4+o(1)}$-pass exact result of computing maximum cardinality bipartite matching.} work of our paper by \cite{ajj+22} computes a matching of size at least $\OPT - O(n^{1/4})$ in $\wt{O}(n^{3/4})$ passes (assuming the maximum matching has size $\OPT = \Theta(n)$), then augments this matching to size of $\OPT$ by a streaming version of \cite{jls19} in $n^{3/4 + o(1)}$ passes. 
In the first released version of our paper ($\wt{O}(\sqrt{m})$-pass semi-streaming algorithm for maximum weight bipartite matching), we were unaware of any approximate matching algorithm that is better than \cite{alt20} (recall that the algorithm of \cite{ag18} does not work in semi-streaming when $\eps$ is too small), therefore we merely stated this framework of augmenting an approximate matching to exact by a streaming version of \cite{jls19}. 
The $n^{3/4 + o(1)}$-pass semi-streaming algorithm (\cite{ajj+22}) can only deal with maximum cardinality bipartite matching, while our $\wt{O}(\sqrt{m})$-pass semi-streaming algorithm can solve maximum weight bipartite matching. 
\subsection{A folklore algorithm with \texorpdfstring{$O(n \log n)$}{} passes}\label{sec:previous_techniques_folklore}

A simple folklore algorithm inspired by the classic algorithm of Hopcroft and Karp \cite{hk73} can actually find the exact maximum cardinality bipartite matching in $\wt{O}(n)$ passes using $\wt{O}(n)$ space. 
The main idea is the following. 
Let $\OPT$ be the size of the maximum matching in the given $n$-vertex bipartite graph.
If the current matching has size $i$, then there must exist $(\OPT-i)$ disjoint augmenting paths, so the shortest augmenting path has length at most $n/(\OPT-i)$. 
Using a breath-first search (simply ignore the edge in the stream that is not incident with the frontier of the breath-first search), one can find this path in $n/(\OPT-i)$ passes and augment the current matching. 
Therefore, the total number of passes to compute the perfect matching is at most
$ \sum_{i = 0}^{\OPT-1} \frac{n}{\OPT-i} = n \cdot \sum_{i = 1}^{\OPT} \frac{1}{i} = O(n\log n) . $
This simple algorithm was state-of-the-art before \cite{ajj+22} and this work.

\section{Solver reductions} \label{sec:solver_reduction}

In this section, we reduce the problem of solving an SDD$_0$ system to solving an SDDM system, giving an SDD$_0$ solver in the streaming model, completing Section~\ref{sec:sdd_solver_streaming_model}.

We first reduce the problem of solving an SDD$_0$ system to solving an SDDM$_0$ system by decomposing $A$ into $D + A_{\mathrm{neg}} + A_{\mathrm{pos}}$ in Section \ref{sec:app_sddm0_sdd0}. Next, we reduce the problem of solving an SDDM$_0$ system to solving an SDDM system by approximating the solution of the SDDM$_0$ system using the approximate solution of the corresponding SDDM system in Section \ref{sec:app_sddm_sddm0}.

\subsection{From SDDM\texorpdfstring{$_0$}{} solver to SDD\texorpdfstring{$_0$}{} solver}\label{sec:app_sddm0_sdd0}
We recall Gremban's reduction in \cite{st04} that reduces the problem of solving an SDD$_0$ system to solving an SDDM$_0$ system. 
Let $A$ be an SDD$_0$ matrix, decompose $A$ into $D + A_{\mathrm{neg}} + A_{\mathrm{pos}}$, where $D$ is the diagonal of $A$, $A_{\mathrm{neg}}$ contains all the negative off-diagonal entries of $A$ with the same size, and $A_{\mathrm{pos}}$ contains all the positive off-diagonal entries of $A$ with the same size. 
Consider the following linear system

\begin{align*}
    \wh{A} 
    \begin{bmatrix}
    x_1\\
    x_2
    \end{bmatrix}
    =\wh{b},
\end{align*}
where
\begin{align*}
    \wh{A}=
    \begin{bmatrix}
    D + A_{\mathrm{neg}} & -A_{\mathrm{pos}}\\
    -A_{\mathrm{pos}} & D + A_{\mathrm{neg}}
    \end{bmatrix}
    ~~~~~~
    \text{and}
    ~~~~~~
    \wh{b}=
    \begin{bmatrix}
    b\\
    -b
    \end{bmatrix}
\end{align*}

The matrix $\widehat{A}$ can be (implicitly) computed in the streaming model: in one pass we compute and store the diagonal matrix $D$ by adding the edge weights incident on each vertex; then $\widehat{A}$ is given as a stream of edges (entries) since whenever an edge (an entry in $A$) arrives, we immediately know its position in $\widehat{A}$. 
Note that if $Ax = b$ admits a solution, then $x = (x_1 - x_2) / 2$ is exactly its solution. 
Moreover, if 
$$ \left\| 
    \begin{bmatrix}
    x_1\\
    x_2
    \end{bmatrix} - \wh{A}^{\dag} \wh{b} \right\|_{\wh{A}} \le \epsilon \|\wh{A}^{\dag} \wh{b}\|_{\wh{A}} ,
$$
then $x$ satisfies $\|x - A^{\dag} b\|_A \le \epsilon \|A^{\dag} b\|_{A}$. 
So we obtain an SDD$_0$ solver with asymptotically the same number of passes and space as an SDDM$_0$ sovler. 

\subsection{From SDDM solver to SDDM\texorpdfstring{$_0$}{} solver}\label{sec:app_sddm_sddm0}

In this section, we show that to approximately solve an SDDM$_0$ system $Ay = b$, it suffices to pre-process the input in $O(1)$ passes, approximately solve an SDDM system $\wt{A} y = \wt{b}$ with at most the same size, and possibly do some post-process in $O(1)$ passes. 

If $A$ is positive definite, then we can solve the system by an SDDM solver, so assume not in the following.

We split $A$ into the blocks of the connected components of its graph, as described below. A block that is not positive definite has all row sums zero, since an irreducible SDDM$_0$ matrix with a positive row sum is nonsingular. So by Fact~\ref{fct:sddm_psd} it is a Laplacian matrix. Positive definite blocks are passed to the SDDM solver. 
Therefore, it remains to reduce the problem of (approximately) solving a Laplacian system $Ly = b$ to (approximately) solving an SDDM system. 
The following facts are well-known.
\begin{fact} \label{fct:connected}
    Given a Laplacian matrix $L$ corresponding to graph $G$, the following holds:
    \begin{itemize}
        \item $L \succeq 0$;
        \item $L \mathbf{1}_n = 0$ and thus $\lambda_1(L) = 0$;
        \item $G$ is connected iff $\lambda_2(L) > 0$.
    \end{itemize}
\end{fact}

Given a Laplacian matrix $L$ as a stream of entries, it is equivalent to treat it as a stream of edges of $G$.
In one pass, we can identify all the connected components of $G$ using $\wt{O}(n)$ space (e.g., by maintaining the spanning forest of $G$). 
Next, any entry in the stream is identified and assigned to the subproblem corresponding to the connected component that contains it.\footnote{The above process is equivalent to partition $L$ into block diagonal matrices, solve each linear system with respect to the submatrices and corresponding entries of $b$, and combine the result.} 
This does not influence the worst-case pass and space complexity, because each subproblem uses space proportional to the size of its connected component and the total number of passes depends on the connected component that takes up the most passes. 
Therefore, we can assume that $G$ is connected, which implies that $\text{rank}(L) = n-1$ by Fact~\ref{fct:connected}.

The goal of approximately solving $Ly = b$ is for given error parameter $\epsilon > 0$, finding an $\epsilon$-approximate solution $x$ satisfying
\begin{equation*} 
    \|x - L^{\dag} b \|_L \le \epsilon \|L^{\dag} b\|_L .
\end{equation*}

If $y$ is an (exact) solution to system $Ly = b$, then $y' := y - y_1 \mathbf{1}_n$ is also a solution, where $y_1$ is the first entry of $y$. 
So we can assume that the first entry of $L^{\dag} b$ is $0$. 
(There might be many solutions, but we fix one with the first entry being $0$.)
Let $\wt{A}$ be the matrix $L$ with the first row and column deleted, and let $\wt{b}$ be the vector $b$ with the first entry deleted. 
Note that $\wt{A} \succ 0$.

Let $\wt{x}$ be an $\epsilon$-approximate solution to the system $\wt{A} x = \wt{b}$, and let $x$ be the vector $\wt{x}$ with $0$ inserted as its first entry. It must be that 
$$\|x - L^{\dag} b\|_L = \|\wt{x} - {\wt{A}}^{-1} \wt{b}\|_{\wt{A}}$$ 
because vector ${\wt{A}}^{-1} \wt{b}$ is the vector $L^{\dag} b$ with the first entry deleted. 

Finally, we have that 
$\|x - L^{\dag} b \|_L \le \epsilon \|L^{\dag} b\|_L$ since 
\begin{equation*}
    \|\wt{x} - {\wt{A}}^{-1} \wt{b} \|_{\wt{A}} \le \epsilon \|{\wt{A}}^{-1} \wt{b}\|_{\wt{A}} \text{~~~and~~~} \|L^{\dag} b\|_{L} = \|{\wt{A}}^{-1} \wt{b}\|_{\wt{A}} ,
\end{equation*}  
which gives an $\epsilon$-approximate solution to the original system $Ly = b$.

\section{Isolation lemma in the streaming model}\label{sec:isolation_lemma}

\cite{crs95} shows how to implement the isolation lemma using a small amount of randomness, and in particular, in our application, the amount of randomness is $O(n\log n)$ and therefore just fits into our memory. However, since their focus is on the number of randomness, they still use extra space that we cannot afford. In this section, we make their algorithm into an oracle. This oracle only stores the random seed and performs exactly the same as the original algorithm, so that we can use it in the streaming model. Formally, our result is stated in  Lemma~\ref{lem:isolation_main_lemma}.

Before the proof, we set up some simple notations.

For a vector $w\in \mathbb{Z}^n$ and a set $S\subseteq[n]$, we denote $w_S:=\sum_{i\in S} w_i$. We can define $w_S:=\sum_{i\in S} w_i$ similarly when $w_i:\mathbb{Z}^t \rightarrow \mathbb{Z}$ is a function.

\begin{lemma}[Streaming implementation of the isolation lemma]\label{lem:isolation_main_lemma}
Let $n,\mathcal{F},Z,w$ be define as in Lemma~\ref{lem:isolation_lemma}. The Algorithm~\ref{alg:isolation_lemma} in Lemma~\ref{lem:isolation_lemma} can be implemented into such an oracle $\mathcal{I}$: $\mathcal{I}$ can output $w_i$ given any $i\in [n]$. Furthermore, if $Z\leq 2^{\poly(n)}$, then $\mathcal{I}$ uses $O(\log (Z) + \log(n))$ space.
\end{lemma}
\begin{proof}
Our $\mathcal{I}$ is the streaming implementation of  Algorithm~\ref{alg:streaming_isolation_lemma}. 
It is easy to see that after running the procedure $\textsc{Initialize}$, the procedure $\textsc{Query}(i)$ will output $w_i$ given any $i\in[n]$.

This oracle stores $m,Z\in \mathbb{N}$ and $r\in \mathbb{N}^t$ in memory. 
Note that $m,Z$ can be stored in $O(\log (Z) + \log (n))$ bits, and $r$ can be stored in 
\begin{align*} 
O(t \cdot \log (nt)) 
= & ~ O(\lceil \log (m) / \log (n) \rceil \cdot \log (nt)) \\
= & ~ O(\log (Z) + \log (n))
\end{align*} 
bits. 
Here the last step uses $\log t=O(\log n)$, which holds since $Z\leq 2^{\poly(n)}$. And it is easy to see in \textsc{Query}, all computation can be done within $O(\log (Z) + \log (n))$ space.
\end{proof}

The rest of this section is organized as follows: In Section~\ref{sec:isolation_lemma:isolation}, we the generalized isolation lemma. In Section~\ref{sec:isolation_lemma:algorithms}, we provide our streaming algorithm (in fact a data-structure). In Section~\ref{sec:proof_step_1}, Section~\ref{sec:proof_step_2}, Section~\ref{sec:proof_step_3} and Section~\ref{sec:proof_step_4}, we provide the proof details of the uniqueness.

\subsection{Isolation lemma}\label{sec:isolation_lemma:isolation}

We state the generalized isolation lemma from previous work \cite{crs95}.

\begin{lemma}[Generalized isolation lemma \cite{crs95}]\label{lem:isolation_lemma}
Fix $n\in \mathbb{N}$. Fix an unknown family $\mathcal{F} \subseteq 2^{[n]}$. Let $Z$ denote a positive integer such that $Z\geq |\mathcal{F}|$, there exists an algorithm (Algorithm~\ref{alg:isolation_lemma}) that
uses $O(t\log(nt))$ random bits (which is $O(\log(Z) + \log(n))$ when $Z\leq 2^{\poly(n)}$) to output a vector $w \in [0,t^2n^5]^n$, where $t:=\lceil \log(4n^2Z^4)/\log n\rceil$, such that with probability at least $1/4$, there is a unique set $S^*\in \mathcal{F}$ that has minimum weight $w_{S^*}$.
\end{lemma}
\begin{proof}
The proof is already done in \cite{crs95}. For the completeness, we rewrite their proof here. 
By Lemma~\ref{lem:proof_step_2}, with probability at least $1/2$, all sets $S\in\mathcal{F}$ has distinct $w^{(2)}_S$. Conditioning on the event, by Lemma~\ref{lem:proof_step_3}, we get that all sets $S\in\mathcal{F}$ has distinct $w^{(3)}_S$. Then by Lemma~\ref{lem:proof_step_4} with $K=n^2$ (the coefficients of the linear forms $w^{(3)}_S$ are sums of at most $n$ digits in $\{0,\cdots,n-1\}$, and the $r_j$ are drawn from $[n^4t]=[K^2t]$), we get that with probability at least $1/2$, the output $w$ will give a unique minimum set in $\mathcal{F}$. Moreover, $w_i=\sum_{j<t}b_{i,j}r_j<t\cdot n\cdot n^4t=t^2n^5$. Since the two events
of success are independent, the final success probability is at least $1/4$.
\end{proof}

\subsection{Algorithms}\label{sec:isolation_lemma:algorithms}

We present an algorithm 
that implements the isolation lemma in the streaming model. We give the original implementation appeared in \cite{crs95} in Algorithm~\ref{alg:isolation_lemma}.
Then in Algorithm~\ref{alg:streaming_isolation_lemma}, we show how to implement the algorithm in the streaming model.

We explain Algorithm~\ref{alg:streaming_isolation_lemma}. Algorithm~\ref{alg:streaming_isolation_lemma} is in fact a data-structure and has three parts. The first part is all the members. The second part is a function for initialization (see \textsc{Initialize}). The third part is a function for query (see \textsc{Query}). The function \textsc{Initialize} is initializing variables $m, t$ and vector $r$. The \textsc{Query} function takes $i \in [n]$ as input and output an integer $w^{(4)}$. We want to remark that this data-structure is a static data-structure. Therefore, it does not need an \textsc{Update} function.

\begin{algorithm}[ht]
\caption{Conceptual implementation of the isolation lemma. Algorithm~\ref{alg:streaming_isolation_lemma} is the implementation of this algorithm in streaming model.}
\label{alg:isolation_lemma}
\begin{algorithmic}[1]
\Procedure{$\textsc{Isolation}$}{$n,Z\in \mathbb{N}$} \Comment{$Z \geq |\mathcal{F}|$, Lemma~\ref{lem:isolation_lemma}}
\State $w_i^{(1)} \leftarrow 2^i$, $\forall i \in[n]$ \label{line:isolation_w_1} \Comment{$w^{(1)} \in \mathbb{N}^n$}
\State Choose $m$ uniformly at random from $\{1,2,\cdots,(2nZ^2)^2\}$.  \Comment{$m \leq 4n^2Z^4$}
\State For each $i$, define $w^{(2)}_i \leftarrow w^{(1)}_i \mod m$. \label{line:isolation_w_2} \Comment{$w^{(2)} \in [m]^n$}
\State $t\leftarrow \lceil \log (m) / \log (n)\rceil$
\For{$i = 1 \to n$}
    \State $\ov{b_{i,t-1},\cdots,b_{i,1},b_{i,0}} \leftarrow w_i^{(2)}$  \Comment{Write $w_i^{(2)}$ in base $n$. $b_{i,j}\in [n]$ are digits.}
    \State \Comment{Note that $t$ is an upper bound on the length}
    \State $w^{(3)}_i(y_0,\cdots,y_{t-1}) \leftarrow \sum_{j=0}^{t-1} b_{i,j}\cdot y_j$ \label{line:isolation_w_3} \Comment{$w^{(3)}_i: \mathbb{Z}^{t}\rightarrow \mathbb{Z}$ is a linear form.}
\EndFor
\State Choose $r_0,\cdots,r_{t-1}$ uniformly and independently at random from $\{1,2,\cdots,n^4 t\}$.
\State $w^{(4)}_i \leftarrow w^{(3)}_i(r_0,\cdots,r_{t-1})$, $\forall i \in [n]$ \label{line:isolation_w_4}
\State \Return $w^{(4)} \in \mathbb{N}^{n}$.
\EndProcedure
\end{algorithmic}
\end{algorithm}

\begin{algorithm}[ht]
\caption{Streaming implementation of Algorithm~\ref{alg:isolation_lemma}.}
\label{alg:streaming_isolation_lemma}
\begin{algorithmic}[1]
\State {\bf data structure} \Comment{Lemma~\ref{lem:isolation_main_lemma}}
\State {\bf members}
\State \hspace{4mm} $n,t \in \mathbb{N}$ \Comment{$t = O(\log(Z) / \log(n))$}
\State \hspace{4mm} $m,Z \in \mathbb{N}$ \Comment{$m = O(n^2 Z^4)$, $Z \geq |\mathcal{F}|$}
\State \hspace{4mm} $r_0,\cdots,r_{t-1} \in \mathbb{N}$ \Comment{$r_i \leq n^4 t$}
\State {\bf end members}
\Procedure{Initialize}{$n,Z\in \mathbb{N}$} \Comment{Initialization}
\State Choose $m\in \mathbb{N}$ uniformly at random from $\{1,2,\cdots,(2nZ^2)^2\}$
\State $t\leftarrow \lceil \log m / \log n\rceil$ \Comment{$t\in \mathbb{N}$}
\State Choose $r_0,\cdots,r_{t-1}$ uniformly and independently at random from $\{1,2,\cdots,n^4 t\}$ \Comment{$r\in \mathbb{N}^t$}
\EndProcedure
\Procedure{Query}{$i\in [n]$}
\State $w^{(2)} \leftarrow 2^i \bmod m$, computed by repeated squaring modulo $m$ \Comment{$O(\log m)$ bits, and $2^i$ is never written out}
\State $\ov{b_{t-1},\cdots,b_{1},b_{0}} \leftarrow w^{(2)}$ \Comment{$b_j \in [n]$, $\forall j\in[t]$}
\State \Comment{Write $w^{(2)}$ in base $n$. Note that $t$ is an upper bound on the length}
\State $w^{(4)} \leftarrow \sum_{j=0}^{t-1} b_j\cdot r_j$
\State \Return $w^{(4)}$ \Comment{$w^{(4)}< t^2n^5$}
\EndProcedure
\State {\bf end data structure}
\end{algorithmic}
\end{algorithm}

\newpage
\subsection{Proof of uniqueness: step 1}\label{sec:proof_step_1}

The goal of this section is to prove Lemma~\ref{lem:proof_step_1}.
The following lemma is from \cite{t93}.
\begin{lemma}[Step 1, \cite{t93}]\label{lem:proof_step_1}
Let $L \geq 100$ and let $S$ be any subset of $\{1,\cdots,L^2\}$ such that $|S|\geq \frac{1}{2}L^2$. Then, the least common multiple of the elements in $S$ exceeds $2^L$.
\end{lemma}

\subsection{Proof of uniqueness: step 2}\label{sec:proof_step_2}

The goal of this section is to prove Lemma~\ref{lem:proof_step_2}.

\begin{lemma}[Step 2] \label{lem:proof_step_2}
With probability at least $1/2$, all sets $S \in \mathcal{F}$ have distinct weights $w^{(2)}_S$.
\end{lemma}
\begin{proof}
We write $\mathcal{F} = \{S_1,\cdots,S_k\}$ where $k\leq Z$. Define
\[
I := \prod_{1\leq i < j\leq k} \big| w^{(1)}_{S_i} - w^{(1)}_{S_j} \big|.
\]

First note that $I\neq 0$ since every set $S\in\mathcal{F}$ have distinct weights $w^{(1)}_S$ by definition of $w^{(1)}_i = 2^i$ (Line~\ref{line:isolation_w_1}, Algorithm~\ref{alg:isolation_lemma}).

Next, we give an upper bound on $I$. For each pair of $1\leq i < j \leq k$, $|w^{(1)}_{S_i} - w^{(1)}_{S_j}| \leq 2^{n+1}$. There are totally $Z^2$ pairs, so $I < 2^{2nZ^2}$.

Let $L = 2nZ^2$. Let $S=\{1,\cdots,L^2\}$ be all the possible choices of $m$. 
We have that at least half choices of $m$ satisfies $I \mod m \neq 0$, since otherwise by applying Lemma~\ref{lem:proof_step_1}, $I$ is at least the least common multiplier of half numbers of $S$, i.e., $I\geq 2^L$, contradicting with the upper bound on $I<2^{2nZ^2}$.

Therefore, with probability at least $1/2$ (the randomness is over the choice of $m$), $I \mod m\neq 0$, which means that all sets $S\in \mathcal{F}$ have distinct weight $w^{(2)}_S$ by our definition of $w^{(2)}_i = w^{(1)}_i \mod m$ (Line~\ref{line:isolation_w_2},  Algorithm~\ref{alg:isolation_lemma}).
\end{proof}

\subsection{Proof of uniqueness: step 3}\label{sec:proof_step_3}
The goal of this section is to prove Lemma~\ref{lem:proof_step_3}.
\begin{lemma}[Step 3] \label{lem:proof_step_3}
If $w^{(2)}_S$ are distinct for all $S\in \mathcal{F}$, then the linear form $w^{(3)}_S$ are all distinct for all $S \in \mathcal{F}$.
\end{lemma}
\begin{proof}
Use proof by contradiction. Suppose there exists two distinct sets $S_1,S_2\in \mathcal{F}$ that $w^{(3)}_{S_1} = w^{(3)}_{S_2}$.
Let $y_i = n^i$, $\forall i\in[t]$. 
We will have 
\[
w^{(2)}_{S_1} = w^{(3)}_{S_1}(y_0,\cdots,y_{t-1}) = w^{(3)}_{S_2}(y_0,\cdots,y_{t-1}) = w^{(2)}_{S_2}\]
by definition of $w^{(3)}_i$ (Line~\ref{line:isolation_w_3}), which is a contradict to our assumption that all $w^{(2)}_S$ are distinct.
\end{proof}

\subsection{Proof of uniqueness: step 4}\label{sec:proof_step_4}
The goal of this section is to prove Lemma~\ref{lem:proof_step_4}.
\begin{lemma}[Step 4] \label{lem:proof_step_4}
Let $\mathcal{C}$ be any collection of distinct linear forms over at most $t$ variables $y = y_0,\cdots,y_{t-1}$ with coefficients in $\{0,1,\cdots,K-1\}$ for some integer $K\geq 2$. Choose a random $r = r_0,\cdots,r_{t-1}$ by assigning each $r_i$ uniformly and independently from $[K^2\cdot t]$. Then in the assignment $y=r$ there will be a unique linear form with minimum value, with probability at least $1/2$.
\end{lemma}
\begin{proof}
We call a variable $y_i$ to be \textit{singular} under an assignment $r\in \Z^t$ if there exist two minimum linear forms in $\mathcal{C}$ under assignment $r$ whose coefficients of $y_i$ differ. Then an assignment $r$ gives unique minimum linear form if and only if no variable $y_i$ is singular under this assignment. We will calculate the probability of $y_i$ being singular under random assignment $r$ and then take union bound over every $y_i$.

For each $y_i$, fix all $r_0,\cdots,r_{i-1},r_{i+1},\cdots,r_{t-1} = a_0,\cdots,a_{i-1},a_{i+1},\cdots,a_{t-1}$ other than $r_i$. Now, every linear form $f$ under this partial assignment can be written as $a_f + b_f \cdot y_i$ with $b_f < K$. We split $\mathcal{C}$ into $K$ classes $\mathcal{C}_0,\cdots,\mathcal{C}_{K-1}$ where $\mathcal{C}_j$ contains all linear forms with $b_f=j$. Let $p_j$ be the minimum $a_f$ among all linear forms in $\mathcal{C}_j$. According to the definition of singular, $y_i$ is singular on assignment $r_i$ if and only if the minimum value in the list 
\[
\{p_0,p_1 + r_i,p_2 + 2r_i,\cdots, p_{K-1} + (K-1) \cdot r_i\}
\]
is not unique, which is upper bounded by the probability that the elements in the list has a collision. 
Since every pair of elements in the list can have at most one choice of $r_i$ such that they are equal, we have
\begin{align*}
    \Pr_{r_i}[y_i \text{~is singular} \mid r_j = a_j,~\forall j\neq i] 
    \leq & ~ \Pr_{r_i}[\exists l,m \text{~s.t.~} p_l + l\cdot r_i = p_m + m\cdot r_i]\\
    \leq & ~ \sum_{l\neq m} \Pr_{r_i}[p_l + l\cdot r_i = p_m + m\cdot r_i]\\
    \leq & ~ \sum_{l\neq m} \frac{1}{K^2t}\\
    \leq & ~ \binom{K}{2}\cdot \frac{1}{K^2t}\\
    \leq & ~ \frac{1}{2t},
\end{align*}
where the first step follows from the definition of singular variable, the second step follows from the definition of union bound, the third step follows from the assigning of $r_i$, the fourth step follows from the number of pairs $\{l,m\}$ being at most $\binom{K}{2}$, and the last step follows from $\binom{K}{2} \leq \frac{K^2}{2}$. 

Finally, by a union bound on the events that every $y_i$ is not singular, the conclusion follows with probability at least $1-t\cdot \frac{1}{2t} = 1/2$.
\end{proof}

\section{Additional algorithms}\label{sec:additional_algorithms}

In the following subsections, we include the small space implementations of various barrier functions in Section~\ref{sec:lp_in_small_space}: 
\begin{itemize}
    \item In Section~\ref{sec:app_log_bar} we introduce the data structure \textsc{LogBarrier}, that approximate the gradient and Hessian of logarithmic barrier function in small space. 
    \item In Section \ref{sec:app_hybrid_bar}, we introduce 
 the data structure \textsc{HybridBarrier}, that approximate the gradient and Hessian of hybrid barrier function in small space.
    \item  In Section \ref{sec:app_lee_bar}, we present the algorithm (Algorithm \ref{alg:lewis_weight_small_space}), that compute the Lewis weight needed for the approximation of the Hessian and the calculation of gradient for Lee-Sidford barrier function.
\end{itemize}
All the data structures and algorithms implemented in this section are subroutines for Algorithm \ref{alg:ipm_stream}.
\subsection{Logarithmic barrier}
\label{sec:app_log_bar}
We start with the Logarithmic barrier function. By definition of $\phi(x) $(Def.~\ref{def:log_barrier_function}), we have
\begin{align*}
    \nabla \phi(x)  & ~ = - \sum_{i \in [m]} \frac{a_i}{s_i(x)} \in \R^{n}; \\
    H(x) = \nabla^2 \phi(x) & ~  = \sum_{i \in [m]} \frac{a_i a_i^{\top}}{s_i(x)^2} \in \R^{n\times n}.
\end{align*}
In procedure \textsc{ApproxGradient}, we first initialize a $0$ vector $f$ and accumulate $-\frac{a_i}{s_i(x)}$ to $f$. After $m$ times accumulation, we get $ \nabla \phi(x)$. In this streaming model, we only need $O(n)$ space to store $f \in \R^n$ and $a_i \in \R^n$.  

Similarly, for the Hessian matrix of the barrier function, by accumulating 
\begin{align*}
   H \gets H + \frac{a_i a_i^{\top}}{s_i(x)^2}
\end{align*}, where $s_i = a_i^{\top}x-b_i$ is the slack variables, we find the exact Hessian matrix. For each iteration, the space is upper bounded by the space needed for $a_i a_i^\top \in \R^{n\times n}$. Hence, the space needed for the calculation of the Hessian is $O(n^2)$. 

Given the Hessian and the gradient of the barrier function, we get $\wt{\delta}_x:= -H(x)^{-1}\nabla f_t(x)$ without any error in overall $O(n^2)$ space. Assumption~\ref{ass:input_of_ipm} holds as the Hessian and the gradient of the logarithmic barrier function can be calculated exactly. By Lemma~\ref{lem:ipm_error}, we show the correctness of our implementation.

\begin{algorithm}[!ht]
\caption{Barrier data structure for logarithmic barrier}
\label{alg:logbar_space}
\begin{algorithmic}[1]
\State {\bf data structure} \textsc{LogBarrier} \Comment{Theorem \ref{thm:nsquare_logarithmic}}
\State 
\Procedure{Init}{$A, b, x$}
\State \Return
\EndProcedure
\State 
\Procedure{ApproxGradient}{$A, b, x, \epsilon_{g},t, c$}
\State \Comment{We compute the exact gradient without using $\epsilon_{g}$}
\State $f\gets {\bf 0}_n$
\For{$i=1\to m$}
\State $s_i(x)=a_i^\top x-b_i$
\State $f\gets f-\frac{a_i}{s_i(x)}$
\EndFor
\State \Return $tc + f$
\EndProcedure
\State 
\Procedure{ApproxHessian}{$A, b, x, \gamma, t, c$}
\State \Comment{We compute the exact Hessian without using $\gamma$}
\State $H\gets {\bf 0}_{n\times n}$
\For{$i=1\to m$}
\State $s_i(x)=a_i^\top x-b_i$
\State $H\gets H+\frac{a_ia_i^\top}{s_i^2(x)}$
\EndFor
\State \Return $H$
\EndProcedure
\State
\State {\bf end data structure}
\end{algorithmic}
\end{algorithm}
\newpage

\subsection{Hybrid barrier}
\label{sec:app_hybrid_bar}
 Instead of taking only one pass, here we need two passes each time to approximate the Hessian or gradient. In this first pass, we accumulate $\frac{a_ia_i^\top}{s_i^2(x)}$, and get $\nabla^2 \phi(x)$. For the second pass, we  accumulate $-(\sigma_i(x) + n/m)\frac{a_i}{s_i(x)}$ to get the exact gradient. For the approximation of Hessian, we calculate $Q(x)$ that defined in definition~\ref{def:Q} by accumulating $\frac{\sigma_i(x)}{s_i(x)^2} a_i a_i^\top$ to $Q(x)$ and approximate $H(x)$ by using $\wt{H}(x) = \big( 5Q(x) + (n/m) \nabla^2 \phi(x) \big)$. The proof of Theorem~\ref{thm:nsquare_hybrid} runs the IPM with the barrier $k\cdot V_{\rho}$. So there, the barrier part $f$ of the gradient and the matrix $\wt{H}(x)$ are multiplied by $k$.

\begin{algorithm}[!ht]
\caption{Barrier data structure for hybrid barrier }
\label{alg:hybridbar_space}
\small
\begin{algorithmic}[1]
\State {\bf data structure} \textsc{HybridBarrier} \Comment{Theorem \ref{thm:nsquare_hybrid} 
}
\State 
\Procedure{Init}{$A, b, x$}
\State \Return
\EndProcedure
\State 
\Procedure{ApproxGradient}{$A, b, x, \epsilon_{g}, t, c$}
\State \Comment{We compute the exact gradient without using $\epsilon_{g}$}
\State $f\gets {\bf 0}_n$
\State $M\gets {\bf 0}_{n\times n}$
\For{$i=1\to m$}
\State $s_i =a_i^\top x-b_i$
\State $M\gets M+\frac{a_ia_i^\top}{s_i^2 }$
\EndFor
\For{$i=1\to m$}
\State $s_i =a_i^\top x-b_i$
\State $\sigma_i  = \frac{ a_i^{\top} M^{-1} a_i }{ s_i ^2 }$
\State $f\gets f- (\sigma_i  + n/m)\frac{a_i}{s_i }$
\EndFor
\State \Return $tc + f$
\EndProcedure
\State 
\Procedure{ApproxHessian}{$A, b, x, \gamma, t, c$}
\State \Comment{We compute the approximate Hessian with $\gamma=\frac{1}{5}$}
\State $M\gets {\bf 0}_{n\times n}$
\State $Q\gets {\bf 0}_{n\times n}$
\For{$i=1\to m$}
\State $s_i =a_i^\top x-b_i$
\State $M\gets M+\frac{a_ia_i^\top}{s_i^2 }$
\EndFor
\For{$i=1\to m$}
\State $s_i =a_i^\top x-b_i$
\State $\sigma_i  = \frac{ a_i^{\top} M^{-1} a_i }{ s_i ^2 }$
\State $Q \gets Q +\frac{\sigma_i }{s_i ^2} a_i a_i^\top$
\EndFor
\State $ H =  5Q + (n/m) M $ 
\State \Return $H$ \Comment{The $\gamma$-approximation of Hessian}
\EndProcedure
\State
\State {\bf end data structure}
\end{algorithmic}
\end{algorithm}
\newpage

\subsection{Lee-Sidford barrier and Lewis weights}
\label{sec:app_lee_bar}

\begin{algorithm}[!ht]
\caption{Lewis Weight Computation with small space}
\label{alg:lewis_weight_small_space}
\begin{algorithmic}[1]
\Procedure{ComputeLewisWeight}{$A \in \R^{m \times n}, p\in \R, \epsilon>0$} \Comment{Theorem \ref{thm:nsquare_lsbarrier}}
\State Let $\alpha = \frac{2}{p-2}$, $\ov{\alpha}=\max(\alpha,1)$, $\wt{\epsilon} = \frac{\alpha^8\epsilon^4}{(25m(\sqrt{n}+\alpha)(\alpha + \alpha^{-1}))^{4}}$
\State $T = O(\max(\alpha^{-1},\alpha)\log(m/\wt{\epsilon}))$, $\eta\gets\frac{1}{3\ov{\alpha}}$
\State $Q^{(0)}\gets (\sum_{i=1}^m \frac{n}{m}\cdot a_ia_i^\top)^{-1}$ \Comment{One pass}
\State $t\gets 0$, $k\gets 0$ \Comment{$t$ counts all calls to \textsc{Descent}, $k$ counts the calls on all rows}
\While{true}
    \If{$\textsc{Rho}(t,i)>1+\alpha$ for at least one $i\in [m]$} \Comment{One pass, Algorithm~\ref{alg:lewis_weight_small_space_subroutine}}
        \State $z_{t+1}\gets 1$ \Comment{The next call to \textsc{Descent} is inside \textsc{Round}}
    \ElsIf{$k<T$}
        \State $z_{t+1}\gets 0$, $k\gets k+1$ \Comment{The next call to \textsc{Descent} is on all rows}
    \Else
        \State \textbf{break} \Comment{The current weights are $w_R$}
    \EndIf
    \State $Q^{(t+1)}\gets (\sum_{i=1}^m \textsc{ComputeW}(t+1,i)\cdot a_ia_i^\top)^{-1}$ \Comment{One pass}
    \State $t\gets t+1$
\EndWhile
\State \Return the oracle that outputs $(a_i^\top Q^{(t)} a_i)^{(1+\alpha)/\alpha}$ when the row $a_i$ is read \Comment{This is $\wh{w}_i^{1+\alpha}$ for the output $\wh{w}$ of Algorithm~\ref{alg:compute_lewis_weight}}
\EndProcedure
\end{algorithmic}
\end{algorithm}

The Lee-Sidford barrier is defined as
\begin{align*}
\psi(x) = \max_{w\in \R^m} \frac{1}{2}f(x,w).
\end{align*}
where $w_x := \arg\max_{w\in \R^m}f(x,w)$ and 
\begin{align*}
f(x,w) := \ln \det ( A_x^{\top} W^{ 1 - 2 / q } A_x ) - ( 1 - 2 / q ) \tr[W].
\end{align*}
The gradient of $\psi(x)$ is $\nabla \psi(x) = - A_x^{\top} w_x $ and approximation of the Hessian matrix is \begin{align*}
\wt{H}(x) = (1+q)A^{\top}S(x)^{-1}W_xS(x)^{-1}A.
\end{align*}
Assume the Lewis weight $w_x$ is given, other matrix computation only need $O(n^2)$ space by the accumulating the outer product as our implementation for the hybrid barrier. So here we only consider the space needed for computing Lewis weight. 
As in the proof of Theorem \ref{thm:nsquare_lsbarrier}, let $v^{(t)}$ be the weight vector of Algorithm~\ref{alg:compute_lewis_weight} after its first $t$ calls to \textsc{Descent}, and let $z_t=1$ if the $t$-th call is inside \textsc{Round} and $z_t=0$ if it is on all rows. As shown in that proof, we need 
\begin{align*} 
\sigma_i(v^{(t)}) = v^{(t)}_i\cdot a_i^{\top}Q^{(t)}a_i.
\end{align*}
to compute $v_i^{(t+1)}$ and need 
\begin{align*}
Q^{(t)} = (K^{(t)} )^{-1} = (A^{\top}V^{(t)}A)^{-1}
\end{align*}
to compute $\sigma_i(v^{(t)})$. 

Hence, to compute $v_i^{(t)}$ from $a_i$, we need to recursively compute $v_i^{(t-1)}, v_i^{(t-2)},\cdots, v_i^{(0)}=\frac{n}{m}$, using the stored matrices $Q^{(0)},\cdots,Q^{(t-1)}$ and the stored bits $z_1,\cdots,z_t$. This recursion is a chain of length $O(t)$ and reads no other row of $A$. Our implementation (Algorithm~\ref{alg:lewis_weight_small_space}) is bottom-up. For each call to \textsc{Descent}, one pass decides whether the call is inside \textsc{Round} or on all rows, and a second pass computes and stores the next matrix. Let $J$ be the total number of calls. One more pass computes $Q^{(0)}$, and one more pass finds that the last call to \textsc{Round} has returned. So the algorithm uses $2J+2$ passes. Since it stores $J+1$ matrices of size $n\times n$ and $J$ bits, its space is $O(Jn^2)$. By Lemma~\ref{lem:flps}, there are $J=O(p^3\log(pm/\epsilon))=\poly(\log m)$ calls for $p=\Theta(\log m)$ and $\epsilon\geq 1/\poly(m)$, so the space we need is $\wt{O} (n^2)$.
\begin{algorithm}[!ht]
\caption{Subroutines: ComputeW and Rho. They are called while the row $a_i$ is read from the stream, use only the matrices and bits stored by Algorithm~\ref{alg:lewis_weight_small_space}, and read no other row.}
\label{alg:lewis_weight_small_space_subroutine}
\begin{algorithmic}[1]
\Procedure{ComputeW}{$t\in \Z_{\geq 0}, i\in [m]$} \Comment{Returns $v^{(t)}_i$ using $Q^{(0)},\cdots,Q^{(t-1)}$ and $z_1,\cdots,z_t$}
\State $w_i\gets \frac{n}{m}$
\For{$s=1\to t$}
    \State $\rho_i\gets \frac{a_i^\top Q^{(s-1)} a_i}{w_i^{\alpha}}$ \Comment{This is $\rho_i(v^{(s-1)})$}
    \If{$z_s=0$ or $\rho_i>1+\alpha$}
        \State $w_i\gets w_i(1+\eta \cdot \frac{\rho_i-1}{\rho_i+1})$ \Comment{The $s$-th call to \textsc{Descent} updates row $i$}
    \EndIf
\EndFor
\State \Return $w_i$
\EndProcedure
\State
\Procedure{Rho}{$t\in \Z_{\geq 0}, i\in [m]$} \Comment{Returns $\rho_i(v^{(t)})$ using $Q^{(0)},\cdots,Q^{(t)}$ and $z_1,\cdots,z_t$}
\State $w_i\gets \textsc{ComputeW}(t, i)$
\State \Return $\frac{a_i^\top Q^{(t)} a_i}{w_i^{\alpha}}$
\EndProcedure
\end{algorithmic}
\end{algorithm}

\clearpage
\newpage
\section{Solving small treewidth LP in small space}
\label{sec:lp_in_small_treewidth}

In addition to general linear program, we also study the setting where the LP has small treewidth. Following the formulation of~\cite{dly21}, we define the treewidth as the treewidth of the graph induced by viewing $A\in \R^{m\times n}$ as a generalized incidence matrix. 


The key to implement the IPM in small space is to compute a space-efficient representation of the Gram matrix $A^\top HA$. As this is an $n\times n$ matrix, $O(n^2)$ space is needed as there are $O(n^2)$ parameters. When the constraint matrix $A$ is an incidence matrix for a graph, it is natural to parameterize the graph in terms of its \emph{treewidth} $\tau$.~\cite{dly21} extends this graph notion into linear program, and one of their contributions is to show that if $A$ has treewidth $\tau$, then one can compute a permutation matrix $P\in \R^n$ such that the Cholesky factorization $PA^\top HAP^\top=LL^\top$ is \emph{sparse}, i.e., $L\in \R^{n\times n}$ has column sparsity $\wt O(\tau)$. This motivates us to design a space-efficient algorithm to compute the Cholesky factor $L$, as it is lower triangular, it becomes much easier to solve linear system with respect to $L$. 

Computing $L$ involves two phases: 1).\ compute the permutation matrix $P$ and 2).\ compute the Cholesky factorization. Let us explain these two phases in reverse order: suppose we have already computed the permutation, then there exists a Cholesky factor that has only $\wt O(n\tau)$ parameters, and it can be computed in $\wt O(n\tau)$ space. Unfortunately, computing the factorization requires a dependence chain of $\Omega(n)$, this means that one pass over $A$ can compute at most $O(n)$ entries. Since $L$ has $\wt O(n\tau)$ nonzero entries, we need an extra of $\wt O(\tau)$ passes. It is then instructive to compute $P$ in $\wt O(n\tau)$ space. To do so, we utilize the recursive algorithm of~\cite{dly21}. Though the algorithm is recursive in natural, each recursion only needs to store an ordering on a partition of vertices, thus, the total space is only $O(n)$. To make sure the ordering reflects the nonzero pattern of $L$, we then need to compute a balanced separator of the graph. We store the graph of $A$, which has at most $n\tau$ edges, in one pass. We then search it for a balanced separator with at most $\tau+1$ vertices, which exists since the treewidth is $\tau$. This search needs $O(n\tau)$ space and no further pass, and we do not bound its running time. The almost-linear time algorithm of~\cite{bgs21} could be used instead, but then we can only bound the space by the running time, which is $(n\tau)^{1+o(1)}$. Note that the space used to find such separator can be reused by other recursive calls, as each recursion only needs to store an ordering on a subset of vertices.

Now that we have a sparse Cholesky factor, we can use it to solve lower triangular system in $O(n)$ space via back substitution, enabling us to compute leverage score and Lewis weights in $\wt O(n\tau)$ space. Essentially, the problem of solving IPM in small space boils down to compute a succinct representation for the inverse Gram matrix $(A^\top A)^{-1}$. Treewidth is one of the effective parameters that can reduce the number of parameters for this Gram from $n^2$ to $n\tau$. This implies a more efficient space algorithm as long as $(A^\top A)^{-1}$ admits a compact representation.

Given an LP with treewidth $\tau$ whose feasible region lies in a ball of radius $R$ and contains a ball of radius $r$ with known center, as in Theorem~\ref{thm:nsquare_lsbarrier}, we show that our dual-only IPM can be implemented in $\wt O(n\tau)$ space and $\wt O(\sqrt{n}\tau\log(mR^2\|c\|_2/(r\epsilon)))$ passes. We want to stress that this result has significant consequences for graphs with polylogarithmic treewidth: if the graph problem can be solved with the incidence matrix, then it can be solved in $\wt O(n)$ space and $\wt O(\sqrt{n}\log(mR^2\|c\|_2/(r\epsilon)))$ passes for $\tau=\wt O(1)$. 

\begin{theorem}[Small treewidth LP, informal version of Theorem \ref{thm:lsbarrier_treewidth}]
\label{thm:lsbarrier_treewidth_informal}

Given a linear program with $m$ constraints, $n$ variables and treewidth $\tau$ in the streaming model, whose feasible region lies in a ball of radius $R$ and contains a ball of radius $r$ with known center, there exists an algorithm that outputs an $\epsilon$-approximate solution to the dual program (Eq.~\eqref{eq:lp_primal_dual}) in $\wt O(n\tau)$ space and $\wt O(\sqrt{n}\tau\log(mR^2\|c\|_2/(r\epsilon)))$ passes.
\end{theorem}
The hypothesis on the feasible region is the one of Theorem~\ref{thm:nsquare_lsbarrier}. As there, it is used to find a starting point near the central path (see the proof of Theorem~\ref{thm:nsquare_logarithmic}). It also gives the logarithmic factor in the number of iterations.

In Section \ref{sec:pre_treewidth}, we introduce some related definitions about treewidth. In Section \ref{sec:space_treewidth}, we present some lemmas to show the space needed in our implementation. In Section \ref{sec:result_treewidth}, we give the main result of our implementation with hybrid barrier and Lee-Sidford barrier. 
We remark that if the LP has polylogarithmic treewidth, then we exhibit a solver that uses $\wt O(\sqrt{n}\log(mR^2\|c\|_2/(r\epsilon)))$ passes in semi-streaming model, in which only $\wt O(n)$ space is allowed.
\subsection{Preliminary}
\label{sec:pre_treewidth}

To motivate the discussion the small treewidth linear program, let us consider the regime where $m\ll n^2$. This means that an algorithm uses $\wt O(n^2)$ is too much and one wonders whether it is possible to obtain an $\wt O(m)$ space implementation. We show that given a linear program with treewidth $\tau$, we can use $\wt O(n\tau)$ space and $\wt O(\sqrt{n}\tau\log(mR^2\|c\|_2/(r\epsilon)))$ passes when the feasible region lies in a ball of radius $R$ and contains a ball of radius $r$ with known center.

\begin{figure*}[!ht]
\centering
\subfloat[]{\includegraphics[width=0.3\textwidth]{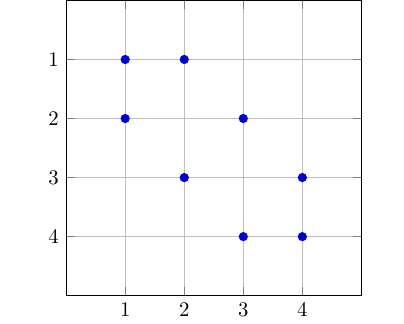}} 
\subfloat[]{\includegraphics[width=0.3\textwidth]{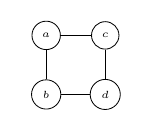}}
\subfloat[]{\includegraphics[width=0.3\textwidth]{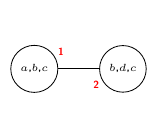}}
\caption{
Example of small treewidth matrix and it's corresponding graph and tree decomposition. (a) A $2$-sparse $4$ by $4$ data matrix $D$, where the blue dots represent the non-zero elements. 
(b) The corresponding graph $G$ of data matrix $D$. By Definition \ref{def:treewidth}, The corresponding graph of matrix $D$ will be a simple cycle of four nodes $a,b,c,d$. Here we assume node $a,b,c,d$ in $G$ represent the columns of $D$ sequentially. (c) The corresponding tree decomposition, which have two bags: $1:\{a,b,c\}$ and $2:\{b,c,d\}$. By Definition \ref{def:tree_decompose}, we know that the treewidth of matrix $D$ is $2$, which satisfies the sparsity of the columns in matrix $D$.
}
\label{fig:treewidth_example}
\end{figure*}
We start with the definition of the dual graph for a given matrix.
\begin{definition}
The generalized dual graph of the matrix $A \in \R^{m \times n}$ with block structure $m = \sum_{i=1}^k m_i$ is the graph $G_A = (V,E) $ with $V= \{1, \cdots, n\}$.

We say an edge $(i,j) \in E$ if and only if $A_{r,i} \neq 0$ and $A_{r,j} \neq 0$ for some $r$, where we use $A_{r,i}$ to mean the submatrix of $A$ in row block $r$ and column $i$.
\end{definition}

We decompose a graph into tree by using tree decomposition. A simple example is shown in Figure \ref{fig:treewidth_example}. The following is the definition of tree decomposition.
\begin{definition} \label{def:tree_decompose}
A tree decomposition is a mapping of graphs into trees. For graph $G$, the tree decomposition is defined as pair $(M,T)$, where $T$ is a tree, and $M : V(T) \rightarrow 2^{V(G)}$ is a family of subsets of $V(G)$ called bags labelling the vertices of $T$, satisfies that:
\begin{itemize}
    \item The vertices maintained by all bags is the same as those of graph $G$: $\cup_{t \in V(T)} M(t) = V(G)$.
    \item For every vertex $v \in V(G)$, the nodes $t \in V(T)$ satisfying $v \in M(t)$ is a connected subgraph of $T$, and
    \item For every edge $e = (u,v) \in E(G)$, there exist a node $t \in V(T)$ so that $u,v \in M(t)$.
\end{itemize}
where $V(\cdot)$ and $E(\cdot)$ denote the vertex set and the edge set of a graph. 

The width of a tree decomposition $(M,T)$ is $\max\{ |M(t)| - 1 : t \in T \}$. The treewidth of $G$ is the minimum width over all tree decompositions of $G$. 
\end{definition}

Next, we give the definition of the treewidth for a given matrix.
\begin{definition}[Treewidth $\tau$]\label{def:treewidth} 
Given a matrix $A \in \R^{m \times n}$, we construct its graph $G = (V,E)$ as follows: The vertex set are columns $[n]$; An edge $(i,j) \in E$ if and only if there exists $k \in [m]$ such that $A_{k,i} \ne 0, A_{k,j}\ne 0$. Then, the treewidth of the matrix $A$ is the treewidth of the constructed graph.
In particular, the support of every row of $A$ is a clique of $G$, so every row of $A$ has at most $\tau+1$ nonzero entries.\footnote{In this paper, we use $A$ or $A^\top$ has treewidth $\tau$ interchangeably.} 
\end{definition}
Then, we present the definition for Cholesky factorization.

\begin{definition}[Cholesky Factorization]\label{def:cholesky_dec}
Given a positive-definite matrix $P\in \R^{n\times n}$, there exists a unique Cholesky factorization $P = LL^\top \in \R^{n \times n}$, where $L \in \R^{n \times n}$ is a lower-triangular matrix with real and positive diagonal entries. 
\end{definition}

The following is a folklore lemma regarding the relationship between $n\tau$ and $m$.

\begin{lemma}
Let $G$ be a graph with treewidth $\tau$, then $m\leq n\tau$.
\end{lemma}

We slightly change the query model of the streaming model.

\begin{assumption}
We assume any entry of $A$ and any entry of $b$ can be queried. Once all entries of $A$ have been queried, we count as one pass.
\end{assumption}

\subsection{Small space implementation}\label{sec:space_treewidth}

We first introduce the time and space needed for the Cholesky factorization with treewidth $\tau$.

\begin{lemma}[\cite{bghk95,d06,dly21}]\label{lem:fast_cholesky}
For any positive diagonal matrix $H \in \R^{m \times m}$, for any matrix $A^\top \in \R^{n \times m}$ with treewidth $\tau$, we can compute a permutation matrix $P$ and the Cholesky factorization $PA^\top H AP^\top = L L^\top \in \R^{n \times n}$ in $\wt O(m \tau^2)+(m\tau)^{1+o(1)}$ time,  where $L \in \R^{n \times n}$ is a lower-triangular matrix with real and positive entries.
$L$ satisfies the property that every column is $\wt O(\tau)$-sparse.
\end{lemma}

We consider the following pseudocode for computing the Cholesky factorization, and analyze its space complexity. 

\begin{algorithm}[!ht]
\caption{Cholesky factorization}
\label{alg:cholesky}
\begin{algorithmic}[1]
\Procedure{Cholesky}{$M\in \R^{n\times n}$} \Comment{Lemma \ref{lem:fast_cholesky}}
\For{$j=1\to n$}
\State $L_{j,j}\gets \sqrt{M_{j,j}-\sum_{k=1}^{j-1} L_{j,k}^2}$
\For{$i=j+1\to n$}
\State $L_{i,j}\gets \frac{1}{L_{j,j}}\cdot (M_{i,j}-\sum_{k=1}^{j-1}L_{i,k}L_{j,k})$
\EndFor
\EndFor
\State \Return $L$
\EndProcedure
\end{algorithmic}
\end{algorithm}

\begin{lemma}
\label{lem:cholesky_space}
For any positive diagonal matrix $H\in \R^{m\times m}$ and $A^\top\in \R^{n\times m}$ with treewidth $\tau$, given the permutation matrix $P$ of Lemma~\ref{lem:elim_tree_space}, we can compute the Cholesky factorization $PA^\top HAP^\top=LL^\top\in \R^{n\times n}$ in $\wt O(n\tau)$ space and $\wt O(\tau)$ passes, where $L\in \R^{n\times n}$ is a lower triangular matrix with the nonzero part being positive. $L$ has column sparsity $\wt O(\tau)$.
\end{lemma}

\begin{proof}
We give a space-efficient implementation for two steps: 1).\ supplying the matrix $A^\top HA$ and 2).\ computing the Cholesky factors. 

First note that $L$ has $\wt O(n\tau)$ nonzero entries, therefore, it suffices to allocate $\wt O(n\tau)$ space to store $L$, which will be crucial for our later procedures.

In such space budget, we can also store all entries of $H$ in $m$ space.

To compute $(A^\top HA)_{i,j}$, we have $(A^\top HA)_{i,j}=\sum_{k=1}^m A_{i,k}H_{k,k}A_{k,j}$, this value can be computed as accumulating $m$ terms, and compute one term only needs $O(1)$ space, so computing one entry of $A^\top HA$ takes $O(1)$ space.

For each entry of $L$, it requires one entry of $A^\top HA$ and at most $n$ reads of prior-computed entries of $L$. Again, this value can be accumulated, so computing one entry of $L$ takes $O(1)$ space.

As $L$ only has $\wt O(n\tau)$ nonzero entries, and each entry takes $O(1)$ space, computing  and storing $L$ takes $\wt O(n\tau)$ space.

As computing $n$ entries of $L$ takes 1 pass over $A$, it takes $\wt O(\tau)$ passes since $L$ has $\wt O(n\tau)$ nonzero entries.
\end{proof}

To compute this column-sparse Cholesky factor, it is imperative to construct a corresponding elimination tree. This tree should be viewed as a permutation matrix $P\in \R^{n\times n}$ that permutes the rows of $A^\top HA$ to achieve good sparsity patterns. We recall their algorithm (Algorithm \ref{alg:perm}).

\begin{algorithm}
\caption{Computing the permutation matrix}
\label{alg:perm}
\begin{algorithmic}[1]
\Procedure{Permutation}{$G$}
\If{$|V(G)|\leq f(\tau)$}
\State let $\pi$ be a random ordering of $V(G)$
\State construct a path on $V(G)$ according to $\pi$, let $u$ be the last vertex of $\pi$
\State \Return $(\pi, u)$
\EndIf
\State $(G_1, S, G_2)\gets \textsc{ApproxBalancedSeparator}(G)$
\State $(\pi_1, v_1)\gets \textsc{Permutation}(G_1)$
\State $(\pi_2, v_2)\gets \textsc{Permutation}(G_2)$
\State $\pi\gets \text{random ordering of $S$}$
\State construct a path on $S$ according to $\pi$, let $u$ be the first vertex of $\pi$ and $v$ be the last vertex
\State set $u$ as the parent of $v_1$ and $v_2$
\State \Return $(\pi_1+\pi_2+\pi, v)$
\EndProcedure
\end{algorithmic}
\end{algorithm}

\begin{lemma}
\label{lem:elim_tree_space}
For any positive diagonal matrix $H\in \R^{m\times m}$ and $A^\top\in \R^{n\times m}$ with treewidth $\tau$, one can compute a permutation matrix $P\in \R^{n\times n}$ such that the Cholesky factor of $PA^\top HAP^\top=LL^\top$ has the property that $L$ has column sparsity $\wt O(\tau)$. Moreover, matrix $P$ can be computed in $\wt O(n\tau)$ space and $\wt O(\tau)$ passes.
\end{lemma}

\begin{proof}
First note that the permutation matrix can be specified by the ordering $\pi$, which can be stored in $O(n)$ space. So it suffices to argue for the space consumption of Algorithm~\ref{alg:perm}. We start with the space usage \emph{without} \textsc{ApproxBalancedSeparator}.

Note that all the algorithm does is to compute an ordering of vertices on a partition of vertices. For each recursive call, it is enough to store the ordering on the subset $S$. Thus, the algorithm uses $O(n)$ space.

We need to analyze the space usage of \textsc{ApproxBalancedSeparator}. We do not bound its running time, since our results in this section only bound space and passes. In one pass over $A$ we store the graph of $A$ from Definition~\ref{def:treewidth}, which we denote by $G_A$. When a row of $A$ arrives, we insert an edge between every pair of its nonzero columns, unless this edge is already stored. A graph on $n$ vertices with treewidth $\tau$ has at most $n\tau$ edges (Lemma 4.4 of~\cite{dly21}), so this takes $O(n\tau)$ space, and no further pass over $A$ is needed.

Let $G$ be the graph of a recursive call and let $k=|V(G)|$. It is the subgraph of $G_A$ induced on a subset of the vertices, so its treewidth is at most $\tau$ (Lemma 4.5 of~\cite{dly21}). Hence there is a set $S\subseteq V(G)$ of at most $\tau+1$ vertices such that every connected component of $G\setminus S$ has at most $k/2$ vertices (Theorem 4.10 of~\cite{dly21}, attributed there to~\cite{bghk95}). We find such a set by going through the subsets of $V(G)$ of size at most $\tau+1$ one at a time. For each candidate $S$, we compute the connected components of $G\setminus S$ and their sizes, which takes $O(n)$ space in addition to the stored graph. Given $S$, we split the components into two groups as follows. If some component has at least $k/3$ vertices, it forms the first group alone. Otherwise we add components to the first group one by one until it has at least $k/3$ vertices or no component is left. All other components form the second group. The first group has at most $2k/3$ vertices. In the first case this holds because a component has at most $k/2$ vertices. In the second case the group had fewer than $k/3$ vertices before its last component was added, and this component has fewer than $k/3$ vertices. The second group is empty or has at most $k-k/3=2k/3$ vertices. We let $G_1$ and $G_2$ be the subgraphs of $G$ induced on the two groups and return $(G_1,S,G_2)$.

This meets the requirements of Theorem 4.16 of~\cite{dly21} with $f(\tau)=\tau+1$. So Algorithm~\ref{alg:perm} returns an ordering and a tree of height $O(\tau\log n)$. As shown in the proof of Theorem 4.1 of~\cite{dly21}, the nonzero entries of each column of $L$ then lie on the path from the corresponding vertex to the root of this tree, so $L$ has column sparsity $O(\tau\log n)=\wt O(\tau)$. The search stores only the graph $G_A$, the current candidate $S$ and the labels of the components. Thus, the overall space usage is $O(n\tau)$.

Regarding the number of passes, storing $G_A$ takes one pass, and the remaining operations take no pass.
\end{proof}

The search for the separators in this proof takes time exponential in $\tau$. The algorithm of~\cite{bgs21} computes a tree decomposition of width $O(\tau\log^3 n)$ in $(n\tau)^{1+o(1)}$ time (Theorem 4.18 of~\cite{dly21}), and balanced separators can be read off from it (Theorem 4.17 of~\cite{dly21}). This gives a polynomial time implementation, but we can only bound its space by its running time, that is, by $(n\tau)^{1+o(1)}$ instead of $\wt O(n\tau)$.

Now that we can compute $L$ in $\wt O(n\tau)$, we show how to implement some fundamental queries with $L$, such as solving the linear system $Lx=v$.

Given a lower triangular $L\in \R^{n\times n}$ with column sparsity $\tau$, we can solve for $x$ from top to bottom. For each coordinate, we need to compute $x_j\leftarrow v_j/L_{j,j}$ in unit space, then compute $v-x_j L_{*,j}$. Each iteration can be implemented in place, in $O(n)$ space. This gives the following lemma.

\begin{lemma}[\cite{dly21}]\label{lem:fast_cholesky_compute}
For any positive diagonal matrix $H \in \R^{m \times m}$, given matrix $A^\top \in \R^{n \times m}$ with treewidth $\tau$ and the Cholesky factorization $A^\top H A = LL^\top \in \R^{n \times n}$.
\begin{itemize}
    \item for any vector $v \in \R^n$, we can compute $L^{-1} v$ in $O(n \tau)$ time and $O(n)$ space.
    \item for any vector $v \in \R^n$, we can compute $L^{-\top}v$ in $O(n \tau)$ time and $O(n)$ space. 
\end{itemize}
\end{lemma}

To implement a pass-efficient LP solver, we use the hybrid barrier-based IPM. Recall the approximate Hessian $\wt H$ of the hybrid barrier is in the form of $A^\top S^{-2}(5\Sigma+(n/m)I) A$, so we first show how to implement the matrix-vector product in the form of $(A^\top HA)^{-1}v$ for positive diagonal matrix $H$.

\begin{lemma}
\label{lem:hessian_inv_space}
Let $A^\top\in \R^{n\times m}$ be with treewidth $\tau$ and $H\in \R^{m\times m}$ be positive diagonal matrix. Let $v\in \R^n$. Then, $(A^\top HA)^{-1}v$ can be computed in $\wt O(n\tau)$ space and $\wt O(\tau)$ passes.
\end{lemma}

\begin{proof}
First, observe that we can compute the Cholesky factors $PA^\top HAP^\top=LL^\top$ in $\wt O(n\tau)$ space owing to Lemma~\ref{lem:elim_tree_space} and Lemma~\ref{lem:cholesky_space}. Given these factors, it is then straightforward that $(A^\top HA)^{-1}=P^\top L^{-\top} L^{-1}P$ and the matrix-vector query can be implemented as computing $L^{-1}Pv$ then $P^\top L^{-\top} (L^{-1}Pv)$. Owing to Lemma~\ref{lem:fast_cholesky_compute}, both of these solves can be implemented in $O(n)$ space. Thus, the total space usage is $\wt O(n\tau)$. The number of passes follows from Lemma~\ref{lem:cholesky_space} and~\ref{lem:elim_tree_space}.
\end{proof}

For hybrid barrier, we need to compute the leverage score of matrix $\sqrt{H}A\in \R^{m\times n}$. We show that leverage score can be computed in small space.

\begin{lemma}
\label{lem:leverage_space}
Let $A^\top\in \R^{n\times m}$ be with treewidth $\tau$ and $H\in \R^{m\times m}$ be positive diagonal matrix. Let $B=\sqrt{H}A$. The $i$-th leverage score is defined as $\sigma_i:=b_i^\top (B^\top B)^{-1} b_i$ . For any $i\in [m]$, $\sigma_i$ can be computed in $\wt O(n\tau)$ space and $\wt O(\tau)$ passes.
\end{lemma}
\begin{proof}

Let $b_i$ denote the $i$-th row of $B$. Note $(B^\top B)^{-1}=(A^\top HA)^{-1}=P^\top L^{-\top}L^{-1}P$, then $\sigma_i$ can be formulated as $(L^{-1}Pb_i)^\top L^{-1}Pb_i$. 

The proof is then similar to Lemma~\ref{lem:hessian_inv_space} and requires $\wt O(n\tau)$ space and $\wt O(\tau)$ passes.
\end{proof}

\subsection{Main result}
\label{sec:result_treewidth}

Next, we give the main results of our implementations of IPM with hybrid barrier and Lee-Sidford barrier given treewidth $\tau$.
\begin{theorem}[Hybrid barrier]
Under Assumption~\ref{ass:input_of_stream}, given any feasible linear program
\begin{align*}
\min_{x\in \R^n, Ax\geq b} c^{\top} x,
\end{align*}
where $A\in \R^{m\times n}$, $b\in \R^{m}$, and $c\in \R^n$. Suppose $A$ has treewidth $\tau$. Suppose the solution exists and let $x^*\in \R^n$ be the solution. For any $\epsilon >0$, we can outputs an $x$ which is a \emph{nearly-optimal solution}
\begin{align*}
    c^{\top} x - c^{\top} x^* \leq \epsilon.
\end{align*}
in $\wt O(n \tau)$ space and $\wt O((nm)^{1/4}\tau\log(mR^2\|c\|_2/(r\epsilon)))$ passes, provided the feasible region $\{x: Ax\geq b\}$ is contained in a ball of radius $R$ and contains a ball of radius $r$ with known center. In particular, the number of passes is $\wt O(\sqrt{n}\tau^{5/4}\log(mR^2\|c\|_2/(r\epsilon)))$ when $m=O(n\tau)$ (e.g., when $A$ is the incidence matrix of a simple graph with box constraints on single variables added).
\end{theorem}
\begin{proof}

The proof is similar to that of Theorem~\ref{thm:nsquare_hybrid}, except we need to give a new space and pass bound for $\wt \delta_x=\wt H(x)^{-1}\nabla f_t(x)$. 

We will first construct the Cholesky factor of $A^\top S(x)^{-2}A$ in $\wt O(n\tau)$ space and $\wt O(\tau)$ passes.

For $\nabla f_t(x)=tc-k\sum_{i=1}^m(\sigma_i(x)+n/m)\frac{a_i}{s_i(x)}$, with $k$ as in the proof of Theorem~\ref{thm:nsquare_hybrid}, per Lemma~\ref{lem:leverage_space}, each $\sigma_i(x)$ can be computed in $\wt O(n\tau)$ space, and we can always reuse the $\wt O(n\tau)$ space for each coordinate. Thus, it can be computed in $\wt O(n\tau)$ space. We do not store the $m$ leverage scores. Instead, each $\sigma_i(x)$ is computed from the stored Cholesky factor when the $i$-th row is read, and its term is added to the gradient right away. So we only need one pass, as the Cholesky factor has been stored.

For $\wt H(x)$, we show that it can be formulated as $A^\top HA$ for some positive diagonal matrix $H$ (up to the factor $k$ of the proof of Theorem~\ref{thm:nsquare_hybrid}, which does not change this form). Recall that 
\begin{align*} 
\wt H(x)=5 \cdot Q(x)+(n/m) \cdot \nabla^2 \phi(x),
\end{align*}
and \begin{align*}
Q(x)=A^\top \Sigma(x)S^{-2}(x) A
\end{align*}
where $\Sigma(x)$ denote the diagonal matrix for leverage score, and
\begin{align*} 
\nabla^2 \phi(x)=A^\top S^{-2}(x)A,
\end{align*}
so 
\begin{align*} 
\wt{H}(x)=A^\top (5\Sigma(x)+(n/m)I)S^{-2}(x) A.
\end{align*}

To compute the Cholesky factor for $\wt H(x)$, we note that Lemma~\ref{lem:cholesky_space} only queries entries of the diagonal matrix on-demand and never need to store them, hence, for each entry query, we can compute the leverage score in $\wt O(n\tau)$ space. The total space consumption of this step is $\wt O(n\tau)$. As we only read through $A$ once, it takes one pass.

To compute the permutation matrix $P$ via Lemma~\ref{lem:elim_tree_space}, notice that it does not depend on the diagonal matrix, so it takes $\wt O(n\tau)$ space and $\wt O(\tau)$ passes.

Finally, to compute $\wt \delta_x$, we invoke Lemma~\ref{lem:hessian_inv_space} and it requires $\wt O(n\tau)$ space. This completes the proof of space.

Regarding the number of passes, hybrid barrier requires $\wt O((nm)^{1/4}\log(mR^2\|c\|_2/(r\epsilon)))$ iterations, as in the proof of Theorem~\ref{thm:nsquare_hybrid}, and we need extra $\wt O(\tau)$ passes to compute the two Cholesky factors per iteration. Thus, the total number of passes is
\begin{align*}
    & ~ \wt O((nm)^{1/4}\tau \log(mR^2\|c\|_2/(r\epsilon)))
\end{align*}
which is $\wt O(\sqrt{n}\tau^{5/4} \log(mR^2\|c\|_2/(r\epsilon)))$ when $m=O(n\tau)$.
\end{proof} 

Take a step further, we show how to compute Lee-Sidford barrier in $\wt O(n\tau)$ space and therefore improve the passes for small treewidth LP.

\begin{theorem}[Lee-Sidford barrier, formal version of Theorem \ref{thm:lsbarrier_treewidth_informal}] \label{thm:lsbarrier_treewidth}
Under Assumption~\ref{ass:input_of_stream}, given any feasible linear program
\begin{align*}
\min_{x\in \R^n, Ax\geq b} c^{\top} x,
\end{align*}
where $A\in \R^{m\times n}$, $b\in \R^{m}$, and $c\in \R^n$. Suppose $A$ has treewidth $\tau$. 
Suppose the solution exists and let $x^*\in \R^n$ be the solution. For any $\epsilon >0$, we can outputs an $x$ which is a \emph{nearly-optimal solution}
\begin{align*}
    c^{\top} x - c^{\top} x^* \leq \epsilon.
\end{align*}
in $\wt O(n \tau)$ space and $\wt O(\sqrt{n}\tau\log(mR^2\|c\|_2/(r\epsilon)))$ passes, provided the feasible region $\{x: Ax\geq b\}$ is contained in a ball of radius $R$ and contains a ball of radius $r$ with known center.
\end{theorem}

\begin{proof}
We follow the proof strategy of Theorem~\ref{thm:nsquare_lsbarrier} and use its notation. So $v^{(t)}$ is the weight vector of Algorithm~\ref{alg:compute_lewis_weight} after its first $t$ calls to \textsc{Descent}, the bit $z_t$ records whether the $t$-th call is inside \textsc{Round} ($z_t=1$) or on all rows ($z_t=0$), and $J=\poly(\log m)$ is the total number of calls. The algorithm starts with the initial weight $v^{(0)}=\frac{n}{m}{\bf 1}$. Each call needs the leverage scores of the single matrix $(V^{(t)})^{1/2}A$ for the current weights, and then performs coordinate-wise updates. The key is to store necessary information to recover the leverage score.

Define $K_t=A^\top V^{(t)}A$, note that
\begin{align*}
\sigma_i(v^{(t)})=v^{(t)}_i\cdot a_i^\top K_t^{-1} a_i,
\end{align*}
 so it suffices to show how to implement the quadratic form in $\wt O(n\tau)$ space. In place of the inverse $K_t^{-1}$, we store the Cholesky factor $L_t$ of $PK_tP^\top$, where $P$ is the permutation matrix of Lemma~\ref{lem:elim_tree_space}. The matrix $P$ does not depend on the weights, so we compute it once. We prove the claim via induction on the number of calls.

When $t=0$, $v^{(0)}_i=\frac{n}{m}$ needs no computation, and we compute $L_0$ in $\wt O(n\tau)$ space and $\wt O(\tau)$ passes by Lemma~\ref{lem:cholesky_space}.

Assume that we have stored the Cholesky factors $L_0,\ldots, L_{t}$ and the bits $z_1,\ldots,z_t$. We show how to compute $z_{t+1}$ and the Cholesky factor $L_{t+1}$ in $\wt O(n\tau)$ space.

Given the row $a_i$, we first compute $v^{(t)}_i$ as in subroutine \textsc{ComputeW} (Algorithm~\ref{alg:lewis_weight_small_space_subroutine}). We start from $v^{(0)}_i=\frac{n}{m}$. For $s=1,\ldots,t$, we compute $a_i^\top K_{s-1}^{-1}a_i=\|L_{s-1}^{-1}Pa_i\|_2^2$ in $O(n)$ space by Lemma~\ref{lem:fast_cholesky_compute} and thus compute $\rho_i(v^{(s-1)})$. If $z_s=0$ or $\rho_i(v^{(s-1)})>1+\alpha$, we apply the update of \textsc{Descent} (Algorithm~\ref{alg:descent}) to $v^{(s-1)}_i$ to get $v^{(s)}_i$. Otherwise $v^{(s)}_i=v^{(s-1)}_i$. In the same way we get $\rho_i(v^{(t)})$ from $L_t$. This uses only $a_i$ and the stored factors and bits, so we never store a weight vector in $\R^m$. As in the proof of Theorem~\ref{thm:nsquare_lsbarrier}, one pass over $A$ then decides $z_{t+1}$, or finds that there is no further call. After that, we can supply the $i$-th diagonal entry of $V^{(t+1)}$ whenever the row $a_i$ is read. So we can compute $L_{t+1}$ in $\wt O(n\tau)$ space and $\wt O(\tau)$ passes owing to Lemma~\ref{lem:cholesky_space}.

We store $J+1$ Cholesky factors, each with $\wt O(n\tau)$ nonzero entries, and $J$ bits, so the space is $\wt O(n\tau)$. Each call takes $\wt O(\tau)$ passes, so computing the Lewis weights takes $\wt O(J\cdot \tau)=\wt O(\tau)$ passes in each iteration of the IPM. As in the proof of Theorem~\ref{thm:nsquare_lsbarrier}, the hypothesis on the feasible region gives the starting point $x_{\start},t_{\start}$ through the two walks along the central path in the proof of Theorem~\ref{thm:nsquare_logarithmic}. Their steps are computed in the same way as the IPM steps above. The known center $x_0$ and the cost vector of the first walk take $O(n)$ space. Counting both walks, the IPM needs $\wt O(\sqrt{n}\log(mR^2\|c\|_2/(r\epsilon)))$ iterations, so the number of passes is $\wt O(\sqrt{n}\tau\log(mR^2\|c\|_2/(r\epsilon)))$ in total.
\end{proof}

\begin{remark}
Let us instantiate the Lee-Sidford barrier result in terms of graph. Let $A\in \R^{m\times n}$ be the signed vertex-edge incidence matrix of a graph, note that our definition of treewidth captures the treewidth of the corresponding graph $G$. 

If the treewidth of $G$ is only polylogarithmic in terms of $n$ and $m$, then any graph problem that can be solved with our dual-only LP, provided its feasible region lies in a ball of radius $R$ and contains a ball of radius $r$ with known center, will require $\wt O(n)$ space and $\wt O(\sqrt{n}\log(mR^2\|c\|_2/(r\epsilon)))$ passes. Box constraints on single variables can be added to make the feasible region bounded. They do not change the treewidth, since each such row of $A$ has one nonzero entry. For the minimum vertex cover problem and exact maximum weight bipartite matching problem, this suggests $\wt O(\sqrt{n})$-pass algorithms in the semi-streaming model for graphs with polylogarithmic treewidth. Such algorithms would close the gap between these two problems and the problems of reachability and single source shortest path on such graphs. We do not work out the details, since the rounding steps of Sections~\ref{sec:minimum_vertex_cover} and~\ref{sec:combine} are stated for the logarithmic barrier.
\end{remark}
\else

\fi

\end{document}